\documentclass[a4paper,UKenglish,cleveref, autoref, thm-restate]{lipics-v2021}
\nolinenumbers

\pdfoutput=1 %
\hideLIPIcs  %

\usepackage{environ}
\usepackage{xspace}

\usepackage[ruled]{algorithm}
\usepackage[noend]{algpseudocode}
\usepackage{amscd}
\usepackage{amsfonts}
\usepackage{amsmath}
\usepackage{amssymb}
\usepackage{booktabs}
\usepackage{color}
\usepackage{datetime}
\usepackage{epsfig}
\usepackage{graphicx}
\usepackage{latexsym}
\usepackage{marginnote}
\usepackage{multirow}
\usepackage{pifont}
\usepackage{placeins}
\usepackage{rotate}
\usepackage[nolabel]{showlabels}
\usepackage[disable]{todonotes}
\usepackage{tikz}
\usepackage{url}
\usepackage{varwidth}
\usepackage{verbatim} 
\usepackage{xcolor,colortbl}
\usepackage{xspace}
 \usetikzlibrary{calc}

\usepackage{lipsum}
\usepackage{setspace}

\creflabelformat{enumi}{\emph{(#2#1#3)}}
\creflabelformat{enumii}{\emph{(#2#1#3)}}
\creflabelformat{enumiii}{\emph{(#2#1#3)}}
\creflabelformat{enumiiii}{\emph{(#2#1#3)}}

\definecolor {snow}                {rgb}{1.00,0.98,0.98}
\definecolor {ghostwhite}          {rgb}{0.97,0.97,1.00}
\definecolor {whitesmoke}          {rgb}{0.96,0.96,0.96}
\definecolor {gainsboro}           {rgb}{0.86,0.86,0.86}
\definecolor {floralwhite}         {rgb}{1.00,0.98,0.94}
\definecolor {oldlace}             {rgb}{0.99,0.96,0.90}
\definecolor {linen}               {rgb}{0.98,0.94,0.90}
\definecolor {antiquewhite}        {rgb}{0.98,0.92,0.84}
\definecolor {papayawhip}          {rgb}{1.00,0.94,0.84}
\definecolor {blanchedalmond}      {rgb}{1.00,0.92,0.80}
\definecolor {bisque}              {rgb}{1.00,0.89,0.77}
\definecolor {peachpuff}           {rgb}{1.00,0.85,0.73}
\definecolor {navajowhite}         {rgb}{1.00,0.87,0.68}
\definecolor {moccasin}            {rgb}{1.00,0.89,0.71}
\definecolor {cornsilk}            {rgb}{1.00,0.97,0.86}
\definecolor {ivory}               {rgb}{1.00,1.00,0.94}
\definecolor {lemonchiffon}        {rgb}{1.00,0.98,0.80}
\definecolor {seashell}            {rgb}{1.00,0.96,0.93}
\definecolor {honeydew}            {rgb}{0.94,1.00,0.94}
\definecolor {mintcream}           {rgb}{0.96,1.00,0.98}
\definecolor {azure}               {rgb}{0.94,1.00,1.00}
\definecolor {aliceblue}           {rgb}{0.94,0.97,1.00}
\definecolor {lavender}            {rgb}{0.90,0.90,0.98}
\definecolor {lavenderblush}       {rgb}{1.00,0.94,0.96}
\definecolor {mistyrose}           {rgb}{1.00,0.89,0.88}
\definecolor {white}               {rgb}{1.00,1.00,1.00}
\definecolor {black}               {rgb}{0.00,0.00,0.00}
\definecolor {darkslategray}       {rgb}{0.18,0.31,0.31}
\definecolor {dimgray}             {rgb}{0.41,0.41,0.41}
\definecolor {slategray}           {rgb}{0.44,0.50,0.56}
\definecolor {lightslategray}      {rgb}{0.47,0.53,0.60}
\definecolor {gray}                {rgb}{0.75,0.75,0.75}
\definecolor {lightgrey}           {rgb}{0.83,0.83,0.83}
\definecolor {midnightblue}        {rgb}{0.10,0.10,0.44}
\definecolor {navy}                {rgb}{0.00,0.00,0.50}
\definecolor {cornflowerblue}      {rgb}{0.39,0.58,0.93}
\definecolor {darkslateblue}       {rgb}{0.28,0.24,0.55}
\definecolor {slateblue}           {rgb}{0.42,0.35,0.80}
\definecolor {mediumslateblue}     {rgb}{0.48,0.41,0.93}
\definecolor {lightslateblue}      {rgb}{0.52,0.44,1.00}
\definecolor {mediumblue}          {rgb}{0.00,0.00,0.80}
\definecolor {royalblue}           {rgb}{0.25,0.41,0.88}
\definecolor {blue}                {rgb}{0.00,0.00,1.00}
\definecolor {dodgerblue}          {rgb}{0.12,0.56,1.00}
\definecolor {deepskyblue}         {rgb}{0.00,0.75,1.00}
\definecolor {skyblue}             {rgb}{0.53,0.81,0.92}
\definecolor {lightskyblue}        {rgb}{0.53,0.81,0.98}
\definecolor {steelblue}           {rgb}{0.27,0.51,0.71}
\definecolor {lightsteelblue}      {rgb}{0.69,0.77,0.87}
\definecolor {lightblue}           {rgb}{0.68,0.85,0.90}
\definecolor {powderblue}          {rgb}{0.69,0.88,0.90}
\definecolor {paleturquoise}       {rgb}{0.69,0.93,0.93}
\definecolor {darkturquoise}       {rgb}{0.00,0.81,0.82}
\definecolor {mediumturquoise}     {rgb}{0.28,0.82,0.80}
\definecolor {turquoise}           {rgb}{0.25,0.88,0.82}
\definecolor {cyan}                {rgb}{0.00,1.00,1.00}
\definecolor {lightcyan}           {rgb}{0.88,1.00,1.00}
\definecolor {cadetblue}           {rgb}{0.37,0.62,0.63}
\definecolor {mediumaquamarine}    {rgb}{0.40,0.80,0.67}
\definecolor {aquamarine}          {rgb}{0.50,1.00,0.83}
\definecolor {darkgreen}           {rgb}{0.00,0.39,0.00}
\definecolor {darkolivegreen}      {rgb}{0.33,0.42,0.18}
\definecolor {darkseagreen}        {rgb}{0.56,0.74,0.56}
\definecolor {seagreen}            {rgb}{0.18,0.55,0.34}
\definecolor {mediumseagreen}      {rgb}{0.24,0.70,0.44}
\definecolor {lightseagreen}       {rgb}{0.13,0.70,0.67}
\definecolor {palegreen}           {rgb}{0.60,0.98,0.60}
\definecolor {springgreen}         {rgb}{0.00,1.00,0.50}
\definecolor {lawngreen}           {rgb}{0.49,0.99,0.00}
\definecolor {green}               {rgb}{0.00,1.00,0.00}
\definecolor {chartreuse}          {rgb}{0.50,1.00,0.00}
\definecolor {mediumspringgreen}   {rgb}{0.00,0.98,0.60}
\definecolor {greenyellow}         {rgb}{0.68,1.00,0.18}
\definecolor {limegreen}           {rgb}{0.20,0.80,0.20}
\definecolor {yellowgreen}         {rgb}{0.60,0.80,0.20}
\definecolor {forestgreen}         {rgb}{0.13,0.55,0.13}
\definecolor {olivedrab}           {rgb}{0.42,0.56,0.14}
\definecolor {darkkhaki}           {rgb}{0.74,0.72,0.42}
\definecolor {khaki}               {rgb}{0.94,0.90,0.55}
\definecolor {palegoldenrod}       {rgb}{0.93,0.91,0.67}
\definecolor {lightgoldenrodyellow} {rgb}{0.98,0.98,0.82}
\definecolor {lightyellow}         {rgb}{1.00,1.00,0.88}
\definecolor {yellow}              {rgb}{1.00,1.00,0.00}
\definecolor {gold}                {rgb}{1.00,0.84,0.00}
\definecolor {lightgoldenrod}      {rgb}{0.93,0.87,0.51}
\definecolor {goldenrod}           {rgb}{0.85,0.65,0.13}
\definecolor {darkgoldenrod}       {rgb}{0.72,0.53,0.04}
\definecolor {rosybrown}           {rgb}{0.74,0.56,0.56}
\definecolor {indianred}           {rgb}{0.80,0.36,0.36}
\definecolor {saddlebrown}         {rgb}{0.55,0.27,0.07}
\definecolor {sienna}              {rgb}{0.63,0.32,0.18}
\definecolor {peru}                {rgb}{0.80,0.52,0.25}
\definecolor {burlywood}           {rgb}{0.87,0.72,0.53}
\definecolor {beige}               {rgb}{0.96,0.96,0.86}
\definecolor {wheat}               {rgb}{0.96,0.87,0.70}
\definecolor {sandybrown}          {rgb}{0.96,0.64,0.38}
\definecolor {tan}                 {rgb}{0.82,0.71,0.55}
\definecolor {chocolate}           {rgb}{0.82,0.41,0.12}
\definecolor {firebrick}           {rgb}{0.70,0.13,0.13}
\definecolor {brown}               {rgb}{0.65,0.16,0.16}
\definecolor {darksalmon}          {rgb}{0.91,0.59,0.48}
\definecolor {salmon}              {rgb}{0.98,0.50,0.45}
\definecolor {lightsalmon}         {rgb}{1.00,0.63,0.48}
\definecolor {orange}              {rgb}{1.00,0.65,0.00}
\definecolor {darkorange}          {rgb}{1.00,0.55,0.00}
\definecolor {coral}               {rgb}{1.00,0.50,0.31}
\definecolor {lightcoral}          {rgb}{0.94,0.50,0.50}
\definecolor {tomato}              {rgb}{1.00,0.39,0.28}
\definecolor {orangered}           {rgb}{1.00,0.27,0.00}
\definecolor {red}                 {rgb}{1.00,0.00,0.00}
\definecolor {hotpink}             {rgb}{1.00,0.41,0.71}
\definecolor {deeppink}            {rgb}{1.00,0.08,0.58}
\definecolor {pink}                {rgb}{1.00,0.75,0.80}
\definecolor {lightpink}           {rgb}{1.00,0.71,0.76}
\definecolor {palevioletred}       {rgb}{0.86,0.44,0.58}
\definecolor {maroon}              {rgb}{0.69,0.19,0.38}
\definecolor {mediumvioletred}     {rgb}{0.78,0.08,0.52}
\definecolor {violetred}           {rgb}{0.82,0.13,0.56}
\definecolor {magenta}             {rgb}{1.00,0.00,1.00}
\definecolor {violet}              {rgb}{0.93,0.51,0.93}
\definecolor {plum}                {rgb}{0.87,0.63,0.87}
\definecolor {orchid}              {rgb}{0.85,0.44,0.84}
\definecolor {mediumorchid}        {rgb}{0.73,0.33,0.83}
\definecolor {darkorchid}          {rgb}{0.60,0.20,0.80}
\definecolor {darkviolet}          {rgb}{0.58,0.00,0.83}
\definecolor {blueviolet}          {rgb}{0.54,0.17,0.89}
\definecolor {purple}              {rgb}{0.63,0.13,0.94}
\definecolor {mediumpurple}        {rgb}{0.58,0.44,0.86}
\definecolor {thistle}             {rgb}{0.85,0.75,0.85}
\definecolor {snow2}               {rgb}{0.93,0.91,0.91}
\definecolor {snow3}               {rgb}{0.80,0.79,0.79}
\definecolor {snow4}               {rgb}{0.55,0.54,0.54}
\definecolor {seashell2}           {rgb}{0.93,0.90,0.87}
\definecolor {seashell3}           {rgb}{0.80,0.77,0.75}
\definecolor {seashell4}           {rgb}{0.55,0.53,0.51}
\definecolor {antiquewhite1}       {rgb}{1.00,0.94,0.86}
\definecolor {antiquewhite2}       {rgb}{0.93,0.87,0.80}
\definecolor {antiquewhite3}       {rgb}{0.80,0.75,0.69}
\definecolor {antiquewhite4}       {rgb}{0.55,0.51,0.47}
\definecolor {bisque2}             {rgb}{0.93,0.84,0.72}
\definecolor {bisque3}             {rgb}{0.80,0.72,0.62}
\definecolor {bisque4}             {rgb}{0.55,0.49,0.42}
\definecolor {peachpuff2}          {rgb}{0.93,0.80,0.68}
\definecolor {peachpuff3}          {rgb}{0.80,0.69,0.58}
\definecolor {peachpuff4}          {rgb}{0.55,0.47,0.40}
\definecolor {navajowhite2}        {rgb}{0.93,0.81,0.63}
\definecolor {navajowhite3}        {rgb}{0.80,0.70,0.55}
\definecolor {navajowhite4}        {rgb}{0.55,0.47,0.37}
\definecolor {lemonchiffon2}       {rgb}{0.93,0.91,0.75}
\definecolor {lemonchiffon3}       {rgb}{0.80,0.79,0.65}
\definecolor {lemonchiffon4}       {rgb}{0.55,0.54,0.44}
\definecolor {cornsilk2}           {rgb}{0.93,0.91,0.80}
\definecolor {cornsilk3}           {rgb}{0.80,0.78,0.69}
\definecolor {cornsilk4}           {rgb}{0.55,0.53,0.47}
\definecolor {ivory2}              {rgb}{0.93,0.93,0.88}
\definecolor {ivory3}              {rgb}{0.80,0.80,0.76}
\definecolor {ivory4}              {rgb}{0.55,0.55,0.51}
\definecolor {honeydew2}           {rgb}{0.88,0.93,0.88}
\definecolor {honeydew3}           {rgb}{0.76,0.80,0.76}
\definecolor {honeydew4}           {rgb}{0.51,0.55,0.51}
\definecolor {lavenderblush2}      {rgb}{0.93,0.88,0.90}
\definecolor {lavenderblush3}      {rgb}{0.80,0.76,0.77}
\definecolor {lavenderblush4}      {rgb}{0.55,0.51,0.53}
\definecolor {mistyrose2}          {rgb}{0.93,0.84,0.82}
\definecolor {mistyrose3}          {rgb}{0.80,0.72,0.71}
\definecolor {mistyrose4}          {rgb}{0.55,0.49,0.48}
\definecolor {azure2}              {rgb}{0.88,0.93,0.93}
\definecolor {azure3}              {rgb}{0.76,0.80,0.80}
\definecolor {azure4}              {rgb}{0.51,0.55,0.55}
\definecolor {slateblue1}          {rgb}{0.51,0.44,1.00}
\definecolor {slateblue2}          {rgb}{0.48,0.40,0.93}
\definecolor {slateblue3}          {rgb}{0.41,0.35,0.80}
\definecolor {slateblue4}          {rgb}{0.28,0.24,0.55}
\definecolor {royalblue1}          {rgb}{0.28,0.46,1.00}
\definecolor {royalblue2}          {rgb}{0.26,0.43,0.93}
\definecolor {royalblue3}          {rgb}{0.23,0.37,0.80}
\definecolor {royalblue4}          {rgb}{0.15,0.25,0.55}
\definecolor {blue2}               {rgb}{0.00,0.00,0.93}
\definecolor {blue4}               {rgb}{0.00,0.00,0.55}
\definecolor {dodgerblue2}         {rgb}{0.11,0.53,0.93}
\definecolor {dodgerblue3}         {rgb}{0.09,0.45,0.80}
\definecolor {dodgerblue4}         {rgb}{0.06,0.31,0.55}
\definecolor {steelblue1}          {rgb}{0.39,0.72,1.00}
\definecolor {steelblue2}          {rgb}{0.36,0.67,0.93}
\definecolor {steelblue3}          {rgb}{0.31,0.58,0.80}
\definecolor {steelblue4}          {rgb}{0.21,0.39,0.55}
\definecolor {deepskyblue2}        {rgb}{0.00,0.70,0.93}
\definecolor {deepskyblue3}        {rgb}{0.00,0.60,0.80}
\definecolor {deepskyblue4}        {rgb}{0.00,0.41,0.55}
\definecolor {skyblue1}            {rgb}{0.53,0.81,1.00}
\definecolor {skyblue2}            {rgb}{0.49,0.75,0.93}
\definecolor {skyblue3}            {rgb}{0.42,0.65,0.80}
\definecolor {skyblue4}            {rgb}{0.29,0.44,0.55}
\definecolor {lightskyblue1}       {rgb}{0.69,0.89,1.00}
\definecolor {lightskyblue2}       {rgb}{0.64,0.83,0.93}
\definecolor {lightskyblue3}       {rgb}{0.55,0.71,0.80}
\definecolor {lightskyblue4}       {rgb}{0.38,0.48,0.55}
\definecolor {slategray1}          {rgb}{0.78,0.89,1.00}
\definecolor {slategray2}          {rgb}{0.73,0.83,0.93}
\definecolor {slategray3}          {rgb}{0.62,0.71,0.80}
\definecolor {slategray4}          {rgb}{0.42,0.48,0.55}
\definecolor {lightsteelblue1}     {rgb}{0.79,0.88,1.00}
\definecolor {lightsteelblue2}     {rgb}{0.74,0.82,0.93}
\definecolor {lightsteelblue3}     {rgb}{0.64,0.71,0.80}
\definecolor {lightsteelblue4}     {rgb}{0.43,0.48,0.55}
\definecolor {lightblue1}          {rgb}{0.75,0.94,1.00}
\definecolor {lightblue2}          {rgb}{0.70,0.87,0.93}
\definecolor {lightblue3}          {rgb}{0.60,0.75,0.80}
\definecolor {lightblue4}          {rgb}{0.41,0.51,0.55}
\definecolor {lightcyan2}          {rgb}{0.82,0.93,0.93}
\definecolor {lightcyan3}          {rgb}{0.71,0.80,0.80}
\definecolor {lightcyan4}          {rgb}{0.48,0.55,0.55}
\definecolor {paleturquoise1}      {rgb}{0.73,1.00,1.00}
\definecolor {paleturquoise2}      {rgb}{0.68,0.93,0.93}
\definecolor {paleturquoise3}      {rgb}{0.59,0.80,0.80}
\definecolor {paleturquoise4}      {rgb}{0.40,0.55,0.55}
\definecolor {cadetblue1}          {rgb}{0.60,0.96,1.00}
\definecolor {cadetblue2}          {rgb}{0.56,0.90,0.93}
\definecolor {cadetblue3}          {rgb}{0.48,0.77,0.80}
\definecolor {cadetblue4}          {rgb}{0.33,0.53,0.55}
\definecolor {turquoise1}          {rgb}{0.00,0.96,1.00}
\definecolor {turquoise2}          {rgb}{0.00,0.90,0.93}
\definecolor {turquoise3}          {rgb}{0.00,0.77,0.80}
\definecolor {turquoise4}          {rgb}{0.00,0.53,0.55}
\definecolor {cyan2}               {rgb}{0.00,0.93,0.93}
\definecolor {cyan3}               {rgb}{0.00,0.80,0.80}
\definecolor {cyan4}               {rgb}{0.00,0.55,0.55}
\definecolor {darkslategray1}      {rgb}{0.59,1.00,1.00}
\definecolor {darkslategray2}      {rgb}{0.55,0.93,0.93}
\definecolor {darkslategray3}      {rgb}{0.47,0.80,0.80}
\definecolor {darkslategray4}      {rgb}{0.32,0.55,0.55}
\definecolor {aquamarine2}         {rgb}{0.46,0.93,0.78}
\definecolor {aquamarine4}         {rgb}{0.27,0.55,0.45}
\definecolor {darkseagreen1}       {rgb}{0.76,1.00,0.76}
\definecolor {darkseagreen2}       {rgb}{0.71,0.93,0.71}
\definecolor {darkseagreen3}       {rgb}{0.61,0.80,0.61}
\definecolor {darkseagreen4}       {rgb}{0.41,0.55,0.41}
\definecolor {seagreen1}           {rgb}{0.33,1.00,0.62}
\definecolor {seagreen2}           {rgb}{0.31,0.93,0.58}
\definecolor {seagreen3}           {rgb}{0.26,0.80,0.50}
\definecolor {palegreen1}          {rgb}{0.60,1.00,0.60}
\definecolor {palegreen2}          {rgb}{0.56,0.93,0.56}
\definecolor {palegreen3}          {rgb}{0.49,0.80,0.49}
\definecolor {palegreen4}          {rgb}{0.33,0.55,0.33}
\definecolor {springgreen2}        {rgb}{0.00,0.93,0.46}
\definecolor {springgreen3}        {rgb}{0.00,0.80,0.40}
\definecolor {springgreen4}        {rgb}{0.00,0.55,0.27}
\definecolor {green2}              {rgb}{0.00,0.93,0.00}
\definecolor {green3}              {rgb}{0.00,0.80,0.00}
\definecolor {green4}              {rgb}{0.00,0.55,0.00}
\definecolor {chartreuse2}         {rgb}{0.46,0.93,0.00}
\definecolor {chartreuse3}         {rgb}{0.40,0.80,0.00}
\definecolor {chartreuse4}         {rgb}{0.27,0.55,0.00}
\definecolor {olivedrab1}          {rgb}{0.75,1.00,0.24}
\definecolor {olivedrab2}          {rgb}{0.70,0.93,0.23}
\definecolor {olivedrab4}          {rgb}{0.41,0.55,0.13}
\definecolor {darkolivegreen1}     {rgb}{0.79,1.00,0.44}
\definecolor {darkolivegreen2}     {rgb}{0.74,0.93,0.41}
\definecolor {darkolivegreen3}     {rgb}{0.64,0.80,0.35}
\definecolor {darkolivegreen4}     {rgb}{0.43,0.55,0.24}
\definecolor {khaki1}              {rgb}{1.00,0.96,0.56}
\definecolor {khaki2}              {rgb}{0.93,0.90,0.52}
\definecolor {khaki3}              {rgb}{0.80,0.78,0.45}
\definecolor {khaki4}              {rgb}{0.55,0.53,0.31}
\definecolor {lightgoldenrod1}     {rgb}{1.00,0.93,0.55}
\definecolor {lightgoldenrod2}     {rgb}{0.93,0.86,0.51}
\definecolor {lightgoldenrod3}     {rgb}{0.80,0.75,0.44}
\definecolor {lightgoldenrod4}     {rgb}{0.55,0.51,0.30}
\definecolor {lightyellow2}        {rgb}{0.93,0.93,0.82}
\definecolor {lightyellow3}        {rgb}{0.80,0.80,0.71}
\definecolor {lightyellow4}        {rgb}{0.55,0.55,0.48}
\definecolor {yellow2}             {rgb}{0.93,0.93,0.00}
\definecolor {yellow3}             {rgb}{0.80,0.80,0.00}
\definecolor {yellow4}             {rgb}{0.55,0.55,0.00}
\definecolor {gold2}               {rgb}{0.93,0.79,0.00}
\definecolor {gold3}               {rgb}{0.80,0.68,0.00}
\definecolor {gold4}               {rgb}{0.55,0.46,0.00}
\definecolor {goldenrod1}          {rgb}{1.00,0.76,0.15}
\definecolor {goldenrod2}          {rgb}{0.93,0.71,0.13}
\definecolor {goldenrod3}          {rgb}{0.80,0.61,0.11}
\definecolor {goldenrod4}          {rgb}{0.55,0.41,0.08}
\definecolor {darkgoldenrod1}      {rgb}{1.00,0.73,0.06}
\definecolor {darkgoldenrod2}      {rgb}{0.93,0.68,0.05}
\definecolor {darkgoldenrod3}      {rgb}{0.80,0.58,0.05}
\definecolor {darkgoldenrod4}      {rgb}{0.55,0.40,0.03}
\definecolor {rosybrown1}          {rgb}{1.00,0.76,0.76}
\definecolor {rosybrown2}          {rgb}{0.93,0.71,0.71}
\definecolor {rosybrown3}          {rgb}{0.80,0.61,0.61}
\definecolor {rosybrown4}          {rgb}{0.55,0.41,0.41}
\definecolor {indianred1}          {rgb}{1.00,0.42,0.42}
\definecolor {indianred2}          {rgb}{0.93,0.39,0.39}
\definecolor {indianred3}          {rgb}{0.80,0.33,0.33}
\definecolor {indianred4}          {rgb}{0.55,0.23,0.23}
\definecolor {sienna1}             {rgb}{1.00,0.51,0.28}
\definecolor {sienna2}             {rgb}{0.93,0.47,0.26}
\definecolor {sienna3}             {rgb}{0.80,0.41,0.22}
\definecolor {sienna4}             {rgb}{0.55,0.28,0.15}
\definecolor {burlywood1}          {rgb}{1.00,0.83,0.61}
\definecolor {burlywood2}          {rgb}{0.93,0.77,0.57}
\definecolor {burlywood3}          {rgb}{0.80,0.67,0.49}
\definecolor {burlywood4}          {rgb}{0.55,0.45,0.33}
\definecolor {wheat1}              {rgb}{1.00,0.91,0.73}
\definecolor {wheat2}              {rgb}{0.93,0.85,0.68}
\definecolor {wheat3}              {rgb}{0.80,0.73,0.59}
\definecolor {wheat4}              {rgb}{0.55,0.49,0.40}
\definecolor {tan1}                {rgb}{1.00,0.65,0.31}
\definecolor {tan2}                {rgb}{0.93,0.60,0.29}
\definecolor {tan4}                {rgb}{0.55,0.35,0.17}
\definecolor {chocolate1}          {rgb}{1.00,0.50,0.14}
\definecolor {chocolate2}          {rgb}{0.93,0.46,0.13}
\definecolor {chocolate3}          {rgb}{0.80,0.40,0.11}
\definecolor {firebrick1}          {rgb}{1.00,0.19,0.19}
\definecolor {firebrick2}          {rgb}{0.93,0.17,0.17}
\definecolor {firebrick3}          {rgb}{0.80,0.15,0.15}
\definecolor {firebrick4}          {rgb}{0.55,0.10,0.10}
\definecolor {brown1}              {rgb}{1.00,0.25,0.25}
\definecolor {brown2}              {rgb}{0.93,0.23,0.23}
\definecolor {brown3}              {rgb}{0.80,0.20,0.20}
\definecolor {brown4}              {rgb}{0.55,0.14,0.14}
\definecolor {salmon1}             {rgb}{1.00,0.55,0.41}
\definecolor {salmon2}             {rgb}{0.93,0.51,0.38}
\definecolor {salmon3}             {rgb}{0.80,0.44,0.33}
\definecolor {salmon4}             {rgb}{0.55,0.30,0.22}
\definecolor {lightsalmon2}        {rgb}{0.93,0.58,0.45}
\definecolor {lightsalmon3}        {rgb}{0.80,0.51,0.38}
\definecolor {lightsalmon4}        {rgb}{0.55,0.34,0.26}
\definecolor {orange2}             {rgb}{0.93,0.60,0.00}
\definecolor {orange3}             {rgb}{0.80,0.52,0.00}
\definecolor {orange4}             {rgb}{0.55,0.35,0.00}
\definecolor {darkorange1}         {rgb}{1.00,0.50,0.00}
\definecolor {darkorange2}         {rgb}{0.93,0.46,0.00}
\definecolor {darkorange3}         {rgb}{0.80,0.40,0.00}
\definecolor {darkorange4}         {rgb}{0.55,0.27,0.00}
\definecolor {coral1}              {rgb}{1.00,0.45,0.34}
\definecolor {coral2}              {rgb}{0.93,0.42,0.31}
\definecolor {coral3}              {rgb}{0.80,0.36,0.27}
\definecolor {coral4}              {rgb}{0.55,0.24,0.18}
\definecolor {tomato2}             {rgb}{0.93,0.36,0.26}
\definecolor {tomato3}             {rgb}{0.80,0.31,0.22}
\definecolor {tomato4}             {rgb}{0.55,0.21,0.15}
\definecolor {orangered2}          {rgb}{0.93,0.25,0.00}
\definecolor {orangered3}          {rgb}{0.80,0.22,0.00}
\definecolor {orangered4}          {rgb}{0.55,0.15,0.00}
\definecolor {red2}                {rgb}{0.93,0.00,0.00}
\definecolor {red3}                {rgb}{0.80,0.00,0.00}
\definecolor {red4}                {rgb}{0.55,0.00,0.00}
\definecolor {deeppink2}           {rgb}{0.93,0.07,0.54}
\definecolor {deeppink3}           {rgb}{0.80,0.06,0.46}
\definecolor {deeppink4}           {rgb}{0.55,0.04,0.31}
\definecolor {hotpink1}            {rgb}{1.00,0.43,0.71}
\definecolor {hotpink2}            {rgb}{0.93,0.42,0.65}
\definecolor {hotpink3}            {rgb}{0.80,0.38,0.56}
\definecolor {hotpink4}            {rgb}{0.55,0.23,0.38}
\definecolor {pink1}               {rgb}{1.00,0.71,0.77}
\definecolor {pink2}               {rgb}{0.93,0.66,0.72}
\definecolor {pink3}               {rgb}{0.80,0.57,0.62}
\definecolor {pink4}               {rgb}{0.55,0.39,0.42}
\definecolor {lightpink1}          {rgb}{1.00,0.68,0.73}
\definecolor {lightpink2}          {rgb}{0.93,0.64,0.68}
\definecolor {lightpink3}          {rgb}{0.80,0.55,0.58}
\definecolor {lightpink4}          {rgb}{0.55,0.37,0.40}
\definecolor {palevioletred1}      {rgb}{1.00,0.51,0.67}
\definecolor {palevioletred2}      {rgb}{0.93,0.47,0.62}
\definecolor {palevioletred3}      {rgb}{0.80,0.41,0.54}
\definecolor {palevioletred4}      {rgb}{0.55,0.28,0.36}
\definecolor {maroon1}             {rgb}{1.00,0.20,0.70}
\definecolor {maroon2}             {rgb}{0.93,0.19,0.65}
\definecolor {maroon3}             {rgb}{0.80,0.16,0.56}
\definecolor {maroon4}             {rgb}{0.55,0.11,0.38}
\definecolor {violetred1}          {rgb}{1.00,0.24,0.59}
\definecolor {violetred2}          {rgb}{0.93,0.23,0.55}
\definecolor {violetred3}          {rgb}{0.80,0.20,0.47}
\definecolor {violetred4}          {rgb}{0.55,0.13,0.32}
\definecolor {magenta2}            {rgb}{0.93,0.00,0.93}
\definecolor {magenta3}            {rgb}{0.80,0.00,0.80}
\definecolor {magenta4}            {rgb}{0.55,0.00,0.55}
\definecolor {orchid1}             {rgb}{1.00,0.51,0.98}
\definecolor {orchid2}             {rgb}{0.93,0.48,0.91}
\definecolor {orchid3}             {rgb}{0.80,0.41,0.79}
\definecolor {orchid4}             {rgb}{0.55,0.28,0.54}
\definecolor {plum1}               {rgb}{1.00,0.73,1.00}
\definecolor {plum2}               {rgb}{0.93,0.68,0.93}
\definecolor {plum3}               {rgb}{0.80,0.59,0.80}
\definecolor {plum4}               {rgb}{0.55,0.40,0.55}
\definecolor {mediumorchid1}       {rgb}{0.88,0.40,1.00}
\definecolor {mediumorchid2}       {rgb}{0.82,0.37,0.93}
\definecolor {mediumorchid3}       {rgb}{0.71,0.32,0.80}
\definecolor {mediumorchid4}       {rgb}{0.48,0.22,0.55}
\definecolor {darkorchid1}         {rgb}{0.75,0.24,1.00}
\definecolor {darkorchid2}         {rgb}{0.70,0.23,0.93}
\definecolor {darkorchid3}         {rgb}{0.60,0.20,0.80}
\definecolor {darkorchid4}         {rgb}{0.41,0.13,0.55}
\definecolor {purple1}             {rgb}{0.61,0.19,1.00}
\definecolor {purple2}             {rgb}{0.57,0.17,0.93}
\definecolor {purple3}             {rgb}{0.49,0.15,0.80}
\definecolor {purple4}             {rgb}{0.33,0.10,0.55}
\definecolor {mediumpurple1}       {rgb}{0.67,0.51,1.00}
\definecolor {mediumpurple2}       {rgb}{0.62,0.47,0.93}
\definecolor {mediumpurple3}       {rgb}{0.54,0.41,0.80}
\definecolor {mediumpurple4}       {rgb}{0.36,0.28,0.55}
\definecolor {thistle1}            {rgb}{1.00,0.88,1.00}
\definecolor {thistle2}            {rgb}{0.93,0.82,0.93}
\definecolor {thistle3}            {rgb}{0.80,0.71,0.80}
\definecolor {thistle4}            {rgb}{0.55,0.48,0.55}
\definecolor {gray1}               {rgb}{0.01,0.01,0.01}
\definecolor {gray2}               {rgb}{0.02,0.02,0.02}
\definecolor {gray3}               {rgb}{0.03,0.03,0.03}
\definecolor {gray4}               {rgb}{0.04,0.04,0.04}
\definecolor {gray5}               {rgb}{0.05,0.05,0.05}
\definecolor {gray6}               {rgb}{0.06,0.06,0.06}
\definecolor {gray7}               {rgb}{0.07,0.07,0.07}
\definecolor {gray8}               {rgb}{0.08,0.08,0.08}
\definecolor {gray9}               {rgb}{0.09,0.09,0.09}
\definecolor {gray10}              {rgb}{0.10,0.10,0.10}
\definecolor {gray11}              {rgb}{0.11,0.11,0.11}
\definecolor {gray12}              {rgb}{0.12,0.12,0.12}
\definecolor {gray13}              {rgb}{0.13,0.13,0.13}
\definecolor {gray14}              {rgb}{0.14,0.14,0.14}
\definecolor {gray15}              {rgb}{0.15,0.15,0.15}
\definecolor {gray16}              {rgb}{0.16,0.16,0.16}
\definecolor {gray17}              {rgb}{0.17,0.17,0.17}
\definecolor {gray18}              {rgb}{0.18,0.18,0.18}
\definecolor {gray19}              {rgb}{0.19,0.19,0.19}
\definecolor {gray20}              {rgb}{0.20,0.20,0.20}
\definecolor {gray21}              {rgb}{0.21,0.21,0.21}
\definecolor {gray22}              {rgb}{0.22,0.22,0.22}
\definecolor {gray23}              {rgb}{0.23,0.23,0.23}
\definecolor {gray24}              {rgb}{0.24,0.24,0.24}
\definecolor {gray25}              {rgb}{0.25,0.25,0.25}
\definecolor {gray26}              {rgb}{0.26,0.26,0.26}
\definecolor {gray27}              {rgb}{0.27,0.27,0.27}
\definecolor {gray28}              {rgb}{0.28,0.28,0.28}
\definecolor {gray29}              {rgb}{0.29,0.29,0.29}
\definecolor {gray30}              {rgb}{0.30,0.30,0.30}
\definecolor {gray31}              {rgb}{0.31,0.31,0.31}
\definecolor {gray32}              {rgb}{0.32,0.32,0.32}
\definecolor {gray33}              {rgb}{0.33,0.33,0.33}
\definecolor {gray34}              {rgb}{0.34,0.34,0.34}
\definecolor {gray35}              {rgb}{0.35,0.35,0.35}
\definecolor {gray36}              {rgb}{0.36,0.36,0.36}
\definecolor {gray37}              {rgb}{0.37,0.37,0.37}
\definecolor {gray38}              {rgb}{0.38,0.38,0.38}
\definecolor {gray39}              {rgb}{0.39,0.39,0.39}
\definecolor {gray40}              {rgb}{0.40,0.40,0.40}
\definecolor {gray42}              {rgb}{0.42,0.42,0.42}
\definecolor {gray43}              {rgb}{0.43,0.43,0.43}
\definecolor {gray44}              {rgb}{0.44,0.44,0.44}
\definecolor {gray45}              {rgb}{0.45,0.45,0.45}
\definecolor {gray46}              {rgb}{0.46,0.46,0.46}
\definecolor {gray47}              {rgb}{0.47,0.47,0.47}
\definecolor {gray48}              {rgb}{0.48,0.48,0.48}
\definecolor {gray49}              {rgb}{0.49,0.49,0.49}
\definecolor {gray50}              {rgb}{0.50,0.50,0.50}
\definecolor {gray51}              {rgb}{0.51,0.51,0.51}
\definecolor {gray52}              {rgb}{0.52,0.52,0.52}
\definecolor {gray53}              {rgb}{0.53,0.53,0.53}
\definecolor {gray54}              {rgb}{0.54,0.54,0.54}
\definecolor {gray55}              {rgb}{0.55,0.55,0.55}
\definecolor {gray56}              {rgb}{0.56,0.56,0.56}
\definecolor {gray57}              {rgb}{0.57,0.57,0.57}
\definecolor {gray58}              {rgb}{0.58,0.58,0.58}
\definecolor {gray59}              {rgb}{0.59,0.59,0.59}
\definecolor {gray60}              {rgb}{0.60,0.60,0.60}
\definecolor {gray61}              {rgb}{0.61,0.61,0.61}
\definecolor {gray62}              {rgb}{0.62,0.62,0.62}
\definecolor {gray63}              {rgb}{0.63,0.63,0.63}
\definecolor {gray64}              {rgb}{0.64,0.64,0.64}
\definecolor {gray65}              {rgb}{0.65,0.65,0.65}
\definecolor {gray66}              {rgb}{0.66,0.66,0.66}
\definecolor {gray67}              {rgb}{0.67,0.67,0.67}
\definecolor {gray68}              {rgb}{0.68,0.68,0.68}
\definecolor {gray69}              {rgb}{0.69,0.69,0.69}
\definecolor {gray70}              {rgb}{0.70,0.70,0.70}
\definecolor {gray71}              {rgb}{0.71,0.71,0.71}
\definecolor {gray72}              {rgb}{0.72,0.72,0.72}
\definecolor {gray73}              {rgb}{0.73,0.73,0.73}
\definecolor {gray74}              {rgb}{0.74,0.74,0.74}
\definecolor {gray75}              {rgb}{0.75,0.75,0.75}
\definecolor {gray76}              {rgb}{0.76,0.76,0.76}
\definecolor {gray77}              {rgb}{0.77,0.77,0.77}
\definecolor {gray78}              {rgb}{0.78,0.78,0.78}
\definecolor {gray79}              {rgb}{0.79,0.79,0.79}
\definecolor {gray80}              {rgb}{0.80,0.80,0.80}
\definecolor {gray81}              {rgb}{0.81,0.81,0.81}
\definecolor {gray82}              {rgb}{0.82,0.82,0.82}
\definecolor {gray83}              {rgb}{0.83,0.83,0.83}
\definecolor {gray84}              {rgb}{0.84,0.84,0.84}
\definecolor {gray85}              {rgb}{0.85,0.85,0.85}
\definecolor {gray86}              {rgb}{0.86,0.86,0.86}
\definecolor {gray87}              {rgb}{0.87,0.87,0.87}
\definecolor {gray88}              {rgb}{0.88,0.88,0.88}
\definecolor {gray89}              {rgb}{0.89,0.89,0.89}
\definecolor {gray90}              {rgb}{0.90,0.90,0.90}
\definecolor {gray91}              {rgb}{0.91,0.91,0.91}
\definecolor {gray92}              {rgb}{0.92,0.92,0.92}
\definecolor {gray93}              {rgb}{0.93,0.93,0.93}
\definecolor {gray94}              {rgb}{0.94,0.94,0.94}
\definecolor {gray95}              {rgb}{0.95,0.95,0.95}
\definecolor {gray97}              {rgb}{0.97,0.97,0.97}
\definecolor {gray98}              {rgb}{0.98,0.98,0.98}
\definecolor {gray99}              {rgb}{0.99,0.99,0.99}
\definecolor {darkgrey}            {rgb}{0.66,0.66,0.66}

\newcommand{\resp}[1]{[resp.\ #1]}

\newcommand{\TODO}[1]{{}}

\newcommand{\ignore}[1]{}

\newcommand{\RSTODO}[1]{{\bf \textcolor{darkgreen}{{\fbox{RS TODO:} #1}}}}
\renewcommand{\RSTODO}[1]{}

\newcommand{\ignoreinshort}[1]{}
\newcommand{\ignoreinlong}[1]{{#1}}

\providecommand{\longversion}{true}

\def\makenewenumerate#1#2{%
    \newcounter{cnt#1}
    \newenvironment{#1}%
    {\begin{list}{\makebox[0pt][r]{#2}}%
            {\setlength{\itemsep}{0pt}%
                \setlength{\parsep}{.2em}%
                \setlength{\leftmargin}{1.5em}%
                \setlength{\labelwidth}{.4em}%
                \usecounter{cnt#1}}}
            {\end{list}}}

\makenewenumerate{myenumerate}{\arabic{cntmyenumerate}.}

\makenewenumerate{renumerate}{\rm(\roman{cntrenumerate})}
\makenewenumerate{renumerateprime}{\rm(\roman{cntrenumerateprime}$'$)}
\makenewenumerate{renumeratesecond}{\rm(\roman{cntrenumeratesecond}$''$)}

\makenewenumerate{aenumerate}{({\it\alph{cntaenumerate}})}
\makenewenumerate{aenumerateprime}{({\it\alph{cntaenumerateprime}$'$})}
\makenewenumerate{aenumeratesecond}{({\it\alph{cntaenumeratesecond}$''$})}

\newcommand{\sref}[1]{\S{}\ref{#1}}

\newcommand{\set}[1]{\ensuremath{\{{#1}\}}\xspace}

\newcommand{\limp}{\ensuremath{\leftarrow}\xspace}
\renewcommand{\iff}{\ensuremath{\leftrightarrow}\xspace}
\newcommand{\defas}{\ensuremath{\stackrel{\text{\scalebox{.7}{def}}}{=}}\xspace}

\newcommand{\obdds}{\text{OBDD}s\xspace}
\newcommand{\obdd}{\textrm{OBDD}\xspace}

\newcommand{\sdds}{\text{SDD}s\xspace}
\newcommand{\sdd}{\textrm{SDD}\xspace}
\newcommand{\sddof}[1]{\textrm{SDD}{\ensuremath{(#1)}}\xspace}

\newcommand{\NNF}{\textrm{NNF}}
\newcommand{\NNFs}{\textrm{NNF}s}

\newcommand{\dDNNF}{\textrm{d-DNNF}}
\newcommand{\dDNNFs}{\textrm{d-DNNF}s}

\newcommand{\dDNNFof}[1]{\textrm{d-DNNF}{\ensuremath{(#1)}}}

\newcommand{\sdDNNF}{\textrm{sd-DNNF}}

\newcommand{\stardDNNFof}[1]{\textrm{*d-DNNF}{\ensuremath{(#1)}}}
\newcommand{\stardDNNF}{\textrm{*d-DNNF}}

\newcommand{\redOf}[2]{#1\ensuremath{_{\allalpha}^{\sf red}(#2)}}
\newcommand{\extOf}[2]{#1\ensuremath{_{\allalpha}^{\sf ext}(#2)}}

\newcommand{\Smooth}[1]{\ensuremath{{\sf Smooth}(#1)}}
\newcommand{\AlphaSmooth}{\ensuremath{\mbox{smooth}_{\allalpha}}}

\newcommand{\vi}{\ensuremath{\varphi}}
\renewcommand{\vi}[1]{\ensuremath{\varphi_{#1}}}

\newcommand{\vip}{\ensuremath{\varphi^p}\xspace}
\renewcommand{\vip}[1]{\ensuremath{\varphi^p_{#1}}\xspace}
\newcommand{\mup}{\ensuremath{\mu^p}\xspace}
\newcommand{\etap}{\ensuremath{\eta^p}\xspace}

\newcommand{\atoms}[1]{\ensuremath{Atoms(#1)}\xspace}
\renewcommand{\atoms}[1]{\ensuremath{\mathsf{Atoms(#1)}\xspace}}

\newcommand{\C}{\ensuremath{\mathcal{C}}\xspace}

\newcommand\mysout{\bgroup \markoverwith{{-}}\ULon}
\newcommand\nosout{\bgroup \markoverwith{{ }}\ULon}
\definecolor{mygray}{rgb}{0.90,0.90,0.90}
\definecolor{mywhite}{rgb}{1.00,1.00,1.00}

\newcommand{\proptofol}{\ensuremath{{\cal B}2{\cal T}}\xspace}
\newcommand{\foltoprop}{\ensuremath{{\cal T}2{\cal B}}\xspace}

\newcommand{\B}{\ensuremath{\mathcal{B}}\xspace}
\newcommand{\T}{\ensuremath{\mathcal{T}}\xspace}

\newcommand{\smttt}[1]{\ensuremath{\text{SMT}(#1)}\xspace}

\newcommand{\euf}{\ensuremath{\mathcal{EUF}}\xspace}

\newcommand{\eq}{\ensuremath{\mathcal{E}}\xspace}

\newcommand{\larat}{\ensuremath{\mathcal{LA}(\mathbb{Q})}\xspace}
\newcommand{\laint}{\ensuremath{\mathcal{LA}(\mathbb{Z})}\xspace}

\renewcommand{\larat}{\ensuremath{\mathcal{LRA}}\xspace}
\renewcommand{\laint}{\ensuremath{\mathcal{LIA}}\xspace}

\newcommand{\bv}{\ensuremath{\mathcal{BV}}\xspace}
\newcommand{\mem}{\ensuremath{\mathcal{AR}}\xspace}

\newcommand{\smtlarat}{\smttt{\larat}}

\newcommand{\Tmodels}{\models_{\T}}

\newcommand{\pmodels}{\models_p}

\newcommand{\mathsat}{\textsc{MathSAT}\xspace}

\newcommand{\mathsatfive}{\textsc{MathSAT5}\xspace}

\renewcommand{\TODO}[1]{\todo[inline,color=green!40,caption={}]{{\small{TODO: #1}}}}

\renewcommand{\RSTODO}[1]{\todo[inline,color=green!40,caption={}]{{\small{RS TODO: #1}}}}

\newcommand{\allA}{\ensuremath{\mathbf{A}}\xspace}
\newcommand{\allalpha}{\ensuremath{\boldsymbol{\alpha}}\xspace}
\newcommand{\allalphaprime}{\ensuremath{\boldsymbol{\alpha}'}\xspace}

\newcommand{\residual}[2]{\ensuremath{#1|_{#2}}\xspace}

\renewcommand{\B}{\ensuremath{\mathbb{B}}\xspace}
\newcommand{\pequiv}{\equiv_p}

\newcommand{\allsym}[1]{\ensuremath{{\underline{\boldsymbol#1}}}}
\renewcommand{\allsym}[1]{\ensuremath{{\boldsymbol#1}}}
\renewcommand{\allA}{\allsym{A}}

\renewcommand{\allalpha}{\allsym{\alpha}}

\renewcommand{\C}[1]{\ensuremath{C_{#1}}}

\newcommand{\via}[1]{\ensuremath{\varphi_{#1}{[\allalpha]}}}

\newcommand{\psia}[1]{\ensuremath{\psi_{#1}{[\allalpha]}}}

\newcommand{\psip}[1]{\ensuremath{\psi_{#1}^p}}
\newcommand{\psiap}[1]{\ensuremath{\psi_{#1}^p{[\allA]}}}

\newcommand{\viap}[1]{\ensuremath{\varphi_{#1}^{p}{[\allA]}}}

\newcommand{\etaprime}[1]{\ensuremath{\eta'_{#1}}}
\newcommand{\rhoprime}[1]{\ensuremath{\rho'_{#1}}}

\newcommand{\bequiv}{\ensuremath{\equiv_{\B}}}
\newcommand{\Tequiv}{\ensuremath{\equiv_{\T}}}

\newcommand{\theorycl}{\ensuremath{\{C_{1},\ldots,C_{K}\}}}

\renewcommand{\pmodels}{\models_{\B}}
\newcommand{\bmodels}{\models_{\B}}

\newcommand{\Aone}{\ensuremath{A_1}}
\newcommand{\Atwo}{\ensuremath{A_2}}

\newcommand{\aone}{\ensuremath{(x \le 0)}\xspace}
\newcommand{\atwo}{\ensuremath{(x = 1)}\xspace}
\newcommand{\athree}{\ensuremath{(x \ge 2)}\xspace}

\newcommand{\btwo}{\ensuremath{(x \ge 1)}\xspace}
\newcommand{\atoma}{\aone}
\newcommand{\atomb}{\btwo}
\newcommand{\atomc}{\athree}

\newcommand{\CTTA}[1]{\ensuremath{H}(#1)} %
\newcommand{\ITTA}[1]{\ensuremath{P}(#1)} %
\renewcommand{\CTTA}[1]{\ensuremath{H_{\allalpha}(#1)}} %
\renewcommand{\ITTA}[1]{\ensuremath{P_{\allalpha}(#1)}} %
\newcommand{\CTTAprime}[1]{\ensuremath{H_{\allalpha{}'}(#1)}} 
\newcommand{\ITTAprime}[1]{\ensuremath{P_{\allalpha{}'}(#1)}} 
\newcommand{\TLEMMAS}[1]{\ensuremath{Cl_{\allalpha}(#1)}}

\newcommand{\OrOf}[2]{\ensuremath{\bigvee_{#2\in#1}#2}}

\newcommand{\Treduceda}{\T-reduced\xspace}
\renewcommand{\Treduceda}{\T-reduced$_{\allalpha}$\xspace}
\renewcommand{\Treduceda}{$\T$-reduced$_{\allalpha}$\xspace}

\newcommand{\Treduction}{\T-reduction\xspace}
\newcommand{\Treduced}{{\T-reduced}\xspace}
\newcommand{\TredOf}[1]{\ensuremath{#1_{\sf red}}}
\renewcommand{\TredOf}[1]{\T{}\!{\sf red}(\ensuremath{#1})}
\renewcommand{\TredOf}[1]{\T{}\!{\sf red}\ensuremath{_{\allalpha}(#1)}}
\newcommand{\Textendeda}{\T-extended\xspace}
\renewcommand{\Textendeda}{\T-extended$_{\allalpha}$\xspace}
\renewcommand{\Textendeda}{$\T$-extended$_{\allalpha}$\xspace}

\newcommand{\Textended}{{\T-extended\xspace}}
\newcommand{\Textension}{\T-extension\xspace}
\newcommand{\TextOf}[1]{\ensuremath{#1_{\sf ext}}}
\renewcommand{\TextOf}[1]{\T{}\!{\sf ext}\ensuremath{_{\allalpha}(#1)}}

\newcommand{\dfour}{\textsc{D4}\xspace}
\newcommand{\cudd}{\textsc{CUDD}\xspace}
\newcommand{\pysdd}{\textsc{PySDD}\xspace}
\newcommand{\ddnnife}{\textsc{ddnnife}\xspace}

\DeclareMathSymbol{\mhyphen}{\mathord}{AMSa}{"39}
\newcommand{\Input}{\textbf{input:}\xspace}
\newcommand{\Output}{\textbf{output:}\xspace}
\newcommand{\vimu}[1]{\ensuremath{\residual{\varphi{}_{#1}}{\mu}}}
\newcommand{\dDNNFcompile}{\ensuremath{\mathsf{dDNNFcompile}}}

\newcommand{\Partition}{\ensuremath{\mathsf{Partition}}}
\newcommand{\Select}{\ensuremath{\mathsf{Select}}}
\newcommand{\LemmaEnumerator}{\ensuremath{\mathsf{\T\!\mhyphen{}lemmaEnum}}}
\renewcommand{\LemmaEnumerator}{\ensuremath{\mathsf{\T\!\mhyphen{}lemmaEnum_{\allalpha}}}}
\newcommand{\BoolNNFCompiler}{\ensuremath{\mathsf{*dDNNFcompile}}}

\newcommand{\TredNNFCompilerOf}[1]{\ensuremath{\mathsf{\T{}^{red}_{\allalpha}}\mhyphen{}\BoolNNFCompiler(#1)}}
\newcommand{\TextNNFCompilerOf}[1]{\ensuremath{\mathsf{\T{}^{ext}_{\allalpha}}\mhyphen{}\BoolNNFCompiler(#1)}}

\newcommand{\SharpSAT}{\ensuremath{{\sf \#SAT\!_{\allA}}\xspace}}
\newcommand{\SharpSMT}{\ensuremath{{\sf \#SMT}\xspace}}
\renewcommand{\SharpSMT}{\ensuremath{{\mbox{$\T_{\allalpha}$-}{\sf \#SAT}}\xspace}}
\renewcommand{\SharpSMT}{\ensuremath{{\mbox{$\T$-}{\sf \#SAT_{\allalpha}}}\xspace}}
\newcommand{\SharpSATOf}[1]{\ensuremath{\mathit{\sf \#SAT\!_{\allA}}(#1)\xspace}}
\newcommand{\SharpSMTOf}[1]{\ensuremath{\mathit{\mbox{$\T_{\allalpha}$-}{\sf \#SAT}}(#1)\xspace}}
\renewcommand{\SharpSMTOf}[1]{\ensuremath{\mathit{\mbox{$\T$-}{\sf \#SAT_{\allalpha}}}(#1)\xspace}}
\newcommand{\AllSAT}{\ensuremath{\sf AllSAT\!_{\allA}}\xspace}
\newcommand{\AllSMT}{{\sf AllSMT}\xspace}
\renewcommand{\AllSMT}{\ensuremath{{\mbox{$\T_{\allalpha}$-}{\sf AllSAT}_{\allalpha}}\xspace}}
\renewcommand{\AllSMT}{\ensuremath{{\mbox{$\T$-}{\sf AllSAT_{\allalpha}}}\xspace}}
\newcommand{\AllSATOf}[1]{\ensuremath{\mathit{\sf AllSAT}\!_{\allA}(#1)\xspace}}
\newcommand{\AllSMTOf}[1]{\ensuremath{\mathit{\sf AllSMT}(#1)\xspace}}
\renewcommand{\AllSMTOf}[1]{\ensuremath{\mathit{\mbox{$\T_{\allalpha}$-}{\sf AllSAT}}(#1)\xspace}}
\renewcommand{\AllSMTOf}[1]{\ensuremath{\mathit{\mbox{$\T$-}{\sf AllSAT_{\allalpha}}}(#1)\xspace}}

\newcommand{\resvil}{\ensuremath{\residual{\varphi{}}{l}}}

\usepackage{tikz}
\usetikzlibrary{arrows.meta, positioning, calc, shadows, shapes.geometric}

\tikzset{
    andNode/.style={circle, draw=teal, text=teal, thick, minimum size=6mm, inner sep=0pt, font=\large},
    orNode/.style={circle, draw=orangered, text=orangered, thick, minimum size=6mm, inner sep=0pt, font=\large},
    lit/.style={font=\Large, thick},
    edge/.style={->, thick, >=stealth}
}

\newif\ifnoproofs
\noproofsfalse

\makeatletter%
\@mparswitchfalse%
\makeatother%
\reversemarginpar%

\newcommand{\citet}[1]{\cite{#1}}

\title{%
  d-DNNF Modulo Theories:\texorpdfstring{\\}{}
   A General Framework for Polytime SMT Queries
} 

\author{Gabriele {Masina}}{DISI, University of Trento, Italy}{gabriele.masina@unitn.it}{https://orcid.org/0000-0001-8842-4913}{}
\author{Emanuele {Civini}}{DISI, University of Trento, Italy}{emanuele.civini@studenti.unitn.it}{https://orcid.org/0009-0009-7212-345X}{}
\author{Massimo {Michelutti}}{DISI, University of Trento, Italy}{massimo.michelutti00@gmail.com}{https://orcid.org/0009-0002-3490-2910}{}
\author{Giuseppe {Spallitta}}{Rice University, TX, USA}{gs81@rice.edu}{https://orcid.org/0000-0002-4321-4995}{}
\author{Roberto {Sebastiani}}{DISI, University of Trento,
Italy}{roberto.sebastiani@unitn.it}{https://orcid.org/0000-0002-0989-6101}{%
}

\titlerunning{d-DNNF Modulo Theories: A General Framework for Polytime SMT Queries} %

\authorrunning{G. Masina, E. Civini, M. Michelutti, G. Spallitta, and  R. Sebastiani} %

\Copyright{Jane Open Access and Joan R. Public} %



\category{} %

\supplement{}

\supplementdetails[subcategory={Tool source code}, cite={}, swhid={}]{Software}{https://github.com/ecivini/tddnnf}

\supplementdetails[subcategory={Experiments Source Code}, cite={}, swhid={}]{Software}{https://github.com/ecivini/tddnnf-testbench}

\EventEditors{Alexey Ignatiev and Stefan Szeider}
\EventNoEds{2}
\EventLongTitle{29th International Conference on Theory and Applications of Satisfiability Testing (SAT 2026)}
\EventShortTitle{SAT 2026}
\EventAcronym{SAT}
\EventYear{2026}
\EventDate{July 20--23, 2026}
\EventLocation{Lisbon, Portugal}
\EventLogo{}
\SeriesVolume{377}
\ArticleNo{31}

\begin{document}

\maketitle

\renewcommand{\longversion}{false}
\ifthenelse{\equal{\longversion}{true}}
{%
  \renewcommand{\ignoreinshort}[1]{{\textcolor{midnightblue}{#1}}}
  \renewcommand{\ignoreinlong}[1]{}
  \specialcomment{IGNOREINSHORT}{\begingroup\color{midnightblue}}{\endgroup}
  \excludecomment{IGNOREINLONG}
}%
{%
 \renewcommand{\ignoreinshort}[1]{}
\renewcommand{\ignoreinlong}[1]{#1}
 \excludecomment{IGNOREINSHORT}
 \specialcomment{IGNOREINLONG}{\begingroup}{\endgroup}
}
 \excludecomment{IGNORE}

\begin{abstract}
    In Knowledge Compilation (KC) a propositional knowledge base is
compiled off-line into some target form, typically into
deterministic decomposable negation normal form  (d-DNNF) or one of
its subcases,
which is then used on-line
to answer a large number of queries in polytime, such as clausal
entailment, model counting, and others.
The general idea is to push as much of the computational effort into 
the off-line compilation phase, which is amortized over all on-line
polytime queries.

In this paper, we present for the first time
a novel and general technique to leverage
d-DNNF compilation and querying to SMT level.
Intuitively, before d-DNNF compilation, the input SMT formula is combined with a
list of pre-computed ad-hoc theory lemmas, so that the queries at SMT level reduce to those
at propositional level. 
This approach has several features:
(i) 
it works for every theory, or theory combination thereof;
(ii)
it works for all forms of d-DNNF;
(iii)
it is easy to
implement on top of any d-DNNF compiler and any %
theory-lemma enumerator,
which are used as black boxes;
(iv)
most importantly,
{\em these compiled SMT d-DNNFs can be queried in
  polytime by means of a standard propositional d-DNNF reasoner.}
As proof of concept,
we have implemented a %
tool on top of state-of-the-art d-DNNF
packages and of the MathSAT SMT solver. Some preliminary empirical
evaluation supports the feasibility and effectiveness of the approach.

\end{abstract}

\section{Introduction}%
\label{sec:intro}

\subparagraph*{Context.}
Knowledge Compilation (KC), see e.g.~\cite{darwicheKnowledgeCompilationMap2002},
is a research field that aims to address the
computational intractability of general propositional reasoning.
In KC, a propositional knowledge base is
compiled off-line into some target form, typically 
(some subcase of) 
{\em deterministic decomposable negation normal form  (\dDNNF{})}~\cite{darwicheKnowledgeCompilationMap2002}.
Such an encoded knowledge base  is then used on-line
to answer a large number of queries in polytime.
The general idea is to push as much of the computational effort as
possible into 
the off-line compilation phase, which may be computationally demanding,
but whose cost is then amortized over all on-line
polytime queries.

\dDNNFs{} are particular forms of DAG-encoded NNF propositional
formulas, which are both {\em deterministic} (in every disjunctive subformula all disjuncts
are mutually inconsistent) and {\em decomposable} (in every
conjunctive subformula all conjuncts do not share variables).
Remarkably, with a \dDNNF{} formula, the checks for 
{\em consistency  (CO)},
{\em validity  (VA)},
{\em clausal entailment (CE)},
{\em implicant (IM)},
{\em model counting (CT)} and
{\em model enumeration (ME)}~%
\footnote{Notationally, we adopt the shortcuts ``CO'', ``VA'', ``CE'', ``IM'', ``CT'', ``ME'', ``EQ'', ``SE'' from~%
\cite{darwicheKnowledgeCompilationMap2002}.}
can be performed in polytime~%
\cite{darwicheKnowledgeCompilationMap2002}.~%
\footnote{For model
enumeration (ME),
``polytime'' means ``polynomial wrt.\ the number of models enumerated''~%
\cite{darwicheKnowledgeCompilationMap2002}.}
Also, subcases of \dDNNFs{} such as SDDs~\cite{darwicheSDDNewCanonical2011} 
and OBDDs~\cite{bryantGraphBasedAlgorithmsBoolean1986}, under specific ordering
conditions, allow for both {\em equivalence (EQ)} and {\em sentential entailment (SE)}
checks in polytime~\cite{darwicheKnowledgeCompilationMap2002}, and
they are {\em canonical}, that is, SDDs/OBDDs encodings of equivalent
formulas are %
syntactically identical.~%
\footnote{Notice that OBDDs~\cite{bryantGraphBasedAlgorithmsBoolean1986} were conceived and are
mostly used in the field of formal verification, and are usually
represented as if-then-else DAGs with two leaves: $\top,\bot$. However, if we rewrite each 
``if $A_i$ then $\residual{\vi{}}{A_i}$ else $\residual{\vi{}}{\neg
A_i}$'' node of an OBDD into
``$(A_i\wedge\residual{\vi{}}{A_i})\vee(\neg
A_i\wedge\residual{\vi{}}{\neg A_i})$'' recursively,
then we obtain a \dDNNF{}~\cite{darwicheKnowledgeCompilationMap2002}.
}

Although extensive, the literature on KC, particularly on \dDNNFs{}, is mostly restricted to
the propositional case, which limits its expressiveness and
application potential.

\subparagraph*{Contributions.}

In this paper, we investigate for the first time the problem of
leveraging \dDNNF{} compilation and querying  to the SMT level, with the general
goal of preserving the properties of the corresponding propositional
\dDNNFs{}, {\em in particular the polynomial-time complexity of the above-mentioned
queries}.
We notice first that the traditional ``lazy'' lemmas-on-demand
approach, which is widely adopted in standard SMT solving and enumeration %
(see
e.g.~\cite{barrettSatisfiabilityModuloTheories2021,lahiriSMTTechniquesFast2006}),
does not seem to be suitable 
for SMT-level \dDNNF{} compilation. 
Then, we analyze the problem theoretically in detail, 
and we introduce a novel and very general framework 
extending \dDNNFs{} to the realm of SMT, identifying and proving the theoretical 
properties that allow us to comply with our general goal above.

Intuitively, before \dDNNF{} compilation, the input SMT formula can be ``eagerly''
combined with an ad-hoc 
list of pre-computed theory lemmas, so that the queries at the SMT level reduce to those
at the propositional level. 
Overall, our approach has several features:
\begin{itemize}
\item
it works for every theory, or theory combination thereof;
\item
it works for all forms of d-DNNF;
\item
it is easy to
implement on top of any d-DNNF compiler and any %
theory-lemma enumerator,
which are used as black boxes;
\item
most importantly,
{\em these compiled SMT d-DNNFs can be queried in
  polytime by means of a standard propositional d-DNNF reasoner.}
\end{itemize}
Notice that %
the ``eager'' generation of all the needed theory lemmas upfront,
although inefficient and obsolete in standard SMT solving, is very suitable for KC,
whose rationale is to shift the bulk of computational effort to the off-line
compilation phase, thereby optimizing the efficiency of the on-line
query-answering phase. To this extent, we move all the 
theory-reasoning effort to the %
compilation phase, by means of off-line
theory-lemma enumeration. %
{(We recall that for many relevant theories %
(e.g.\ \laint{}, \mem{}, \bv{}, \ldots) even the satisfiability of the
conjunctive fragment is NP-hard~\cite{barrettSatisfiabilityModuloTheories2021,sebastianiColorsMakeTheories2016}.)}

We have implemented an SMT-level \dDNNF{} compilation and querying
tool on top of state-of-the-art d-DNNF
packages and of our novel theory-agnostic lemma enumerator~%
\cite{civiniEagerEncodingsTheoryAgnostic2026}.
A preliminary empirical
evaluation supports the feasibility and effectiveness of the approach.

\subparagraph*{Related Work.}
 To the best of our knowledge, there is no previous work on leveraging
 general \dDNNFs{} to SMT level.
 Some works are restricted to OBDDs
 ~\cite{mollerDifferenceDecisionDiagrams1999,bryantBooleanSatisfiabilityTransitivity2002,goelBDDBasedProcedures2003,vandepolBDDRepresentationLogicEquality2005,cavadaComputingPredicateAbstractions2007,chakiDecisionDiagramsLinear2009,micheluttiCanonicalDecisionDiagrams2024},
 a couple to SDDs
 \cite{dosmartiresExactApproximateWeighted2019,micheluttiCanonicalDecisionDiagrams2024}.
(In particular, \cite{dosmartiresExactApproximateWeighted2019}
 produces SDDs of the Boolean abstraction of the formula and filters out
 the \T-inconsistent models during enumeration.)
 \cite{derkinderenTopDownKnowledgeCompilation2023} sketches a possible 
 procedure  to produce
 a decision \dDNNF{} which is 
 free from \T-inconsistent cubes by DPLL-like lazy SMT enumeration (without
 partitioning, see \sref{sec:background}).
 A comprehensive survey can be found
 in~\cite{micheluttiCanonicalDecisionDiagrams2024}.
 
In~\cite{vandenbroeckLiftedProbabilisticInference2011}, the authors propose a KC framework for weighted model counting for a function-free finite-domain fragment of uninterpreted first-order logic, thus addressing a related, yet substantially different problem from ours. In contrast, our setting targets SMT, requiring reasoning modulo background theories through \T-lemmas, and supports modulo-theory extensions of the full range of polytime queries available on \dDNNFs{}.

 The closest to this paper is our own work in
 \cite{micheluttiCanonicalDecisionDiagrams2024},
 where we proposed a
 theory-agnostic method for building {\em canonical decision diagrams}
 (e.g., \obdds{} and \sdds{}) {\em modulo %
 theories}, by augmenting the input formula $\vi{}$ with all the \T-lemmas
 which could be obtained from a total SMT enumeration call on $\vi{}$. 
 The main focus and key contribution of
 \cite{micheluttiCanonicalDecisionDiagrams2024} is that the resulting decision diagrams are  {\em
   theory-canonical}, i.e., \T-equivalent
 \T-formulas are encoded into the same OBDD/SDD, so that their
 \T-equivalence can be checked in
 constant time. 
 To this extent, this paper 
 generalizes and
 extends \cite{micheluttiCanonicalDecisionDiagrams2024} to all forms
 of 
 \dDNNFs{}, focusing on leveraging all the respective
 polytime queries to the SMT level. In particular, all the results in
 \cite{micheluttiCanonicalDecisionDiagrams2024} are captured by those
 in this paper.

 Also, %
 the approach in this paper is 
 enabled technologically by our recent theory-agnostic 
 lemma-enumeration techniques described in
 \cite{civiniEagerEncodingsTheoryAgnostic2026}, which are highly
 parallelizable and dramatically
 faster than the %
 technique we used in
 \cite{micheluttiCanonicalDecisionDiagrams2024}, and enable scaling to
 much larger problems.
 These lemma-enumeration techniques are very
 different from those used in
 the early ``eager encodings'' in
 SMT~\cite{velevEffectiveUseBoolean2001,strichmanDecidingSeparationFormulas2002,strichmanSolvingPresburgerLinear2002}, since the latter are specific for very few and easy theories, they do not comply with
theory combination, and may produce lots of lemmas that are not
necessary to our goals~\cite{civiniEagerEncodingsTheoryAgnostic2026}.

\subparagraph*{Content.}
In \sref{sec:background}, we introduce the necessary background on \dDNNFs{} and SMT. Next, in \sref{sec:problem}, we analyze the challenges of lifting \dDNNFs{} and their properties to the SMT level. In \sref{sec:framework}, we present a formal framework for reducing several SMT-level queries to propositional queries, when the input formula is augmented with a suitable set of \T-lemmas. 
In \sref{sec:nnf} we show how to apply this framework to \dDNNFs{}, exploiting their structural properties to support a wide range of queries in polynomial time. In \sref{sec:experiments}, we present preliminary empirical results on the effectiveness of our approach. Finally, in \sref{sec:conclusions}, we draw our conclusions and outline future research directions.

\section{Background}%
\label{sec:background}

\subparagraph*{Notation \& terminology.}
We assume the reader is familiar with the syntax, semantics, and results of propositional and first-order logics.
We adopt the following terminology and notation.

{Propositional satisfiability (SAT) is the problem of
deciding the satisfiability of propositional formulas.}
A propositional formula \vi{} is either a Boolean constant ($\top$ or $\bot$
representing ``true'' and ``false'', respectively), a Boolean atom, or
a formula built from propositional connectives
($\neg,\land,\lor,\to,\limp,\leftrightarrow,\oplus$) over propositional
formulas.
A {\em literal} is either an atom $A_i$ (a {\em
  positive literal}) or its negation $\neg A_i$ (a {\em negative
  literal}); a \emph{clause} is a
disjunction of literals, and a \emph{cube} is a conjunction of
literals.

A formula is in \emph{Negation Normal Form (NNF)}, if it only contains $\wedge$ or $\vee$ connectives, and negations only occur at the level of literals.
A formula is in \emph{Conjunctive Normal Form (CNF)} if it is a
conjunction ($\wedge$) of clauses. 
We represent (possibly partial) truth assignments as cubes
$\mu\defas\bigwedge l_i$,
with the intended meaning that positive and negative literals $A_i$
and $\neg A_i$ in $\mu$ mean that $A_i$ is assigned to $\top$ and
$\bot$ respectively.
We denote by ``$\residual{\vi{}}{\mu}$'' (``residual of $\vi{}$ under
$\mu$'') the formula obtained by substituting in $\vi{}$ the atoms in $\mu$ with
their assigned truth constants and propagate them in the standard way
($\vi{}\wedge\top\Rightarrow \vi{}$, $\vi{}\wedge\bot\Rightarrow
\bot$, etc.).%

Satisfiability Modulo Theories (SMT) extends SAT to the context of
first-order formulas modulo some background theory \T, which provides
an intended interpretation for constant, function, and predicate
symbols (see, e.g.,~\cite{barrettSatisfiabilityModuloTheories2021}). We restrict to quantifier-free formulas. 
A \T-formula is a combination of theory-specific atoms (\T-atoms) and
Boolean atoms (\B-atoms) via Boolean connectives
(``atoms'' denotes \T- and \B-atoms indifferently).

We denote by \atoms{\vi{}} the set of atoms %
occurring in \vi{}.
For instance, the theory of linear real arithmetic (\larat) provides
the standard interpretations of arithmetic operators ($+$, $-$,
$\cdot$) and relations ($=$, $\neq$, $\le$, $\ge$ $<$, $>$) over the reals. %
\larat-atoms are linear (in)equalities over
rational variables. 
An example of \larat-formula is $((x-y\le 3)\vee (x=z))$, where $x,y,z$ are \larat-variables, and $(x-y\le 3),(x=z)$ are \larat-atoms.
Other theories of interest include, e.g., equalities (\eq), equalities with uninterpreted functions (\euf), bit-vectors (\bv), arrays (\mem), and combinations thereof.
We assume w.l.o.g.\ that we have no \T-valid
  or \T-inconsistent \T-atom 
 (like e.g. $(x\le x)$ or $(x>x)$), because we may convert them
 upfront into $\top$ and $\bot$ respectively.

We assume that Boolean and \T-formulas are represented as
single-rooted Directed Acyclic Graphs (DAGs), in which internal nodes
are labelled with Boolean connectives, and leaf nodes with literals or
Boolean constants, so that syntactically identical sub-formulas are
represented by the same sub-DAG. The \emph{size} of a formula is the
number of nodes in its DAG representation. 

\foltoprop{} is a bijective function (``theory to Boolean''),
called {\em Boolean 
abstraction},
which maps Boolean atoms into themselves,
 \T{}-atoms into fresh Boolean atoms, 
and is homomorphic wrt.\ Boolean connectives and set inclusion.
The function \proptofol{} (``Boolean to theory''), called {\em
  refinement},  is the inverse of \foltoprop. 
  (For instance
$\foltoprop(\{((x-y\le 3)\vee (x=z)) \}) =
\{(A_1\vee A_2) \}$, $A_1$ and $A_2$ being fresh
Boolean variables, and $\proptofol(\{\neg A_1, A_2\})=
\{\neg (x-y\le 3), (x=z)\}$.)%

The symbol
$\allalpha\defas\set{\alpha_i}_i$,
possibly with subscripts or superscripts, denotes
a set of atoms, %
and $\allA\defas\set{A_i}_i$ denotes its Boolean abstraction.
We
denote by
$2^{\allalpha}$ the set of all total truth assignments on \allalpha.
The symbols
$\vi{}$, $\psi$ denote \T{}-formulas, and 
$\mu$, $\eta$, $\rho$ denote conjunctions of \T{}-literals;
 $\vip{}$, $\psi^p$  denote Boolean formulas,
 $\mu^p$,  $\eta^p$, $\rho^p$  
denote conjunctions of Boolean literals 
and we use them as synonyms for the Boolean abstraction of
$\vi{}$, $\psi$, $\mu$, $\eta$,  and $\rho$ respectively,  and vice versa 
(e.g., $\vip{}$ denotes $\foltoprop(\vi{})$,
$\eta$ denotes $\proptofol(\etap)$). 
If $\foltoprop(\eta) \models \foltoprop(\vi{})$, then we say that $\eta$
\emph{propositionally (or \B-)satisfies} $\vi{}$, written
$\eta\pmodels\vi{}$. 
(Notice that if $\eta\pmodels\vi{}$ then
$\eta\Tmodels\vi{}$, but not vice versa.)
The notion of propositional/\B- satisfiability,
entailment and validity follow straightforwardly.
When both $\vi{}\pmodels\psi$ and $\psi\pmodels\vi{}$, we say that
$\vi{}$ and $\psi$ are {\em propositionally/\B- equivalent}, written 
``$\vi{}\bequiv\psi$''.
When both $\vi{}\Tmodels\psi$ and $\psi\Tmodels\vi{}$, we say that
$\vi{}$ and $\psi$ are {\em \T-equivalent}, written 
``$\vi{}\Tequiv\psi$''.
 (Notice that if $\eta\bequiv\vi{}$ then
 $\eta\Tequiv\vi{}$, but not vice versa.)
We call a {\em \T-lemma} any \T-valid clause.

\subparagraph*{Deterministic Decomposable Negation Normal Form, \dDNNFs{}.}
\begin{figure}
  \centering
  \begin{subfigure}[T]{0.25\textwidth}
    \centering
    \begin{tikzpicture}[scale=0.7,transform shape]
    \node[andNode] (root) at (0, 0) {$\wedge$};

    \node[orNode] (or1) at (-1, -1) {$\vee$};
    \node[orNode] (or2) at (1, -1) {$\vee$};

    \node[lit] (na1)  at (-1.5, -2) {$\neg A_1$};
    \node[lit] (a2) at (-0.5, -2) {$A_2$};
    \node[lit] (na2) at (0.5, -2) {$\neg A_2$};
    \node[lit] (a3) at (1.5, -2) {$A_3$};

    \draw[edge] (root) -- (or1);
    \draw[edge] (root) -- (or2);

    \draw[edge] (or1) -- (na1);
    \draw[edge] (or1) -- (a2);

    \draw[edge] (or2) -- (na2);
    \draw[edge] (or2) -- (a3);
\end{tikzpicture}
  \end{subfigure}
  \begin{subfigure}[T]{0.3\textwidth}
    \centering
    \begin{tikzpicture}[scale=0.7,transform shape]
    \node[orNode] (root) at (0, 0) {$\lor$};

    \node[andNode] (and1) at (-1, -1) {$\land$};
    \node[andNode] (and2) at (1, -1) {$\land$};

    \node[lit] (a1)  at (-2, -2) {$A_1$};
    \node[lit] (na1) at (0.5, -2) {$\neg A_1$};
    \node[orNode]  (or2) at (2, -2) {$\lor$};

    \node[andNode]  (and3) at (1.5, -3) {$\land$};
    \node[lit] (na2)  at (2.5, -3) {$\neg A_2$};

    \node[lit] (a2) at (-1, -4) {$A_2$};
    \node[lit] (a3) at (1.5, -4) {$A_3$};

    \draw[edge] (root) -- (and1);
    \draw[edge] (root) -- (and2);

    \draw[edge] (and1) -- (a1);
    \draw[edge] (and1) -- (a2);
    \draw[edge] (and1) -- (a3);

    \draw[edge] (and2) -- (na1);
    \draw[edge] (and2) -- (or2);

    \draw[edge] (or2) -- (and3);
    \draw[edge] (or2) -- (na2);

    \draw[edge] (and3) -- (a2);
    \draw[edge] (and3) -- (a3);
\end{tikzpicture}
  \end{subfigure}
  \begin{subfigure}[T]{0.35\textwidth}
    \centering
    \begin{tikzpicture}[scale=0.7,transform shape]
    \node[orNode] (root) at (0, 0) {$\lor$};

    \node[andNode] (and1) at (-1, -1) {$\land$};
    \node[andNode] (and2) at (1, -1) {$\land$};

    \node[lit] (a1)  at (-2, -2) {$A_1$};
    \node[lit] (na1) at (0.5, -2) {$\neg A_1$};
    \node[orNode]  (or2) at (2, -2) {$\lor$};

    \node[andNode]  (and3) at (1.1, -3) {$\land$};
    \node[andNode] (and4)  at (3, -3) {$\land$};

    \node[orNode] (or3) at (2.5, -4) {$\lor$};
    \node[lit] (na2)  at (3.5, -4) {$\neg A_2$};

    \node[lit] (a2) at (-1, -5) {$A_2$};
    \node[lit] (a3) at (2, -5) {$A_3$};
    \node[lit] (na3) at (3, -5) {$\neg A_3$};

    \draw[edge] (root) -- (and1);
    \draw[edge] (root) -- (and2);

    \draw[edge] (and1) -- (a1);
    \draw[edge] (and1) -- (a2);
    \draw[edge] (and1) -- (a3);

    \draw[edge] (and2) -- (na1);
    \draw[edge] (and2) -- (or2);

    \draw[edge] (or2) -- (and3);
    \draw[edge] (or2) -- (and4);

    \draw[edge] (and3) -- (a2);
    \draw[edge] (and3) -- (a3);

    \draw[edge] (and4) -- (na2);
    \draw[edge] (and4) -- (or3);

    \draw[edge] (or3) -- (a3);
    \draw[edge] (or3) -- (na3);

\end{tikzpicture}
  \end{subfigure}

  \caption{
    Different NNF representations for the formula $(\neg A_1\vee A_2)\wedge(\neg A_2\vee A_3)$: a NNF (left), a \dDNNF{} (center), and a \sdDNNF{} (right).
  }%
  \label{fig:ddnnfprop}
\end{figure}
We recall some definitions and results from~\cite{darwicheKnowledgeCompilationMap2002}.
An NNF formula is \emph{decomposable} if every $\wedge$-node $\psi_1\wedge\ldots\wedge\psi_k$ is such that $\atoms{\psi_i}\cap\atoms{\psi_j}=\emptyset$ for every $i\neq j$.
An NNF formula is \emph{deterministic} if every $\vee$-node $\psi_1\vee\ldots\vee\psi_k$ is such that $\psi_i\wedge\psi_j\models\bot$ for every $i\neq j$.
An NNF formula is \emph{smooth} if every $\vee$-node $\psi_1\vee\ldots\vee\psi_k$ is such that $\atoms{\psi_i}=\atoms{\psi_j}$ for every $i\neq j$.

\emph{\dDNNF{}} is the class of NNF formulas that are both decomposable and deterministic.
\emph{\sdDNNF{}} is the class of NNF formulas that are decomposable, deterministic, and smooth.
  Every \dDNNF{} $\vi{}$ can be converted into an equivalent $\sdDNNF{}$
  in polynomial time, by applying the transformation \Smooth{\vi{}},
  which replaces bottom-up every non-smooth $\vee$-node
  $\psi\defas\psi_1\vee\ldots\vee\psi_K$ with
  $\bigvee_{i=1}^K{\psi_i\wedge\bigwedge_{\alpha\in\atoms{\psi}\setminus\atoms{\psi_i}}(\alpha\vee\neg\alpha)}$
  (Lemma A.1 in~\cite{darwicheKnowledgeCompilationMap2002}).
\cref{fig:ddnnfprop} shows a graphical representation of a NNF formula, and equivalent \dDNNF{} and \sdDNNF{} formulas.
A \dDNNF{} formula $\vi{}$ allows for a wide range of queries to be computed in polynomial time, including:
\begin{itemize}
  \item \emph{[CO] consistency}: decide whether $\vi{}$ is satisfiable;
  \item \emph{[VA] validity}: decide whether $\vi{}$ is valid;
  \item \emph{[CE] clausal entailment}: decide whether $\vi{}\models C$ for a clause $C$;
  \item \emph{[IM] implicant}: decide whether $\gamma\models\vi{}$ for a cube $\gamma$;
  \item \emph{[CT] model counting}: compute the number of total truth
    assignments satisfying $\vi{}$; %
  \item \emph{[ME] model enumeration}: enumerate all total truth
    assignments satisfying $\vi{}$. %
\end{itemize}

\dDNNF{}s have a number of subclasses, which are obtained by imposing further structural properties on their DAG representation~\cite{darwicheKnowledgeCompilationMap2002}. 
Notable ones are \obdds{}~\cite{bryantGraphBasedAlgorithmsBoolean1986} and \sdds{}~\cite{darwicheSDDNewCanonical2011}.

\obdds~\cite{bryantGraphBasedAlgorithmsBoolean1986} are \dDNNFs{} where the root node is a decision node and a total order ``$<$'' on the atoms is imposed.
A decision node is either a constant $\top, \bot$, or a $\vee$-node
having the form $(A \wedge \residual{\vi{}}{A}) \vee (\neg A \wedge
\residual{\vi{}}{\neg A})$, where $A$ is an atom, and
$\residual{\vi{}}{A},\residual{\vi{}}{\neg A}$ are decision nodes. 
In every path from the root to a leaf, each atom is tested only once, following the order ``$<$''. Given a variable ordering, \obdds{} are \emph{canonical}.

\sdds~\cite{darwicheSDDNewCanonical2011} generalize \obdds{} by branching on sentences rather than atoms. An \sdd{} decomposes according to a \emph{vtree} $v$ ---a rooted binary tree whose leaves correspond to atoms. An \sdd{} respecting $v$ is either: a constant $\top, \bot$; a literal if $v$ is a leaf; or a decomposition $\bigvee_{i=1}^n (\vi{i} \wedge \psi_i)$ if $v$ is internal. In the latter case, the \emph{primes} $\vi{i}$ and \emph{subs} $\psi_i$ are \sdds{} respecting the left and right subtrees of $v$, respectively. Moreover, the primes must form a \emph{partition}, i.e., they are consistent, mutually exclusive, and their disjunction is valid. Given a fixed vtree, \sdds{} are \emph{canonical}.

Beyond queries supported by general \dDNNFs{}, given two \obdds{} \resp{\sdds{}} $\vi{}$ and $\psi$ on the same variable ordering \resp{vtree}, the following queries %
are computable
in polynomial time:
\begin{itemize}
  \item \emph{[EQ] equivalence}: decide whether $\vi{}\equiv\psi$;
  \item \emph{[SE] sentential entailment}: decide whether $\vi{}\models\psi$.
\end{itemize}

{Any formula can be converted into an equivalent \dDNNF{},
  though the resulting DAG might be exponentially larger than the input
  formula in the worst case.}
\cref{alg:dDNNFcompile} shows  (a much simplified version of) a
procedure for \dDNNF{} compilation.
  The algorithm assumes the input formula $\vi{}$ to be in CNF.~%
It is first invoked as $\dDNNFcompile(\vi{},\top)$ and it works
recursively, as a classic DPLL-style AllSAT enumeration procedure,
except for the {\em partitioning step} (lines~\ref{alg:dDNNFcompile:partition1}-\ref{alg:dDNNFcompile:partition2}):
if the current residual formula \vimu{} can be partitioned into the conjunction
of $k>1$ residual formulas $\vimu{1},\ldots,\vimu{k}$ which do not share atoms, then
each $\vimu{i}$ is \dDNNF{}-encoded independently and the resulting
\dDNNFs{} are conjoined. 
(Notice that in \cref{alg:dDNNFcompile} the ``$\wedge$'' and
``$\vee$'' node
constructors are implicitly assumed to simplify and propagate the $\top$ and $\bot$
constants:
$\psi\wedge\top\Rightarrow\psi$,
$\psi\wedge\bot\Rightarrow\bot$,
$\psi\vee\top\Rightarrow\top$,
$\psi\vee\bot\Rightarrow\psi$.)
Actual compilers, e.g.,~\cite{lagniezImprovedDecisionDNNFCompiler2017}, use more sophisticated enumeration schemas (e.g.\ CDCL-style)
and, most importantly, adopt forms of {\em component caching}: the \dDNNF{} encodings
of all residuals \vimu{} are cached and reused through the search, so
that multiple instances of the same  \dDNNF{} subformula are shared, yielding
a DAG structure rather than a tree.

\begin{algorithm}[t]\caption{\dDNNFcompile($\vimu{},\mu$)
    \Comment{%
      Recursive, first invoked as
      $\dDNNFcompile(\vi{},\top)$}\label{alg:dDNNFcompile}}
  \Input $\vimu{}$: residual of input formula $\vi{}$ wrt.\ current truth assignment $\mu$  (initially $\vi{}$);\\
  \phantom{\Input}$\mu$: \ \ \ current truth assignment to \atoms{\vi{}}
  (initially $\top$)
\\\Output \dDNNFof{\vimu{}} 
\begin{algorithmic}[1]
  \If {($\vimu{}=\top$)} {\Return $\top$} \EndIf
  \If {($\vimu{}=\bot$)} {\Return $\bot$}\EndIf
  \If {($(l)$ is a unit clause in \vimu{})} {\Return
    $(l\wedge\dDNNFcompile(\residual{\vi{}}{\mu\wedge l},\mu\wedge l))$}\EndIf
  \State $\set{\vimu{1},\ldots,\vimu{k}}\gets\Partition(\vimu{})$
  \label{alg:dDNNFcompile:partition1}
  \Comment{$\vimu{}\!=\!\bigwedge_i\vimu{i}$, s.t.\ $\forall ij,\atoms{\vimu{i}}\cap\atoms{\vimu{j}}\!=\!\emptyset$}
  \If {$(k>1)$} 
  \Return $\bigwedge_{i=1}^{k}\dDNNFcompile(\vimu{i},\mu)$
  \EndIf
  \label{alg:dDNNFcompile:partition2}
  \State $l\gets\Select(\vimu{})$ \Comment{Select one literal $l$ on
    \atoms{\vimu{}}}
  \State \Return $(l\wedge\dDNNFcompile(\residual{\vi{}}{\mu\wedge l},\mu\wedge l)) \vee
  (\neg l \wedge \dDNNFcompile(\residual{\vi{}}{\mu\wedge \neg l},\mu\wedge\neg l)) $
\end{algorithmic}
\end{algorithm}  

\subparagraph*{\T-lemma enumeration.}
Our framework requires enumerating a set of \T-lemmas ruling out all \T-inconsistent total truth assignments propositionally satisfying a given \T-formula~$\vi{}$ (see \sref{sec:framework}).
To this end, we briefly recap the \T-lemma enumeration techniques of~\cite{micheluttiCanonicalDecisionDiagrams2024,civiniEagerEncodingsTheoryAgnostic2026}.
In~\cite{micheluttiCanonicalDecisionDiagrams2024}, we proposed a technique based on an enumeration of all the \T-consistent \emph{total} truth assignments satisfying $\vi{}$, which, as a by-product, produces a suitable set of \T-lemmas.
In~\cite{civiniEagerEncodingsTheoryAgnostic2026}, we introduced three more
orthogonal and composable techniques that drastically improve
scalability and that are highly parallelizable.

First, we proposed a \emph{cube-and-conquer} enumeration %
algorithm, which
enumerates a set of \T-consistent \emph{partial} truth assignments
$\mu_i$ which do not falsify $\vi{}$,  %
and then performs an independent %
run on each subproblem~$\vi{}\wedge\mu_i$. This partitions the search
space into independent, smaller and simpler subproblems, which (i) are easier to solve
and (ii) can be computed independently in parallel.

Second, we proved that the \emph{projected}  enumeration onto theory-atoms only is sufficient to derive a complete set of \T-lemmas, while being typically much faster in practice.

Finally, we showed that, if the atoms can be \emph{partitioned} into
sets $\allalpha_1,\ldots, \allalpha_k$ that are disjoint wrt.\
variables and uninterpreted functions, then we can perform independent runs projected on each set~$\allalpha_i$.
The resulting sets of \T-lemmas can then be combined to obtain a complete set for~$\vi{}$. 
This strategy allows us to drastically reduce the overall enumeration cost.

\section{d-DNNF{} Modulo Theories: An Analysis of the Problem}%
\label{sec:problem}

As stated in \sref{sec:intro}, our general goal is to encode every
\T-formula \vi{} into a \T-equivalent \dDNNF{}
\T-formula s.t.\ the following checks can be performed in polytime
by means of their corresponding propositional checks:
{\em \T-consistency  (CO)},
{\em \T-validity  (VA)},
{\em clausal \T-entailment (CE)},
{\em \T-implicant (IM)},
{\em \T-consistent assignment counting (CT)}   and
{\em \T-consistent assignment enumeration (ME)}.
Also, we want that  {\em \T-equivalence (EQ)} and {\em sentential
\T-entailment (SE)} checks to be polytime for those forms of \dDNNF{}
\T-formulas such that EQ and SE are
polytime for their propositional counterpart  (e.g., SDDs and OBDDs);
we also wish the latter to be \T-canonical. 
Importantly, the whole process should be {\em theory-agnostic}.

\begin{remark}%
\label{remark:allsmtcountsmt}
Here CT and ME %
refer respectively to counting and
enumerating all \T-consistent total truth 
assignments on a given atom set
  $\allalpha\supseteq\atoms{\vi{}}$ which propositionally satisfy the
input \T-formula $\vi{}$, and we also refer to them as ``\SharpSMT{}''
and ``\AllSMT{}'' (also ``\SharpSAT{}'' and ``\AllSAT{}'' for the
propositional case).
\footnote{%
In this paper we prefer ``\SharpSMT{}'' and
    ``\AllSMT{}'' to
    ``\#SMT{}$_{\allalpha}$'' and ``AllSMT$_{\allalpha}$'' respectively to avoid confusion,
    because we are enumerating \T-satisfiable satisfying assignments
    rather than theory-specific FOL models. 
    In particular, ``\#SMT{}'' or ``model counting
  modulo theories'' and, more in general,
  ``extension of model counting to SMT'' have been proposed with
  very different meanings in the literature (e.g.,
  \cite{maVolumeComputationBoolean2009,phanModelCountingModulo2015,chistikovApproximateCountingSMT2017,morettin-wmi-aij19,derkinderenTopDownKnowledgeCompilation2023,shawEfficientVolumeComputation2025,shawApproximateSMTCounting2025}).
}
We stress the fact that, regardless of the
\dDNNF{} format, CT and ME for \T-formulas 
are much trickier than
their propositional counterpart problems.
In fact, in the propositional case, %
producing a {\em partial} assignment $\mup$ 
s.t.\ $\mup\bmodels\vip{}$ prevents from producing all the $2^k$ total assignments $\etap_{i}$
extending $\mup$, $k$ being the number of unassigned atoms in $\mup$.
Unfortunately, this is not the case with \T-formulas, %
where producing a {\em  partial} \T-consistent{} assignment $\mu$ s.t.\ $\mu\bmodels\vi{}$
does {\em not} prevent from producing all the $2^k$ total assignments $\eta_{i}$
extending $\mu$, {\em because not all of them are necessarily \T-consistent}.
Therefore, being able to reduce CT and ME to their propositional counterparts may be
very valuable (see  \sref{sec:experiments}).
\end{remark}  

Our idea is thus to move all the effort of theory reasoning to the
compilation phase, and encode all the necessary theory-related information %
directly into the \dDNNF{} \T-formulas, so that the latter can be queried by
means of standard polytime propositional queries.
Notice that, in the propositional case, there is a form of duality
between the queries 
CO, CE, CT, and ME, which search   for {\em satisfying} truth assignments,
  and  {VA and IM}, which search   for {\em falsifying} ones.
  This duality suggests that we need two distinct encodings for the two groups of queries. 
  (For EQ and SE, this works either way, because
$\vi{1}\equiv\vi{2}$ iff $\neg\vi{1}\equiv\neg\vi{2}$ and
$\vi{1}\models\vi{2}$ iff $\neg\vi{2}\models\neg\vi{1}$.)

We start from the observation that, if we encoded a \T-formula $\vi{}$ simply into
$\dDNNFof{\vi{}}$, then the latter could be propositionally satisfied
\resp{falsified} by many \T-inconsistent truth assignments,
which would prevent CO, CE, CT, ME, EQ, and SE
\resp{VA, IM, EQ, and SE} queries to be correctly implemented by means
of their propositional
counterpart queries.~%
Therefore, we must get rid of such truth assignments, in the two
respective cases. %
Our idea is to do this  in the compilation phase, by adding proper \T-lemmas
ruling out these assignments.

Unfortunately, the traditional ``lazy'' lemmas-on-demand
approach, which is widely adopted in standard SMT %
tools
(see e.g.~%
\cite{barrettSatisfiabilityModuloTheories2021,lahiriSMTTechniquesFast2006}),
does not seem to be applicable here.
Consider the case we want to encode $\vi{}$ into a \T-equivalent
\dDNNF{} \T-formula with no satisfying \T-inconsistent 
truth assignments, so as to allow polytime CO, CE, CT, and ME
queries.
Consider \cref{alg:dDNNFcompile} (or its more-sophisticated variants).
The natural candidate way to implement a lemma-on-demand SMT extension of this
procedure and   its variants
would be to
apply a \T-consistency check on the current truth assignment
$\mu$ as soon as a new \T-literal is added to it, and to backtrack when this check fails, adding a
\T-lemma ruling  out $\mu$.
The problem with this approach would be in the partitioning step which generates the $\wedge$-nodes (lines~\ref{alg:dDNNFcompile:partition1}-\ref{alg:dDNNFcompile:partition2} in
\cref{alg:dDNNFcompile}).
In fact, even if $\mu$ is \T-consistent and $\vimu{1},\ldots,\vimu{k}$ do
not share atoms, the $\vimu{i}$'s cannot be safely \dDNNF{}-ized
independently, because their \T-atoms may be implicitly linked by some
\T-lemma, so that $\bigwedge_{i=1}^k\dDNNFcompile(\vimu{i},\mu)$
(line~\ref{alg:dDNNFcompile:partition2}  in
\cref{alg:dDNNFcompile}) could be propositionally
satisfied by some \T-inconsistent truth assignment.

Therefore, to cope with this fact, we need adding {\em a priori} an
ad-hoc set of \T-lemmas ruling out those undesired truth assignments. 

\begin{figure}
  \centering
  \begin{subfigure}[T]{0.45\textwidth}
    \centering
    \newcommand{\atomone}{\ensuremath{(x_1\le 0)}}
\newcommand{\atomtwo}{\ensuremath{(x_2\le 0)}}
\newcommand{\atomthree}{\ensuremath{(x_1\ge 1)}}
\newcommand{\atomfour}{\ensuremath{(x_2\ge 1)}}
\begin{tikzpicture}[scale=0.7,transform shape]
    \node[andNode] (root) at (0, 0) {$\land$};

    \node[orNode] (or1) at (-2, -1) {$\lor$};
    \node[orNode] (or2) at (2, -1) {$\lor$};

    \node[lit] (a1)  at (-3, -2) {\textcolor{blue}{$\atomone$}};
    \node[andNode] (and1) at (-1, -2) {$\land$};
    \node[lit] (a3)  at (1, -2) {\textcolor{blue}{$\atomthree$}};
    \node[andNode] (and2) at (3, -2) {$\land$};

    \node[lit] (na1)  at (-2, -3) {$\neg \atomone$};
    \node[lit] (a2)  at (0, -3) {$\atomtwo$};
    \node[lit] (na3)  at (2, -3) {$\neg \atomthree$};
    \node[lit] (a4)  at (4, -3) {$\atomfour$};

    \draw[edge,blue] (root) -- (or1);
    \draw[edge,blue] (root) -- (or2);
    
    \draw[edge,blue] (or1) -- (a1);
    \draw[edge] (or1) -- (and1);

    \draw[edge] (and1) -- (na1);
    \draw[edge] (and1) -- (a2);

    \draw[edge,blue] (or2) -- (a3);
    \draw[edge] (or2) -- (and2);
    \draw[edge] (and2) -- (na3);
    \draw[edge] (and2) -- (a4);
\end{tikzpicture}
  \end{subfigure}
  \begin{subfigure}[T]{0.5\textwidth}
    \centering
    \newcommand{\atomone}{\ensuremath{(x_1\le 0)}}
\newcommand{\atomtwo}{\ensuremath{(x_2\le 0)}}
\newcommand{\atomthree}{\ensuremath{(x_1\ge 1)}}
\newcommand{\atomfour}{\ensuremath{(x_2\ge 1)}}
\begin{tikzpicture}[scale=0.7,transform shape]
    \node[orNode] (root) at (0, 0) {$\lor$};

    \node[andNode] (and1) at (-2, -1) {$\land$};
    \node[andNode] (and2) at (2, -1) {$\land$};

    \node[lit] (a1)  at (-4, -2) {$\atomone$};
    \node[lit] (na4) at (4, -2) {$\neg\atomfour$};

    \node[lit] (na2)  at (-3, -3) {$\neg\atomtwo$};
    \node[lit] (a4)  at (-1, -3) {$\atomfour$};
    \node[lit] (na1) at (1, -3) {$\neg\atomone$};
    \node[lit] (a3)  at (3, -3) {$\atomthree$};

    \node[lit] (na3)  at (-2, -4) {$\neg\atomthree$};
    \node[lit] (a2)  at (2, -4) {$\atomtwo$};

    \draw[edge] (root) -- (and1);
    \draw[edge] (root) -- (and2);
    
    \draw[edge] (and1) -- (a1);
    \draw[edge] (and1) -- (na3);
    \draw[edge] (and1) -- (a4);
    \draw[edge] (and1) -- (na2);

    \draw[edge] (and2) -- (na1);
    \draw[edge] (and2) -- (a3);
    \draw[edge] (and2) -- (na4);
    \draw[edge] (and2) -- (a2);
\end{tikzpicture}
  \end{subfigure}
  \caption{\dDNNF{} \T-formula with \T-inconsistent satisfying truth assignment (blue paths) from \cref{ex:ddnnf-bad-assignment} (left), and a \T-equivalent \dDNNF{} \T-formula with no \T-inconsistent satisfying truth assignment (right).}
  \label{fig:ddnnf-bad-assignment}
\end{figure}
\begin{example}%
\label{ex:ddnnf-bad-assignment}
Consider the \larat-formula
$\vi{}\defas(
(x_1\le 0) \vee (x_2\le 0)
)\wedge(
(x_1\ge 1) \vee (x_2\ge 1)
)$,
and assume we apply to $\vi{}$ the hypothetical SMT variant of
\cref{alg:dDNNFcompile} described above. Since the two clauses do not
share any \T-atom, the algorithm partitions the two clauses (lines~\ref{alg:dDNNFcompile:partition1}-\ref{alg:dDNNFcompile:partition2})
and recursively \dDNNF{}-izes them independently.
Since inside each clause the two \larat-atoms cannot generate conflicting
literals, 
in each independent
call no \T-inconsistent truth assignment is generated. 
Therefore the procedure may return
the \dDNNF{} \T-formula shown in \cref{fig:ddnnf-bad-assignment} (left).
As highlighted by the blue paths,
we notice that the \dDNNF{}
is propositionally
satisfied by some \T-inconsistent truth assignment, e.g.,
$\mu_1\defas(x_1\le 0)\wedge (x_2\le 0) \wedge (x_1\ge 1) \wedge \neg
(x_2\ge 1)$%
.
The problem could have been avoided by adding a priori to $\vi{}$ the \T-lemmas
$\neg (x_1\le 0) \vee\neg (x_1\ge 1)$ and
$\neg (x_2\le 0) \vee\neg (x_2\ge 1)$.
The resulting \dDNNF{} could be the one shown in \cref{fig:ddnnf-bad-assignment} (right), which is \T-equivalent to the previous one, but has no \T-inconsistent satisfying truth assignment.
\end{example}

\section{A Formal Framework for \T-formulas}%
\label{sec:framework}
\noindent
In this section, we introduce novel definitions and results %
  used in the rest of the paper.
For the sake of compactness, all the
  proofs of the theorems are deferred to Appendix~\ref{sec:proofs}.%

\subsection{Basic definitions and properties}%
\label{sec:theory-definitions}

Hereafter, unless specified otherwise, we implicitly assume that
  $\allalpha,\allalphaprime$ denote supersets of the atoms of the
  formulas $\vi{},\vi{i},\psi,\psi_{i}$ which we are considering,
and $\allA, \allA'$ denote their Boolean abstraction for
  $\vip{},\vip{i},\psip{},\psip{i}$.
We sometimes adopt the notation ``\via{}''
    to stress the fact that  \allalpha{} is the
superset of \atoms{\vi{}} whose truth assignments we are
interested in (the same applies to ``$\vip{}[\allA]$'').
 We consider {\em supersets} of \atoms{\vi{}} 
to compare formulas with different sets of atoms:
$\vi{1}$ and $\vi{2}$ can be compared only if they are both
considered as formulas on some $\allalpha{}\supseteq(\atoms{\vi{1}}\cup\atoms{\vi{2}})$.
(E.g., in order to check that the propositional formulas $A_1$, $(A_1\vee A_2)\wedge (A_1\vee\neg A_2)$ and
$(A_1\vee A_3)\wedge (A_1\vee\neg A_3)$ are all equivalent, we need
considering them as formulas %
on \set{A_1,A_2,A_3}.)

\begin{definition}[\CTTA{\vi{}},\ITTA{\vi{}} \cite{micheluttiCanonicalDecisionDiagrams2024}]%
\label{def:cttaitta}
Given a set of %
atoms \allalpha{} and a \T-formula \via{},
we denote by
$\CTTA{\vi{}}\defas\set{\eta_{i}}_{i}$ and 
$\ITTA{\vi{}}\defas\set{\rho_{j}}_{j}$ respectively
   the set of all
{\em \T-satisfiable} and that of all {\em \T-unsatisfiable} total truth
assignments on %
$\allalpha{}$ which
propositionally satisfy $\vi{}$, s.t.
\begin{eqnarray}
  \label{eq:decomposition}
  \vi{} \bequiv \bigvee_{\eta{}\in\CTTA{\vi{}}}\eta{} \vee \bigvee_{\rho{}\in\ITTA{\vi{}}}\rho{}.
\end{eqnarray}  
\end{definition}
Consequently, \CTTA{\top} and \ITTA{\top} are respectively  the sets of
 all  \T-satisfiable and \T-unsatisfiable total truth assignments which one
   can build on the atoms in \allalpha.%
   Notice that it is important to specify
which superset \allalpha{} of \atoms{\vi{}} the sets of truth assignments refer to.

\begin{example}
\label{ex:proposition1}
  Let $\allalpha{}=\set{\alpha_1,\alpha_2}\defas\set{\aone,\atwo}$
  and let $\allA{}=\set{A_1,A_2}$.
  Consider the \T-formulas
  $\vi{1}\defas \aone \vee \atwo$ and
  $\vi{2}\defas \neg\aone \iff \atwo$, so that
  $\vip{1}\defas \Aone \vee \Atwo$ and
  $\vip{2}\defas \neg\Aone \iff \Atwo$.  
It is easy to see that
$\vi{1}\not\bequiv\vi{2}$ and 
$\vi{1}\Tequiv\vi{2}$.
Then \\
$\CTTA{\vi{1}}=\CTTA{\vi{2}}= %
\set{\aone\wedge\neg\atwo,\neg\aone\wedge\atwo}$,\\ %
$\ITTA{\vi{1}}=\set{\aone\wedge\atwo}$ and
$\ITTA{\vi{2}}=\emptyset$,\\
  $\CTTA{\neg\vi{1}}=\CTTA{\neg\vi{2}}=\set{\neg\aone\wedge\neg\atwo}$,\\
  $\ITTA{\neg\vi{1}}=\emptyset$ and
  $\ITTA{\neg\vi{2}}=\set{\aone\wedge\atwo}$.\\
\end{example}

\subsection{\Treduced{} and \Textended{} \T-formulas, and their properties}

\label{sec:treduced}  

\begin{definition}[ \Treduceda \T-formula]%
\label{def:treduced}
  We say that a  \T-formula $\vi{}$ is  {\bf theory-reduced} for a
  given superset $\allalpha$ of $\atoms{\vi{}}$ (\Treduceda{}) if
  and only if
  $\ITTA{\vi{}}=\emptyset$, so that~%
  \eqref{eq:decomposition} reduces to
 \begin{eqnarray}
 \label{eq:treduction1}
    \vi{} \bequiv \bigvee_{\eta{}\in\CTTA{\vi{}}}\eta{}.
 \end{eqnarray}
\end{definition}

The following result shows that, with \Treduceda
\T-formulas, some fundamental reasoning functionalities can be reduced to
their Boolean counterpart.

\begin{theorem}%
\label{teo:treducedallqueries}
Let \via{}, \via{1} and \via{2} be   \Treduceda \T-formulas.
Let $C$ be some clause on \allalpha{}.
Then we have the following facts.
\begin{enumerate}[(a)]
\item{} [CO] $\vi{}$ is \T-satisfiable if and only if it is \B-satisfiable.
\item{} [CE] $\vi{}\Tmodels C$ if and only if $\vi{} \pmodels C$. 
\item{} [EQ] $\vi{1}\Tequiv\vi{2}$ if and only if $\vi{1}\bequiv\vi{2}$. 
\item{} [SE] $\vi{1}\Tmodels\vi{2}$ if and only if $\vi{1}\pmodels\vi{2}$. 
\item{} [CT] $\SharpSMTOf{\vi{}}=\SharpSATOf{\vip{}}$.
\item{} [ME] $\AllSMTOf{\vi{}}$ is the refinement of $\AllSATOf{\vip{}}$. 
\end{enumerate}
\end{theorem}

\label{sec:textendeda}

\begin{definition}[ \Textendeda \T-formula]
\label{def:textended}
We say that a  \T-formula $\vi{}$ is  {\bf theory-extended} for a
given superset $\allalpha$ of $\atoms{\vi{}}$ (\Textendeda{}) if
  and only if $\neg\vi{}$ is \Treduceda (that is, 
$\ITTA{\neg\vi{}}=\emptyset$%
).
\end{definition}

The following results show that, with \Textendeda
\T-formulas, some fundamental reasoning functionalities can be reduced to
their Boolean counterpart.

\begin{theorem}
\label{teo:textendedallqueries}
Let \via{}, \via{1} and \via{2} be   \Textendeda \T-formulas.
Let $\gamma$ be some cube on \allalpha{}.
Then we have the following facts.
\begin{enumerate}[(a)]
\item{} [VA] $\vi{}$ is \T-valid if and only if it is \B-valid.
\item{} [IM] $\gamma\Tmodels \vi{}$ if and only if $\gamma\pmodels\vi{}$.
\item{} [EQ] $\vi{1}\Tequiv\vi{2}$ if and only if $\vi{1}\bequiv\vi{2}$.
\item{} [SE] $\vi{1}\Tmodels\vi{2}$ if and only if $\vi{1}\bmodels\vi{2}$.
\end{enumerate}
\end{theorem}
 
\subsection{\Treduction and \Textension via \T-lemmas}
\label{sec:teo}

We present a general way to transform a generic \T-formula $\vi{}$
into a \T-equivalent one which is also \Treduceda{} \resp{\Textendeda{}}.  

\begin{definition}[\cite{micheluttiCanonicalDecisionDiagrams2024}]%
  \label{def:ruleout-nobetas}
  We say that a set \theorycl{} of \T-lemmas on \allalpha{} {\bf rules out} a set
  $\set{\rho_{1},\ldots,\rho_{M}}$ of 
  \T-unsatisfiable total truth assignments  on \allalpha{}
  if and only if,
  for every $\rho_{j}$ in the set, there exists a $\C{l}$
  s.t.\ $\rho_{j}\pmodels \neg \C{l}$, that is, if and only if
  $\bigvee_{j=1}^M\rho_{j}\wedge\bigwedge_{l=1}^K\C{l}\bequiv\bot$.
\end{definition}

 Given \allalpha{} and some \T-formula \vi{}, we denote
 as $\TLEMMAS{\vi{}}$ a function which returns a set 
 of \T-lemmas \theorycl{} on \allalpha{} which rules out \ITTA{\vi{}}.
 To this extent, $\TLEMMAS{\top}$ returns a set of \T-lemmas ruling
 out all \T-inconsistent truth assignments on $\allalpha$.

\begin{theorem}%
    \label{teo:treduce}
    Let \vi{} be a \T-formula s.t.\ $\atoms{\vi{}}\in\allalpha{}$.
    Let $\TLEMMAS{\vi{}}\defas\theorycl{}$  be a set of \T-lemmas which rules out 
    \ITTA{\vi{}}.
    Let
        \begin{eqnarray}
      \label{eq:treduce}
      \TredOf{\vi{}} &\defas&
        \vi{}\  \wedge \bigwedge_{\C{l}\in\TLEMMAS{\vi{}}} \C{l}.
     \end{eqnarray}
Then we have that:
\begin{enumerate}[(i)]
\item\label{item:treduce:tequiv} $\TredOf{\vi{}}\Tequiv\vi{}$;
\item\label{item:treduce:bmodels} $\TredOf{\vi{}}\bmodels\vi{}$;
\item\label{item:treduce:treduced} $\TredOf{\vi{}}$ is \Treduceda.
\end{enumerate}
\end{theorem}
\noindent
Notice that \TredOf{\vi{}} makes $\vi{}$ \Treduceda by moving all \T-inconsistent assignments $\rho_j$
from $\ITTA{\vi{}}$ to $\ITTA{\neg\vi{}}$ while preserving
$\CTTA{\vi{}}$ and $\CTTA{\neg\vi{}}$.

\begin{theorem}%
    \label{teo:textend}
    Let \vi{} be a \T-formula s.t.\ $\atoms{\vi{}}\in\allalpha{}$.
    Let $\TLEMMAS{\neg\vi{}}\defas\theorycl{}$  be a set of \T-lemmas which rules out 
    \ITTA{\neg\vi{}}.
    Let
        \begin{eqnarray}
      \label{eq:textend}
      \TextOf{\vi{}} &\defas&
        \vi{}\  \vee \neg(\bigwedge_{\C{l}\in\TLEMMAS{\neg\vi{}}} \C{l}).
     \end{eqnarray}
Then we have that:
\begin{enumerate}[(i)]
\item\label{item:textend:tequiv} $\TextOf{\vi{}}\Tequiv\vi{}$;
\item\label{item:textend:bmodels} $\vi{}\bmodels\TextOf{\vi{}}$;
\item\label{item:textend:textended} $\TextOf{\vi{}}$ is \Textendeda.
\end{enumerate}
\end{theorem}
\noindent
Notice that \TextOf{\vi{}} makes $\vi{}$ \Textendeda by moving all \T-inconsistent assignments $\rho_j$
from $\ITTA{\neg\vi{}}$ to $\ITTA{\vi{}}$ while preserving
$\CTTA{\vi{}}$ and $\CTTA{\neg\vi{}}$.

\begin{example}%
\label{ex:tredtext}
Let $\allalpha\defas\set{\atoma,\atomb,\atomc}$ and $\vi{}\defas\atoma \vee \atomb$.
$\ITTA{\vi{}}=\{\atoma\wedge\atomb\wedge\atomc,\;
\atoma\wedge\atomb\wedge\neg\atomc,\;
\atoma\wedge\neg\atomb\wedge\atomc\}$
and $\TLEMMAS{\vi{}}=\{\neg\atoma\vee\neg\atomb,\;\neg\atoma\vee\neg\atomc\}$. Hence,
\begin{align*}
  \TredOf{\vi{}}
    &= \big(\atoma \vee \atomb\big)\wedge\big(\neg\atoma\vee\neg\atomb\big)\wedge\big(\neg\atoma\vee\neg\atomc\big).
\end{align*}

Let $\neg\vi{}\defas\neg\atoma \wedge \neg\atomb$.
We have $\ITTA{\neg\vi{}}=\set{\neg\atoma\wedge\neg\atomb\wedge\atomc}$
and $\TLEMMAS{\neg\vi{}}=\{\atomb\vee\neg\atomc\}$. Hence,
\begin{align*}
  \TextOf{\vi{}}
    = \big(\atoma \vee \atomb\big) \vee
      \big(\neg\atomb\wedge\atomc\big)
      \bequiv \atoma \vee \atomb \vee \atomc.
\end{align*}
Notice that $\TredOf{\vi{}}$ is \Treduceda
and $\TextOf{\vi{}}$ is \Textendeda, as expected.
\end{example}

\begin{IGNOREINLONG}
\begin{remark}
\label{remark:texttred}
With \cref{def:ruleout-nobetas},
the set of \T-lemmas is not unique, so that the definition
of $\TredOf{\vi{}}$ \resp{of $\TextOf{\vi{}}$} is not unique.
Nevertheless,
due to \cref{teo:treduce} and
\cref{teo:treducedallqueries} [EQ],
\resp{to \cref{teo:textend} and
\cref{teo:textendedallqueries} [EQ]},
two different versions of
$\TredOf{\vi{}}$ \resp{$\TextOf{\vi{}}$} based on different \T-lemma
sets would produce two formulas
$\vi{1}$ and $\vi{2}$ s.t.\ $\vi{1}\bequiv\vi{2}$.
Thus, different \T-lemma sets \TLEMMAS{\vi{}} 
\resp{\TLEMMAS{\neg\vi{}}}
would produce syntactically different, though
\B-equivalent, versions of $\TredOf{\vi{}}$ \resp{of
  $\TextOf{\vi{}}$}.
\end{remark}  
\end{IGNOREINLONG}

\section{d-DNNF \T-Formulas and their Polytime Queries}%
\label{sec:nnf}
We say that a \T-formula $\psi{}$ is in \NNF{} \resp{\dDNNF{}, \sdDNNF{}, etc.} iff
its Boolean abstraction is in \NNF{}  \resp{\dDNNF, \sdDNNF, etc.}.
To this extent, we 
extend w.l.o.g.\ the notion of smoothness to all atoms in \allalpha{}, that is,
we say that \via{} is \AlphaSmooth{} if $\atoms{\vi{}}=\allalpha$ and
\vi{} is smooth. The function $\Smooth{}$ is extended accordingly.
\Cref{fig:ddnnftheory} shows some examples of \NNF{} \T-formulas.

\begin{figure}
  \centering
  \begin{subfigure}[T]{0.178\textwidth}
    \centering
    \begin{tikzpicture}[scale=0.7,transform shape]
    \node[orNode] (root) at (0, 0) {$\lor$};

    \node[lit] (a1) at (-1, -1) {$\atoma$};
    \node[lit] (a2) at (1, -1) {$\atomb$};

    \draw[edge] (root) -- (a1);
    \draw[edge] (root) -- (a2);
\end{tikzpicture}
  \end{subfigure}
  \begin{subfigure}[T]{0.228\textwidth}
    \centering
    \begin{tikzpicture}[scale=0.7,transform shape]
    \node[orNode] (root) at (0, 0) {$\vee$};

    \node[lit] (a1) at (1, -1) {$\atoma$};
    \node[andNode] (and1) at (-1, -1) {$\wedge$};

    \node[lit] (na1) at (0, -2) {$\neg\atoma$};
    \node[lit] (a2) at (-2, -2) {$\atomb$};

    \draw[edge] (root) -- (a1);
    \draw[edge] (root) -- (and1);

    \draw[edge] (and1) -- (na1);
    \draw[edge] (and1) -- (a2);
\end{tikzpicture}
  \end{subfigure}
  \begin{subfigure}[T]{0.33\textwidth}
    \centering
    \begin{tikzpicture}[scale=0.7,transform shape]
    \node[orNode] (root) at (0, 0) {$\vee$};

    \node[andNode] (and1) at (-1.5, -1) {$\wedge$};
    \node[andNode] (and2) at (1.5, -1) {$\wedge$};

    \node[lit] (a1)  at (-2.5, -2) {$\atoma$};
    \node[lit] (na2)  at (-1.5, -3) {$\neg\atomb$};
    \node[lit] (na3)  at (-0.5, -2) {$\neg\atomc$};

    \node[lit] (na1) at (0.5, -3) {$\neg\atoma$};
    \node[lit] (a2) at (2.5, -3) {$\atomb$};

    \draw[edge] (root) -- (and1);
    \draw[edge] (root) -- (and2);

    \draw[edge] (and1) -- (a1);
    \draw[edge] (and1) -- (na2);
    \draw[edge] (and1) -- (na3);

    \draw[edge] (and2) -- (na1);
    \draw[edge] (and2) -- (a2);
\end{tikzpicture}
  \end{subfigure}
  \begin{subfigure}[T]{0.24\textwidth}
    \centering
    \begin{tikzpicture}[scale=0.7,transform shape]
    \node[orNode] (root) at (0, 0) {$\vee$};

    \node[andNode] (and1) at (1, -1) {$\wedge$};
    \node[lit] (a1) at (-1, -1) {$\atoma$};

    \node[orNode] (or1) at (0, -2) {$\vee$};
    \node[lit] (na1) at (2, -2) {$\neg\atoma$};

    \node[andNode] (and2) at (1, -3) {$\wedge$};
    \node[lit] (a2) at (-1, -3) {$\atomb$};

    \node[lit] (na2) at (0, -4) {$\neg\atomb$};
    \node[lit] (a3) at (2, -4) {$\atomc$};

    \draw[edge] (root) -- (and1);
    \draw[edge] (root) -- (a1);
    
    \draw[edge] (and1) -- (or1);
    \draw[edge] (and1) -- (na1);

    \draw[edge] (or1) -- (and2);
    \draw[edge] (or1) -- (a2);

    \draw[edge] (and2) -- (na2);
    \draw[edge] (and2) -- (a3);
\end{tikzpicture}
  \end{subfigure}
  \caption{
    \NNF{} representations for formulas in \cref{ex:tredtext}.
    Left: the \NNF{} of $\vi{}$; 
    Center left: %
    $\dDNNFof{\vi{}}$;
    Center right: %
    \redOf{\dDNNF}{\vi{}}; Right: %
    \extOf{\dDNNF}{\vi{}}.
  }%
  \label{fig:ddnnftheory}
\end{figure}

We present some results on polytime queries for \dDNNF{} \T-formulas. We note that these results can be extended to other \NNFs{} which admit polytime queries
(e.g., DNF, CNF, MODS, PI, IP, etc.~\cite{darwicheKnowledgeCompilationMap2002}), 
though we omit them for simplicity.

\subsection{Polytime queries on d-DNNF{} \T-formulas by exploiting \Treduction}
\label{sec:operations:treduction}  

\cref{teo:treducedallqueries}
shows that, for a \Treduceda{} \T-formula \vi{},
\T-satisfiability (CO),
clause \T-entailment (CE),
\T-consistent assignment counting (CT),
\T-consistent assignment enumeration (ME),
\T-equivalence (EQ) and
\T-entailment (SE),
reduce to the
(computationally much cheaper) 
Boolean case, regardless of the theory \T considered.
This leads to the following results.   
\begin{theorem}
\label{teo:polinomialwithtreduced}
Let $\psi{}$ be a \Treduceda \dDNNF{} \T-formula, and let $C$ be
some clause on \allalpha{}.
\begin{enumerate}[(a)]
\item{\label{item:poly:co}}[CO] The \T-satisfiability of $\psi{}$ can be computed in polynomial
  time wrt.\ the size of $\psi{}$. 
\item{\label{item:poly:ce}}[CE] The \T-entailment $\psi\Tmodels C$ can be computed in polynomial
  time wrt.\ the size of $\psi{}$ and $C$. 
\item{\label{item:poly:ct}}[CT] $\SharpSMTOf{\psi{}}$  can be computed in polynomial
  time wrt.\ the size of $\psi{}$. 
\item{\label{item:poly:me}}[ME] $\AllSMTOf{\psi{}}$ can be computed in polynomial
  time wrt.\ the size of $\psi{}$ and the size of the output set. 
\end{enumerate}
\end{theorem}

\begin{theorem}
\label{teo:polinomialequivalencewithtreduced}
Let $\psi_{1}$ and $\psi_{2}$ be \Treduceda \dDNNF{}.
\begin{enumerate}[(a)]
\item{}\label{item:poly:tred:eq} [EQ] If  $\psi_{1}$ and $\psi_{2}$ are in one form s.t.\ $\psi_{1}\bequiv\psi_{2}$ can be computed in polynomial time (e.g., OBDDs and SDDs), then
$\psi_{1}\Tequiv\psi_{2}$ can be computed in polynomial time.
\item{}\label{item:poly:tred:se} [SE] If  $\psi_{1}$ and $\psi_{2}$ are in one form s.t.\ $\psi_{1}\bmodels\psi_{2}$ can be computed in polynomial time (e.g., OBDDs and SDDs), then
$\psi_{1}\Tmodels\psi_{2}$ can be computed in polynomial time.

\end{enumerate}
\end{theorem}

\subsection{Polytime queries on d-DNNF{}  \T-formulas by exploiting \Textension}
\label{sec:operations:textension}

\cref{teo:textendedallqueries}
shows that, for a \Textendeda \T-formula \vi{},
\T-validity (VA),
\T-implicant check (IM),
\T-equivalence (EQ) and
\T-entailment (SE),
reduce to the
(computationally much cheaper) 
Boolean case, regardless of the theory \T considered.
This leads to the following results.

\begin{theorem}
\label{teo:polinomialwithtextended}
Let $\psi{}$ be a \Textendeda \dDNNF{} \T-formula, and let $\gamma$ be
some cube on \allalpha{}.
\begin{enumerate}[(a)]
\item{\label{item:poly:va}}[VA] The \T-validity of $\psi{}$ can be computed in polynomial
  time wrt.\ the size of $\psi{}$. 
\item{\label{item:poly:im}}[IM] The \T-entailment $\gamma\Tmodels\psi{}$ can be computed in polynomial
  time wrt.\ the size of $\psi{}$ and $\gamma$. 
\end{enumerate}
\end{theorem}

\begin{theorem}
\label{teo:polinomialequivalencewithtextended}
Let $\psi_{1}$ and $\psi_{2}$ be \Textendeda \dDNNF{}.
\begin{enumerate}[(a)]
\item{}\label{item:poly:text:eq} [EQ] If  $\psi_{1}$ and $\psi_{2}$ are in one form s.t.\ $\psi_{1}\bequiv\psi_{2}$ can be computed in polynomial time (e.g., OBDDs and SDDs), then
$\psi_{1}\Tequiv\psi_{2}$ can be computed in polynomial time.
\item{}\label{item:poly:text:se} [SE] If  $\psi_{1}$ and $\psi_{2}$ are in one form s.t.\ $\psi_{1}\bmodels\psi_{2}$ can be computed in polynomial time  (e.g., OBDDs and SDDs), then
$\psi_{1}\Tmodels\psi_{2}$ can be computed in polynomial time.

\end{enumerate}
\end{theorem}

\subsection{Producing \Treduced{} and \Textended{} d-DNNF{} \T-formulas}%
\label{sec:producingtreducedtextended}

We denote with \stardDNNF{} an arbitrary \dDNNF{} form from~%
\cite{darwicheKnowledgeCompilationMap2002,darwicheSDDNewCanonical2011} 
(e.g., \dDNNF, \sdDNNF, \obdd{}, \sdd, etc.),
and with \stardDNNFof{.} a
corresponding encoding function into  \stardDNNF{} form.
Then, for each NNF form \stardDNNF{}, we define:
\begin{eqnarray}
  \label{eq:starnnfred}
  \redOf{\stardDNNF}{\vi{}}&\defas&\stardDNNFof{\TredOf{\vi{}}};\\
  \label{eq:starnnfext}
  \extOf{\stardDNNF}{\vi{}}&\defas&\stardDNNFof{\TextOf{\vi{}}}.
\end{eqnarray}
For instance, $\redOf{\dDNNF}{\vi{}}\defas\dDNNFof{\TredOf{\vi{}}}$ and
$\extOf{\sdd}{\vi{}}\defas\sddof{\TextOf{\vi{}}}$.

Since each encoding function \stardDNNFof{.} is \B-equivalence preserving, then by
\cref{teo:treduce,,teo:textend},
all \redOf{\stardDNNF}{\vi{}}'s are \Treduceda and all
 \extOf{\stardDNNF}{\vi{}}'s are \Textendeda. 
\cref{fig:ddnnftheory} (center right and right) shows examples of \dDNNF{} \T-formulas from
\cref{ex:tredtext}, obtained by applying the  above definitions.

We can compute $\redOf{\stardDNNF}{\vi{}}$ \resp{$\extOf{\stardDNNF}{\vi{}}$} as shown in \cref{alg:treducednnfcompiler} \resp{\cref{alg:textendednnfcompiler}}. 
The algorithm takes as input a \T-formula \vi{}, and uses a theory-lemma enumerator \LemmaEnumerator{} (e.g.,~\cite{civiniEagerEncodingsTheoryAgnostic2026}), and a Boolean \stardDNNF{} compiler \BoolNNFCompiler{} (e.g., \cref{alg:dDNNFcompile} or more advanced compilers~\cite{somenzi2009cudd,lagniezImprovedDecisionDNNFCompiler2017}).
The algorithm first computes the set of theory lemmas \theorycl{}
on \allalpha{} by calling \LemmaEnumerator{} on \vi{} \resp{on $\neg\vi{}$}.~%
Then, it builds the \Treduceda{} \resp{\Textendeda{}} version of \vi{}
as in \eqref{eq:treduce} \resp{\eqref{eq:textend}.)} 
Finally, it compiles the Boolean abstraction of the resulting formula
into \stardDNNF{} and returns the corresponding \T-formula.

\begin{figure}
\begin{minipage}[t]{0.48\textwidth}
  \begin{algorithm}[H]
    \caption{\TredNNFCompilerOf{\vi{}}}%
    \label{alg:treducednnfcompiler}
    \Input{} $\vi{}$: a \T-formula over atoms \allalpha{}\\
    \Output{} \redOf{\stardDNNF}{\vi{}}
    \begin{algorithmic}[1]
      \State $\theorycl{} \gets \LemmaEnumerator(\vi{})$\label{line:lemmaenumerator-tred}
      \State $\TredOf{\vi{}} \gets \vi{}\wedge\left(\bigwedge_{l=1}^{K}C_{l}\right)$ \Comment{see~\eqref{eq:treduce}}\label{line:tred}
      \State $\psip{} \gets \BoolNNFCompiler(\foltoprop(\TredOf{\vi{}}))$\label{line:boolnnfcompiler-tred}
      \State \Return $\proptofol(\psip{})$\label{line:proptofol-tred}
    \end{algorithmic}
  \end{algorithm}
\end{minipage}
\hfill
\begin{minipage}[t]{0.48\textwidth}
  \begin{algorithm}[H]
    \caption{\TextNNFCompilerOf{\vi{}}}%
    \label{alg:textendednnfcompiler}
    \Input{} $\vi{}$: a \T-formula over atoms \allalpha{}\\
    \Output{} \extOf{\stardDNNF}{\vi{}}
    \begin{algorithmic}[1]
      \State $\theorycl{} \gets \LemmaEnumerator(\neg\vi{})$\label{line:lemmaenumerator-text}
      \State $\TextOf{\vi{}} \gets \vi{} \vee \neg\!\left(\bigwedge_{l=1}^{K}\neg C_{l}\right)$ \Comment{see~\eqref{eq:textend}}\label{line:text}
      \State $\psip{} \gets \BoolNNFCompiler(\foltoprop(\TextOf{\vi{}}))$\label{line:boolnnfcompiler-text}
      \State \Return $\proptofol(\psip{})$\label{line:proptofol-text}
    \end{algorithmic}
  \end{algorithm}
\end{minipage}
\end{figure}

\subsection{Querying \Treduced{} and \Textended{} d-DNNF{} \T-formulas}%
\label{sec:queryingtreducedtextended}
 Combining the facts above with \cref{teo:polinomialwithtreduced,,teo:polinomialwithtextended} we have the following facts:
 \begin{itemize}
 \item every $\redOf{\stardDNNF}{\vi{}}$ %
  verifies the hypotheses of
   \cref{teo:polinomialwithtreduced}, so that \T-satisfiability (CO), clause
   \T-entailment (CE), \SharpSMT{} (CT), and \AllSMT{} (ME) with these formulas
   can be computed in polynomial time;
 \item every  $\extOf{\stardDNNF}{\vi{}}$ %
  verifies the hypotheses of
   \cref{teo:polinomialwithtextended}, so that \T-validity (VA) and
   \T-implicant check (IM) with these formulas can be computed in polynomial time;
 \item every  $\extOf{\stardDNNF}{\vi{}}$ \resp{$\redOf{\stardDNNF}{\vi{}}$}
   s.t.\ \stardDNNF{} is one of OBDD %
   and SDD,
   verifies the hypotheses of
   \cref{teo:polinomialequivalencewithtreduced} \cref{item:poly:tred:eq} 
   \resp{\cref{teo:polinomialequivalencewithtextended} \cref{item:poly:text:eq}},
   so that \T-equivalence with these formulas can be computed in polynomial time;
 \item every  $\extOf{\stardDNNF}{\vi{}}$ \resp{$\redOf{\stardDNNF}{\vi{}}$}
   s.t.\ \stardDNNF{} is one of OBDD %
   and SDD,
   verifies the hypotheses of
   \cref{teo:polinomialequivalencewithtreduced} \cref{item:poly:tred:se} 
   \resp{\cref{teo:polinomialequivalencewithtextended} \cref{item:poly:text:se}},
   so that \T-entailment with these formulas can be computed in polynomial time.
 \end{itemize}
Algorithmically, we can compute these queries by means of a Boolean \stardDNNF{}-reasoner, feeding it with the Boolean abstraction of the compiled \T-formulas and the Boolean abstraction of the query, and then interpreting the result in the theory \T.

\section{A Preliminary Experimental Evaluation}%
\label{sec:experiments}

As proof of concept, we present a preliminary experimental
evaluation of the framework described in the previous sections.
(We recall that we currently do not have any direct competitor for SMT-level \dDNNF{} querying.)
We focus on \Treduced{} \dDNNFs{}, as they allow for a wide range of queries.
We consider both general \dDNNFs{}, and the subclasses of \obdds{} and \sdds{} also considered in~\cite{micheluttiCanonicalDecisionDiagrams2024}.
We assess the following aspects:
\begin{enumerate}[(i)]
    \item the compilation time, and the contribution of each compilation step to the total time;
    \item the size of the compiled representation, compared to the size of the input \T-formula;
    \item the effectiveness and efficiency of
     the compiled representation in answering queries against SMT-based approaches. We focus on CE, and on CT under assumptions, as detailed next.
\end{enumerate}

\subparagraph*{Implementation.}
We have implemented the framework described in
\sref{sec:producingtreducedtextended} in a prototype Python tool, which is publicly available.~\footnote{\url{https://github.com/ecivini/tddnnf}}
The tool %
uses PySMT~\cite{garioPySMTSolveragnosticLibrary2015} to handle \T-formulas.
For enumeration of \T-lemmas we use our novel technique  from~\cite{civiniEagerEncodingsTheoryAgnostic2026}.~\footnote{\url{https://github.com/ecivini/tlemmas-enumeration}}
(Notably, this technique has been shown to be much faster, and to produce smaller and fewer lemmas compared to the technique used in~\cite{micheluttiCanonicalDecisionDiagrams2024}.)
We use
\dfour{}~\cite{lagniezImprovedDecisionDNNFCompiler2017}~\footnote{\url{https://github.com/crillab/d4v2}}
for d-DNNF compilation, using the %
algorithm for non-CNF formulas~\cite{derkinderenCircuitAwareDDNNFCompilation2025};
we use %
\cudd{}~\cite{somenzi2009cudd} for OBDD compilation, and
\pysdd{}~\footnote{\url{https://github.com/ML-KULeuven/PySDD}} for SDD
compilation. 

Since \obdds, \sdds, and \dDNNFs{} are obtained with different
compilers and use different representations, queries are answered
using different tools. For \obdds{} and \sdds, we use the respective
compilers' APIs. For \dDNNFs{}, we use
\ddnnife{}~\cite{sundermannReusingDDNNFsEfficient2024},~\footnote{\url{https://github.com/SoftVarE-Group/d-dnnf-reasoner}}
which supports a wide range of queries on propositional \dDNNFs,
including CO (possibly under conditioning), ME (possibly under
conditioning), and CT (possibly under conditioning and/or projected on
a subset of atoms).  CE can be reduced to CO under conditioning, and
thus it is also supported. 

The code to run the experiments is also publicly available.~\footnote{\url{https://github.com/ecivini/tddnnf-testbench}}
All the experiments have been executed on three identical machines with an Intel(R) Xeon(R) Gold 6238R @ 2.20GHz CPU and 128GB of RAM.
For theory-lemma enumeration, we use 45 parallel processes.
We set a timeout of 3600s for theory-lemma enumeration, and further 3600s for Boolean compilation.
For query answering, we set a timeout of 600s for each query.

\subparagraph*{Benchmarks.}
We consider all the problems from~\cite{micheluttiCanonicalDecisionDiagrams2024}, %
consisting of 450 synthetic non-CNF \smtlarat{} instances.
These problems are inspired by Weighted Model Integration problems from~\cite{morettin-wmi-aij19}, and were generated by randomly nesting Boolean operators up to fixed depths.
For CE queries, we generate 10 random clauses of size 1 to 3 for each problem, with literals over atoms occurring in the formula.
The second type of queries we consider is CT under assumptions, i.e.,
given a formula $\vi{}$ and a cube $\mu$, we ask for
$\SharpSMTOf{\vi{} \wedge \mu}$, %
where $\allalpha=\atoms{\vi{}}$.
We consider cubes deriving from the negated clauses used for CE. 
This task %
counts the number of ``counterexamples'' to clausal entailment, where a counterexample is a truth assignment to the atoms of the formula violating the entailment.
Hence, we only consider problems where CE does not hold.

\subparagraph*{Baselines.}
Since there is no direct competitor for SMT-level \dDNNF{} querying, 
we compare on query answering against SMT-based approaches, using the SMT solver \mathsatfive{}~\cite{mathsat5_tacas13}.

For CE, given a formula $\vi{}$ and a clause $C$, we check the \T-(un)satisfiability of $\vi{} \wedge \neg C$.
We use \mathsat{} both in non-incremental and incremental mode, using SMT under assumptions.

For CT, to the best of our knowledge, there is no SMT tool specifically
designed for counting the number of \T-satisfiable truth assignments.
This is not surprising since this task is very tricky, because
{\em it is not possible to exploit \T-consistent partial-assignment enumeration}, 
as partial assignments \emph{do not} guarantee that \emph{all} their
extensions are \T-satisfiable
(recall \cref{remark:allsmtcountsmt}).
Thus, we are forced to use the \AllSMT{} functionality of \mathsat{} to enumerate all
the total \T-satisfiable truth assignments, counting and discharging
them as they are produced.
(We reckon this could be considered as an easy-to-win comparison, but
we see no alternative competitor.)

\subparagraph*{Results on compilation time.}
\begin{figure}[t]
    \centering
    \begin{subfigure}{\textwidth}
        \centering
        \includegraphics[width=.7\textwidth]{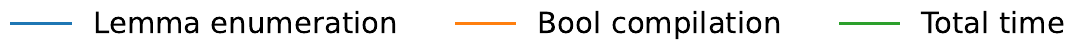}
    \end{subfigure}

    \begin{subfigure}{0.3\textwidth}
        \centering
        \includegraphics[width=.92\textwidth]{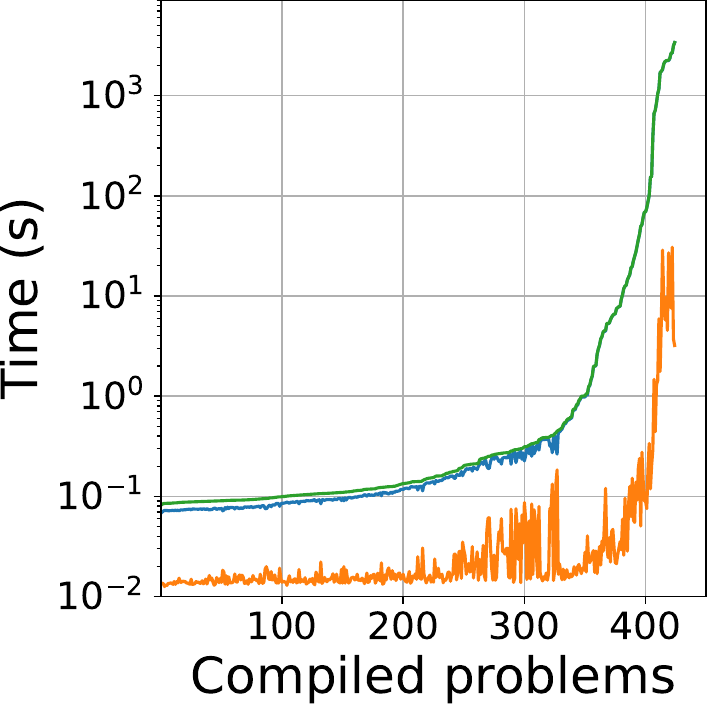}
        \caption{\centering \T-red \dDNNF}%
        \label{fig:compilation-time-ddnnf}
    \end{subfigure}
    \hfill
    \begin{subfigure}{0.3\textwidth}
        \centering
        \includegraphics[width=.92\textwidth]{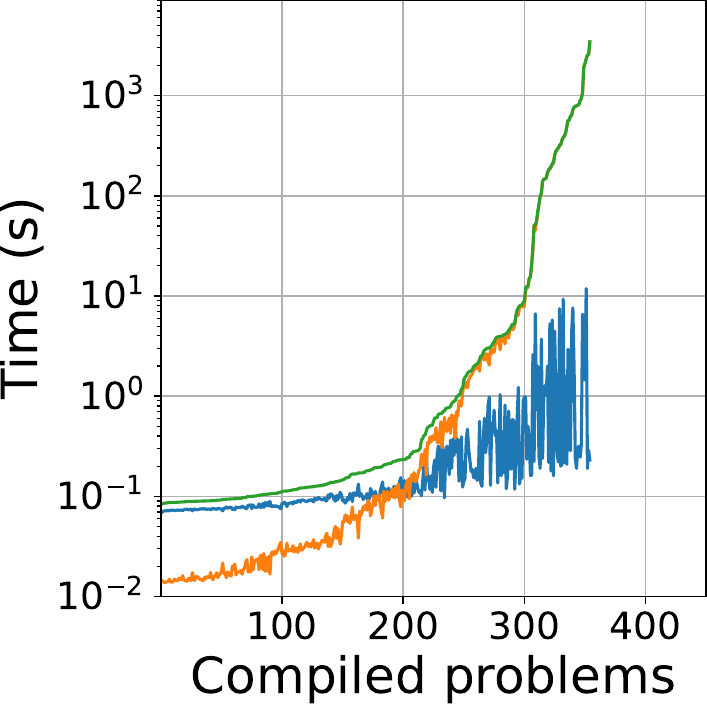}
        \caption{\centering \T-red \obdd}%
        \label{fig:compilation-time-tbdd}
    \end{subfigure}
    \hfill
    \begin{subfigure}{0.3\textwidth}
        \centering
        \includegraphics[width=.92\textwidth]{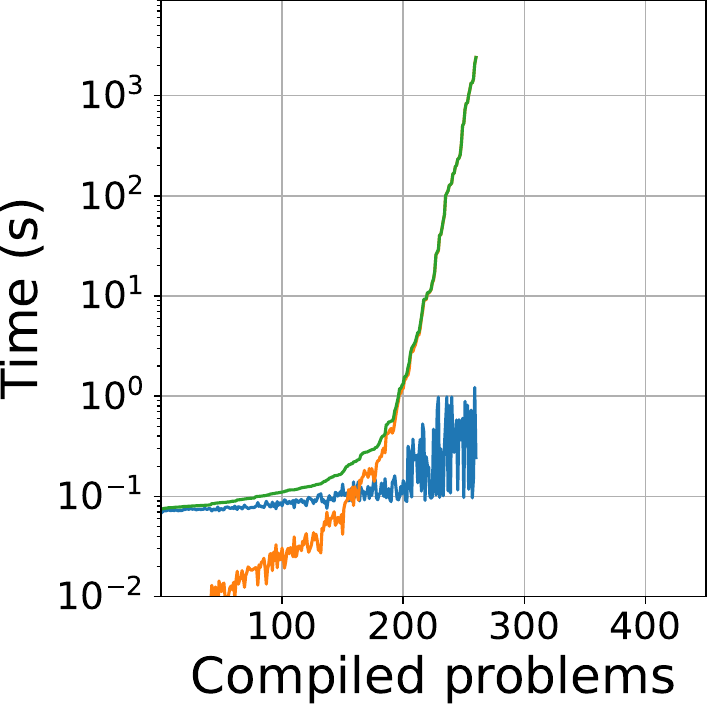}
        \caption{\centering \T-red \sdd}%
        \label{fig:compilation-time-tsdd}
    \end{subfigure}

    \begin{subfigure}{\textwidth}
        \small
        \centering
        \begin{tabular}{lccc}
            \textbf{Method} & \textbf{Lemma enum. T.O.} & \textbf{Compilation T.O.} & \textbf{Tot.} \\
            \hline
            \T-red \dDNNF{}     & 26                       & 0                       & 450           \\
            \T-red \obdd{}       & 26                       & 70                      & 450           \\
            \T-red \sdd{}        & 26                       & 164                     & 450           \\
        \end{tabular}
    \end{subfigure}
    \caption{Compilation times for \Treduced{} \NNFs.
        Top: cactus plots showing the total compilation time for \dDNNF{} (left), \obdd{} (center), and \sdd{} (right). We also show the time for lemma enumeration (blue) and Boolean compilation (orange) for each instance. 
        For \dDNNF{}, the green and blue lines overlap; for \obdd{} and \sdd{}, instead, the green and orange lines overlap.
        We only show instances compiled within the timeout.
    Bottom: Number of timeouts for each compilation step.
    }%
    \label{fig:compilation-time}
    
    \begin{subfigure}{0.3\textwidth}
        \centering
        \includegraphics[width=\textwidth]{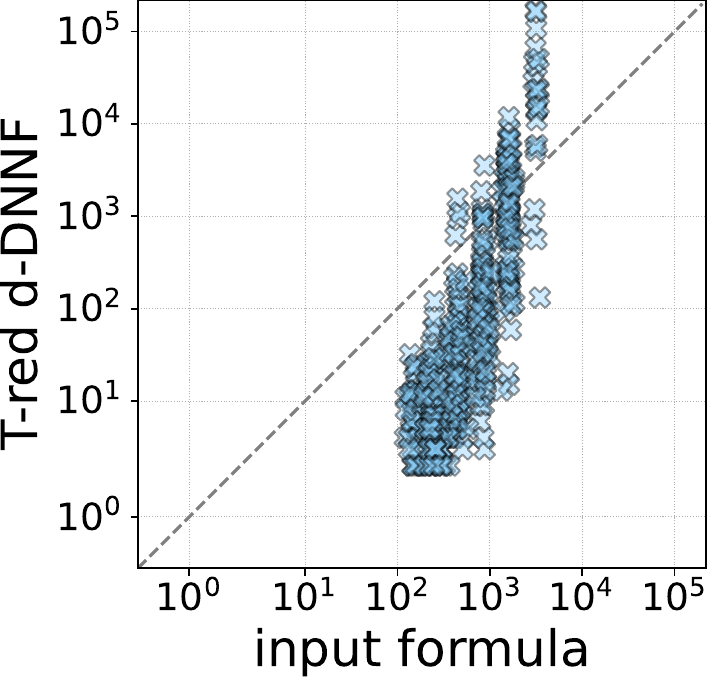}
    \end{subfigure}
    \hfill
    \begin{subfigure}{0.3\textwidth}
        \centering
        \includegraphics[width=\textwidth]{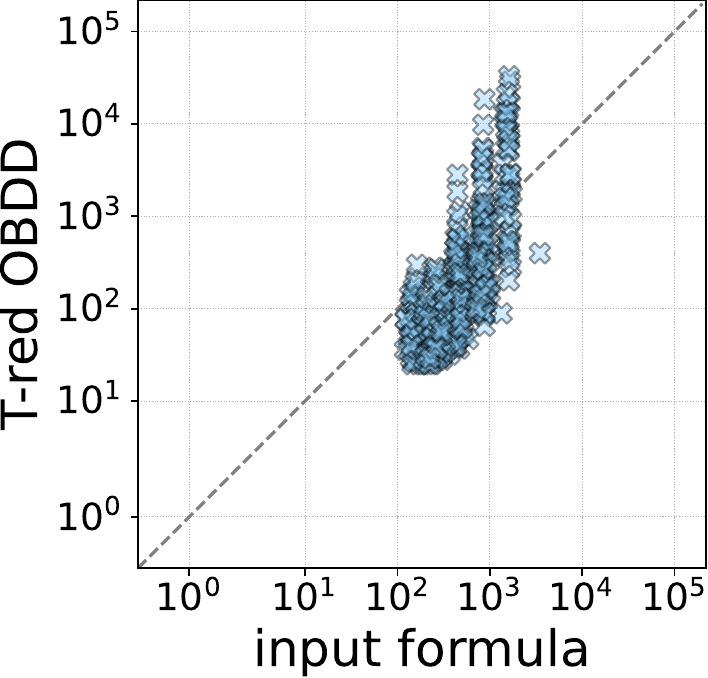}
    \end{subfigure}
    \hfill
    \begin{subfigure}{0.3\textwidth}
        \centering
        \includegraphics[width=\textwidth]{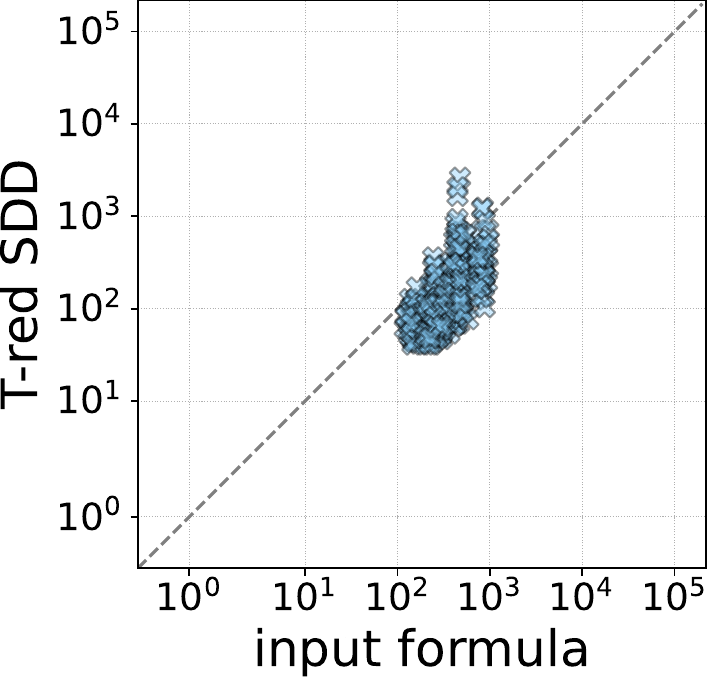}
    \end{subfigure}
    \caption{Size of the compiled \Treduced{} \dDNNF{}s (left), OBDDs (center), and SDDs (right) vs size of the input formulas, measured as the number of nodes in their DAG representation.}%
    \label{fig:size-compared-to-input}
\end{figure}

We first analyze the compilation time for \Treduced{} \NNFs{}.
In \cref{fig:compilation-time}, we show cactus plots of the compilation times for 
\dDNNFs{} (left), \obdds{} (center), and \sdds{} (right).
We show both the total time, and the contributions of the two compilation steps in \cref{alg:treducednnfcompiler}, i.e., lemma enumeration (line~\ref{line:lemmaenumerator-tred}) and Boolean compilation
(line~\ref{line:boolnnfcompiler-tred}). The table at the bottom reports the number of timeouts for each compilation step.

First, we observe that substantially more instances can be compiled
into \dDNNF{} and in significantly less time than into \obdds{} or
\sdds{}. This is not surprising since \obdds{} and \sdds{} impose
stronger structural restrictions on the compiled representation, and
are canonical under a fixed variable order or vtree, which typically
makes compilation much more expensive.

Second, for \dDNNFs{}, the compilation time is dominated by
lemma enumeration, whereas for \obdds{} and \sdds{} it is dominated by the Boolean compilation step, which also explains the smaller number of compiled instances. 
This is expected, since \dDNNF{} compilation is typically much faster than
\obdd{} and \sdd{} compilation. Also, 
 despite the optimizations
of~\cite{civiniEagerEncodingsTheoryAgnostic2026}, the enumeration
procedure is still expensive, dominating the relatively-fast
\dDNNF{} compilation. 
(The strategy
of~\cite{civiniEagerEncodingsTheoryAgnostic2026} is highly
parallelizable, however, suggesting that the enumeration time can be significantly reduced with more computational resources.)
Also more efficient strategies for conjoining \T-lemmas during
construction  could potentially  improve the compilation for \obdds{} and 
\sdds{}.

\subparagraph*{Results on size.}
Next, we analyze the size of the compiled representations.
The scatter plots in \cref{fig:size-compared-to-input} compare the size of the compiled \Treduced{} \NNFs{} with the size of the input formulas, both measured as the number of nodes in their DAG representations.
We observe that the addition of lemmas does not seem to increase 
the size of the compiled representations, which, instead, are most
often (much) smaller than the input formulas.
This is a very interesting and promising result. 
Note, however, that in these plots we only consider the instances that
could be compiled within the timeout, which are typically those with
a smaller input size. In particular, for \obdds{} and \sdds{}, we plot
much fewer points than for \dDNNFs{}.

\subparagraph*{Results on clausal-entailment queries.}
\begin{figure}[t]
    \centering
    \begin{subfigure}[c]{0.3\textwidth}
        \centering
        \includegraphics[width=\textwidth]{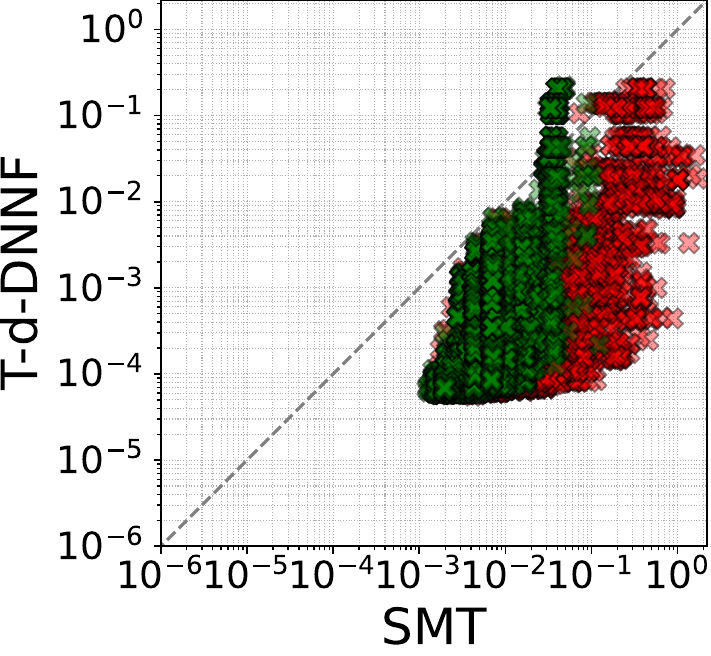}
    \end{subfigure}
    \hfill
    \begin{subfigure}[c]{0.3\textwidth}
        \centering
        \includegraphics[width=\textwidth]{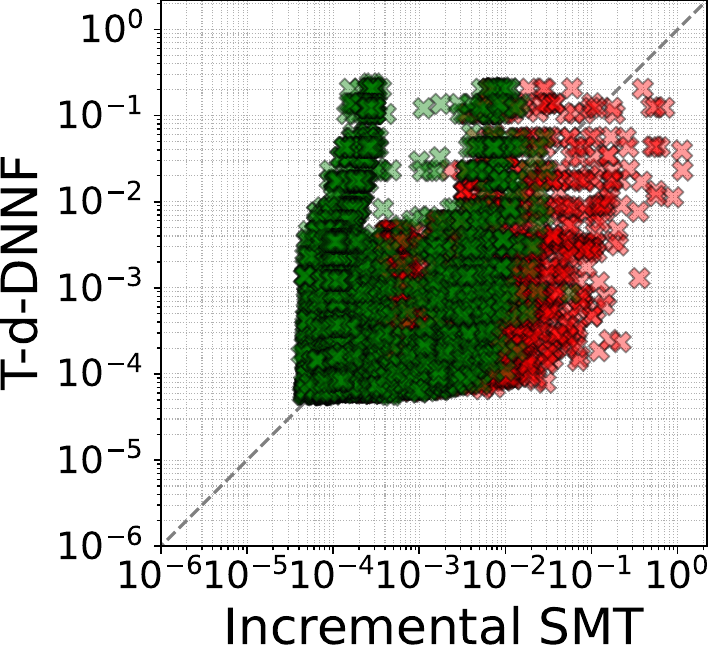}
    \end{subfigure}
    \hfill
    \begin{subfigure}[c]{0.3\textwidth}
        \centering
        \includegraphics[width=.92\textwidth]{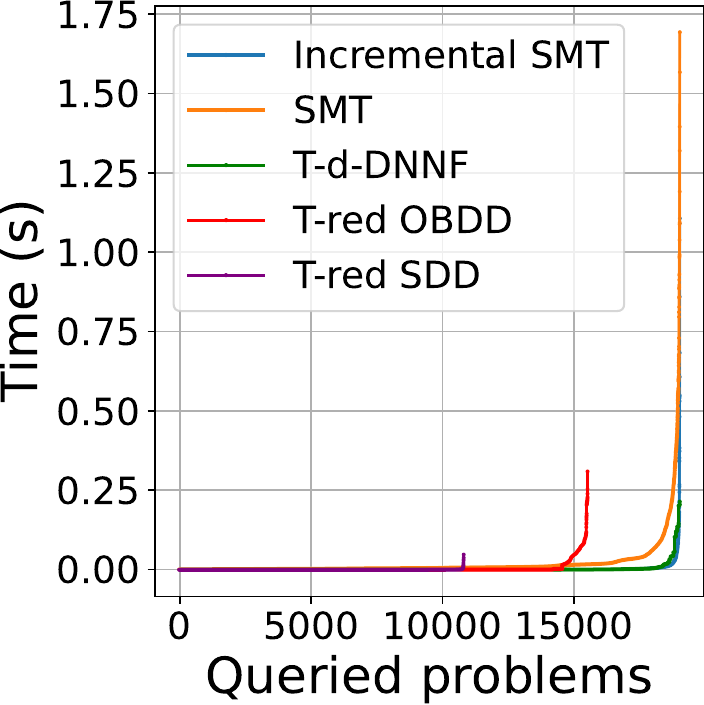}
    \end{subfigure}
    \caption{Answering times for CE queries. Left and center: scatter plots comparing the time for \Treduceda{} \dDNNFs{} vs non-incremental and incremental SMT, respectively. Green points correspond to instances where CE holds, red points correspond to instances where it does not.
    Right: comparison of all methods with a cactus plot, including \Treduceda{} \obdds{} and \sdds{}.}%
    \label{fig:entailment-queries-time}
    
    \centering
    \begin{subfigure}[c]{0.3\textwidth}
        \centering
        \includegraphics[width=\textwidth]{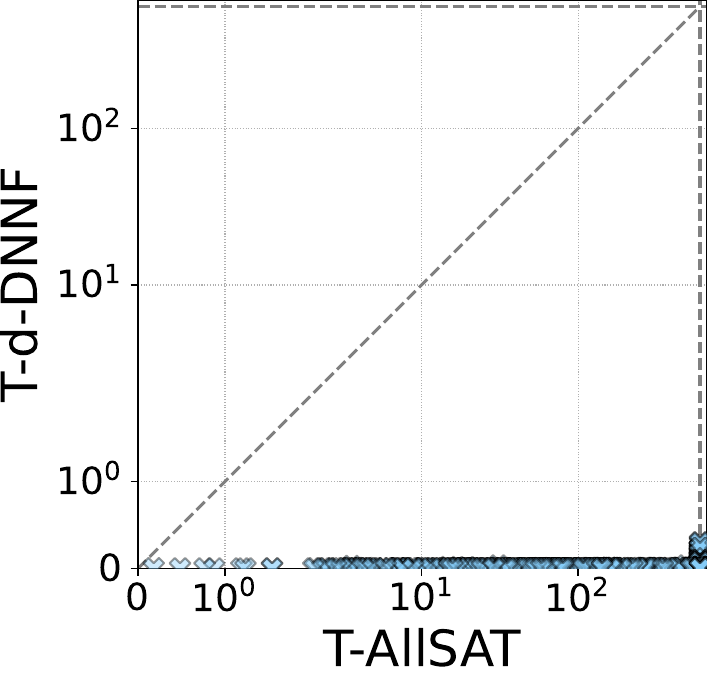}
    \end{subfigure}
    \hfill
    \begin{subfigure}[c]{0.3\textwidth}
        \centering
        \includegraphics[width=0.92\textwidth]{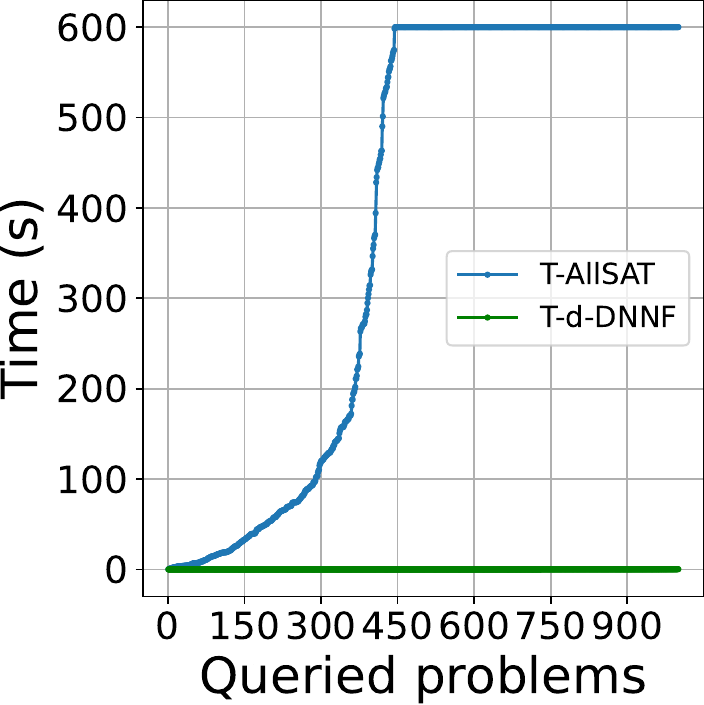}
    \end{subfigure}
    \hfill
    \begin{subfigure}[c]{0.3\textwidth}
        \small
        \centering
        \begin{tabular}{lcc}
            \hline
            \textbf{Method} & \textbf{T.O.} & \textbf{Tot.} \\
            \hline
            \AllSMT          & 556           & 1000           \\
            d-DNNF{}     & 0             & 1000           \\
            \hline
        \end{tabular}
    \end{subfigure}
    \caption{Answering times for CT queries. Left: scatter plot comparing \dDNNF{} vs \AllSMT-based counting. Center: comparison with cactus plot. Right: number of timeouts for each method.
    }%
    \label{fig:counting-queries-time}
\end{figure}
We now proceed to analyze the effectiveness and efficiency of \Treduced{} \NNFs{} for answering CE queries.
In \cref{fig:entailment-queries-time} left and center, we compare with scatter plots the time for \Treduceda{} \dDNNF{} against non-incremental and incremental SMT, respectively. On the right, a cactus plot compares all methods,
including \obdds{} and \sdds{}.

We observe that \Treduceda{} \dDNNF{}, \obdds{} and \sdds{} are all effective and efficient for answering CE queries.
In particular, \dDNNFs{} drastically outperform non-incremental SMT,
and are mostly better than incremental SMT. 
Note that SMT-solvers are highly optimized for checking satisfiability, and are thus expected to perform well on CE queries. 

Notice also that the problems and queries considered are easy and can
be answered by both methods in a fraction of a second, whereas they
were quite challenging for the experiments in \cite{micheluttiCanonicalDecisionDiagrams2024}.

\subparagraph*{Results on counting queries.}

Finally we analyze the results for  queries on \SharpSMT{} under different
sets of assumptions.
For these queries, we only consider the 100 most challenging problems according to lemma enumeration time, for a total of 1000 queries. As very few of them could be compiled into \obdds{} and \sdds{}, we only compare \dDNNFs{} against the \AllSMT-based approach.
In \cref{fig:counting-queries-time}, we compare the two approaches with a scatter plot (left), and with a cactus plot (center). On the right, we report the number of timeouts for each method.
On these queries, the advantage of \Treduced{} \dDNNF{}s over SMT-based approaches is dramatic. Indeed, %
\dDNNFs{} can solve all the problems in a fraction of a second, whereas \AllSMT{} fails to solve the big majority of the queries within the timeout.
As mentioned above, counting the number of \T-satisfiable truth
assignments with an \AllSMT{}-based approach is extremely expensive. %
On the other hand, \Treduced{} \dDNNFs{} can answer these queries in linear time with respect to the size of the compiled representation.
The results for \Treduced{} \obdds{} and \sdds{} are similar to those for \dDNNFs{}, but limited to fewer instances due to compilation timeouts.

Overall, the results show the potential of \Treduced{} \dDNNFs{} for queries intractable for standard SMT approaches, amortizing the cost of compilation over many expensive queries.

\section{Conclusions and Future Work}%
\label{sec:conclusions}

In this paper, we have introduced for the first time a formal framework for compiling SMT formulas into \T-equivalent \dDNNF{} \T-formulas, %
supporting polytime SMT queries.
The approach is theory-agnostic, independent of the target \dDNNF{} language, and can be implemented on top of existing tools. In particular, compilation relies on a lemma enumerator and a propositional compiler, both treated as black boxes.
Queries can be reduced to their propositional counterpart, and can be implemented by means of any propositional \dDNNF{} reasoner.
We have implemented our approach in a prototype tool and performed a preliminary experimental evaluation on the benchmarks from~\cite{micheluttiCanonicalDecisionDiagrams2024}.
The results confirm the feasibility and effectiveness of our approach.

Remarkably, our approach inherits from 
propositional \dDNNFs{} %
the intrinsic limitation %
that the (super)set of
atoms \allalpha{} in the formula and in the queries should be known at
compilation time. 
Relaxing this assumption is non-trivial,
and opens interesting directions for future work, including
incremental lemma-generation techniques, incremental \dDNNF{}
compilation, 
and hybrid \dDNNF{}/SMT querying on \Treduceda or \Textendeda \T-formulas. 
We also plan to investigate the applicability %
to real-world problems where fast on-line query answering is crucial.

\newpage
\bibliography{bibliography}

@incollection{barrettSatisfiabilityModuloTheories2021,
  title = {Satisfiability {{Modulo Theories}}},
  booktitle = {Handbook of {{Satisfiability}}},
  author = {Barrett, Clark and Sebastiani, Roberto and Seshia, Sanjit A. and Tinelli, Cesare},
  year = 2021,
  edition = {2},
  volume = {336},
  pages = {1267--1329},
  publisher = {IOS Press},
  doi = {10.3233/FAIA201017},
  isbn = {978-1-64368-160-3 978-1-64368-161-0},
  series = {{{FAIA}}}
}

@article{bryantBooleanSatisfiabilityTransitivity2002,
  title = {Boolean {{Satisfiability}} with {{Transitivity Constraints}}},
  author = {Bryant, Randal E. and Velev, Miroslav N.},
  year = 2002,
  journal = {ACM Trans. Comput. Logic},
  volume = {3},
  number = {4},
  pages = {604--627},
  issn = {1529-3785},
  doi = {10.1145/566385.566390},
  urldate = {2024-03-11}
}

@article{bryantGraphBasedAlgorithmsBoolean1986,
  title = {Graph-{{Based Algorithms}} for {{Boolean Function Manipulation}}},
  author = {Bryant, Randal E.},
  year = 1986,
  journal = {IEEE Trans Comput},
  volume = {C-35},
  number = {8},
  pages = {677--691},
  issn = {1557-9956},
  doi = {10.1109/TC.1986.1676819},
  urldate = {2024-03-11}
}

@inproceedings{cavadaComputingPredicateAbstractions2007,
  title = {Computing {{Predicate Abstractions}} by {{Integrating BDDs}} and {{SMT Solvers}}},
  author = {Cavada, Roberto and Cimatti, Alessandro and Franz{\'e}n, Anders and Kalyanasundaram, Krishnamani and Roveri, Marco and Shyamasundar, R.K.},
  year = 2007,
  pages = {69--76},
  doi = {10.1109/FAMCAD.2007.35},
  booktitle = {{{FMCAD}} 2007}
}

@inproceedings{chakiDecisionDiagramsLinear2009,
  title = {Decision Diagrams for Linear Arithmetic},
  author = {Chaki, Sagar and Gurfinkel, Arie and Strichman, Ofer},
  year = 2009,
  pages = {53--60},
  doi = {10.1109/FMCAD.2009.5351143},
  booktitle = {{{FMCAD}} 2009}
}

@article{chistikovApproximateCountingSMT2017,
  title = {Approximate {{Counting}} in {{SMT}} and {{Value Estimation}} for {{Probabilistic Programs}}},
  author = {Chistikov, Dmitry and Dimitrova, Rayna and Majumdar, Rupak},
  year = 2017,
  journal = {Acta Inform.},
  volume = {54},
  number = {8},
  pages = {729--764},
  issn = {1432-0525},
  doi = {10.1007/s00236-017-0297-2},
  urldate = {2023-08-24},
  langid = {english}
}

@misc{civiniEagerEncodingsTheoryAgnostic2026,
  title = {Beyond {{Eager Encodings}}: {{A Theory-Agnostic Approach}} to {{Theory-Lemma Enumeration}} in {{SMT}}},
  shorttitle = {Beyond {{Eager Encodings}}},
  author = {Civini, Emanuele and Masina, Gabriele and Spallitta, Giuseppe and Sebastiani, Roberto},
  year = 2026,
  number = {arXiv:2602.14634},
  eprint = {2602.14634},
  primaryclass = {cs},
  publisher = {arXiv},
  doi = {10.48550/arXiv.2602.14634},
  urldate = {2026-02-17},
  archiveprefix = {arXiv},
  langid = {english},
  note = {Extended version with proofs. Conference version published in Proceedings of IJCAR 2026.}
}

@article{darwicheKnowledgeCompilationMap2002,
  title = {A Knowledge Compilation Map},
  author = {Darwiche, Adnan and Marquis, Pierre},
  year = 2002,
  journal = {J Artif Intell Res},
  volume = {17},
  number = {1},
  pages = {229--264},
  issn = {1076-9757},
  doi = {10.1613/jair.989}
}

@inproceedings{darwicheSDDNewCanonical2011,
  title = {{{SDD}}: {{A New Canonical Representation}} of {{Propositional Knowledge Bases}}},
  shorttitle = {{{SDD}}},
  author = {Darwiche, Adnan},
  year = 2011,
  series = {{{IJCAI}}'11},
  volume = {2},
  pages = {819--826},
  publisher = {AAAI Press},
  doi = {10.5591/978-1-57735-516-8/IJCAI11-143},
  urldate = {2024-03-04},
  isbn = {978-1-57735-514-4},
  booktitle = {{{IJCAI}} 2011}
}

@inproceedings{derkinderenCircuitAwareDDNNFCompilation2025,
  title = {Circuit-{{Aware}} d-{{DNNF Compilation}}},
  author = {Derkinderen, Vincent and Lagniez, Jean-Marie},
  year = 2025,
  volume = {1},
  pages = {4454--4462},
  issn = {1045-0823},
  doi = {10.24963/ijcai.2025/496},
  urldate = {2025-11-05},
  langid = {english},
  booktitle = {{{IJCAI}} 2025}
}

@misc{derkinderenTopDownKnowledgeCompilation2023,
  title = {Top-{{Down Knowledge Compilation}} for {{Counting Modulo Theories}}},
  author = {Derkinderen, Vincent and Martires, Pedro Zuidberg Dos and Kolb, Samuel and Morettin, Paolo},
  year = 2023,
  number = {arXiv:2306.04541},
  eprint = {2306.04541},
  primaryclass = {cs},
  publisher = {arXiv},
  url = {http://arxiv.org/abs/2306.04541},
  urldate = {2024-03-06},
  archiveprefix = {arXiv},
  langid = {english},
  note = {Workshop on Counting and Sampling at SAT 2023}
}

@inproceedings{dosmartiresExactApproximateWeighted2019,
  title = {Exact and {{Approximate Weighted Model Integration}} with {{Probability Density Functions Using Knowledge Compilation}}},
  author = {Dos Martires, Pedro Zuidberg and Dries, Anton and De Raedt, Luc},
  year = 2019,
  volume = {33},
  pages = {7825--7833},
  doi = {10.1609/aaai.v33i01.33017825},
  urldate = {2022-12-19},
  booktitle = {{{AAAI}} 2019}
}

@inproceedings{garioPySMTSolveragnosticLibrary2015,
  title = {{{PySMT}}: A Solver-Agnostic Library for Fast Prototyping of {{SMT-based}} Algorithms},
  booktitle = {{{SMT Workshop}} 2015},
  author = {Gario, Marco and Micheli, Andrea},
  year = 2015,
  url = {https://github.com/pysmt/pysmt}
}

@article{goelBDDBasedProcedures2003,
  title = {{{BDD Based Procedures}} for a {{Theory}} of {{Equality}} with {{Uninterpreted Functions}}},
  author = {Goel, Anuj and Sajid, Khurram and Zhou, Hai and Aziz, Adnan and Singhal, Vigyan},
  year = 2003,
  journal = {Form Methods Syst Des},
  volume = {22},
  number = {3},
  pages = {205--224},
  issn = {1572-8102},
  doi = {10.1023/A:1022988809947},
  urldate = {2024-02-28},
  langid = {english}
}

@inproceedings{lagniezImprovedDecisionDNNFCompiler2017,
  title = {An {{Improved Decision-DNNF Compiler}}},
  author = {Lagniez, Jean-Marie and Marquis, Pierre},
  year = 2017,
  pages = {667--673},
  publisher = {International Joint Conferences on Artificial Intelligence Organization},
  doi = {10.24963/ijcai.2017/93},
  urldate = {2024-07-24},
  isbn = {978-0-9992411-0-3},
  langid = {english},
  booktitle = {{{IJCAI}} 2017}
}

@inproceedings{lahiriSMTTechniquesFast2006,
  title = {{{SMT Techniques}} for {{Fast Predicate Abstraction}}},
  author = {Lahiri, Shuvendu K. and Nieuwenhuis, Robert and Oliveras, Albert},
  year = 2006,
  pages = {424--437},
  doi = {10.1007/11817963_39},
  urldate = {2023-11-17},
  isbn = {978-3-540-37406-0},
  langid = {english},
  series = {{{LNCS}}},
  publisher = {Springer},
  booktitle = {{{CAV}} 2006}
}

@inproceedings{mathsat5_tacas13,
  title = {The {{MathSAT5 SMT Solver}}},
  author = {Cimatti, Alessandro and Griggio, Alberto and Schaafsma, Bastiaan Joost and Sebastiani, Roberto},
  year = 2013,
  pages = {93--107},
  publisher = {Springer},
  doi = {10.1007/978-3-642-36742-7_7},
  isbn = {978-3-642-36742-7},
  langid = {english},
  series = {{{LNCS}}},
  booktitle = {{{TACAS}} 2013}
}

@inproceedings{maVolumeComputationBoolean2009,
  title = {Volume {{Computation}} for {{Boolean Combination}} of {{Linear Arithmetic Constraints}}},
  author = {Ma, Feifei and Liu, Sheng and Zhang, Jian},
  year = 2009,
  pages = {453--468},
  publisher = {Springer},
  doi = {10.1007/978-3-642-02959-2_33},
  isbn = {978-3-642-02959-2},
  langid = {english},
  series = {{{LNCS}}},
  booktitle = {{{CADE}} 22}
}

@inproceedings{micheluttiCanonicalDecisionDiagrams2024,
  title = {Canonical {{Decision Diagrams Modulo Theories}}},
  author = {Michelutti, Massimo and Masina, Gabriele and Spallitta, Giuseppe and Sebastiani, Roberto},
  year = 2024,
  volume = {392},
  pages = {4319--4327},
  publisher = {IOS Press},
  doi = {10.3233/FAIA241007},
  urldate = {2024-11-28},
  series = {{{FAIA}}},
  booktitle = {{{ECAI}} 2024}
}

@inproceedings{mollerDifferenceDecisionDiagrams1999,
  title = {Difference {{Decision Diagrams}}},
  author = {M{\o}ller, Jesper and Lichtenberg, Jakob and Andersen, Henrik Reif and Hulgaard, Henrik},
  year = 1999,
  volume = {1683},
  pages = {111--125},
  doi = {10.1007/3-540-48168-0_9},
  urldate = {2023-09-18},
  isbn = {978-3-540-66536-6 978-3-540-48168-3},
  langid = {english},
  series = {{{LNCS}}},
  publisher = {Springer},
  booktitle = {{{CSL}} 1999}
}

@article{morettin-wmi-aij19,
  title = {Advanced {{SMT}} Techniques for {{Weighted Model Integration}}},
  author = {Morettin, Paolo and Passerini, Andrea and Sebastiani, Roberto},
  year = 2019,
  journal = {Artif Intell},
  volume = {275},
  number = {C},
  pages = {1--27},
  issn = {00043702},
  doi = {10.1016/j.artint.2019.04.003},
  langid = {english}
}

@phdthesis{phanModelCountingModulo2015,
  type = {Thesis},
  title = {Model {{Counting Modulo Theories}}},
  author = {Phan, Quoc-Sang},
  year = 2015,
  url = {https://qmro.qmul.ac.uk/xmlui/handle/123456789/15130},
  urldate = {2024-07-16},
  langid = {english},
  school = {Queen Mary University of London}
}

@inproceedings{sebastianiColorsMakeTheories2016,
  title = {Colors {{Make Theories Hard}}},
  author = {Sebastiani, Roberto},
  year = 2016,
  series = {{{LNCS}}},
  volume = {9706},
  pages = {152--170},
  doi = {10.1007/978-3-319-40229-1_11},
  urldate = {2026-03-09},
  isbn = {978-3-319-40228-4 978-3-319-40229-1},
  publisher = {Springer},
  booktitle = {{{IJCAR}} 2016}
}

@inproceedings{shawApproximateSMTCounting2025,
  title = {Approximate {{SMT Counting Beyond Discrete Domains}}},
  author = {Shaw, Arijit and Meel, Kuldeep S.},
  year = 2025,
  pages = {1--7},
  doi = {10.1109/DAC63849.2025.11133351},
  urldate = {2026-03-02},
  booktitle = {{{DAC}}}
}

@inproceedings{shawEfficientVolumeComputation2025,
  title = {Efficient {{Volume Computation}} for {{SMT Formulas}}},
  author = {Shaw, Arijit and Sarkar, Uddalok and Meel, Kuldeep S.},
  year = 2025,
  volume = {22},
  pages = {544--554},
  issn = {2334-1033},
  doi = {10.24963/kr.2025/53},
  urldate = {2026-03-05},
  langid = {english},
  booktitle = {{{KR}} 2025}
}

@misc{somenzi2009cudd,
  title = {{{CUDD}}: {{CU}} Decision Diagram Package},
  author = {Somenzi, Fabio},
  year = 2009,
  url = {https://github.com/cuddorg/cudd},
  howpublished = {University of Colorado Boulder}
}

@inproceedings{strichmanDecidingSeparationFormulas2002,
  title = {Deciding {{Separation Formulas}} with {{SAT}}},
  author = {Strichman, Ofer and Seshia, Sanjit A. and Bryant, Randal E.},
  year = 2002,
  pages = {209--222},
  publisher = {Springer},
  doi = {10.1007/3-540-45657-0_16},
  isbn = {978-3-540-45657-5},
  langid = {english},
  booktitle = {{{CAV}} 2002}
}

@inproceedings{strichmanSolvingPresburgerLinear2002,
  title = {On {{Solving Presburger}} and {{Linear Arithmetic}} with {{SAT}}},
  author = {Strichman, Ofer},
  year = 2002,
  pages = {160--170},
  publisher = {Springer},
  doi = {10.1007/3-540-36126-X_10},
  isbn = {978-3-540-36126-8},
  langid = {english},
  booktitle = {{{FMCAD}} 2002}
}

@article{sundermannReusingDDNNFsEfficient2024,
  title = {Reusing D-{{DNNFs}} for {{Efficient Feature-Model Counting}}},
  author = {Sundermann, Chico and Raab, Heiko and He{\ss}, Tobias and Th{\"u}m, Thomas and Schaefer, Ina},
  year = 2024,
  journal = {ACM Trans. Softw. Eng. Methodol.},
  volume = {33},
  number = {8},
  pages = {1--32},
  issn = {1049-331X, 1557-7392},
  doi = {10.1145/3680465},
  urldate = {2026-03-03},
  langid = {english}
}

@inproceedings{vandenbroeckLiftedProbabilisticInference2011,
  title = {Lifted {{Probabilistic Inference}} by {{First-Order Knowledge Compilation}}},
  author = {{Van den Broeck}, Guy and Taghipour, Nima and Meert, Wannes and Davis, Jesse and De Raedt, Luc},
  year = 2011,
  series = {{{IJCAI}}'11},
  pages = {2178--2185},
  publisher = {AAAI Press},
  doi = {10.5591/978-1-57735-516-8/IJCAI11-363},
  urldate = {2026-04-14},
  isbn = {978-1-57735-515-1},
  booktitle = {{{IJCAI}} 2011}
}

@inproceedings{vandepolBDDRepresentationLogicEquality2005,
  title = {A {{BDD-Representation}} for the {{Logic}} of {{Equality}} and {{Uninterpreted Functions}}},
  author = {{van de Pol}, Jaco and Tveretina, Olga},
  year = 2005,
  pages = {769--780},
  publisher = {Springer},
  doi = {10.1007/11549345_66},
  isbn = {978-3-540-31867-5},
  langid = {english},
  series = {{{LNCS}}},
  booktitle = {{{MFCS}} 2005}
}

@inproceedings{velevEffectiveUseBoolean2001,
  title = {Effective Use of {{Boolean}} Satisfiability Procedures in the Formal Verification of Superscalar and {{VLIW}}},
  author = {Velev, Miroslav N. and Bryant, Randal E.},
  year = 2001,
  pages = {226--231},
  doi = {10.1145/378239.378469},
  urldate = {2026-02-13},
  isbn = {978-1-58113-297-7},
  publisher = {ACM},
  booktitle = {{{DAC}} 2001}
}

  \newpage
  \appendix
  \section{Proofs}\label{sec:proofs}

{\em In what follows ``Proposition'' denotes facts which are straightforward
consequences of the  definitions, whereas ``Lemma'' and ``Theorem''
denote facts for which we provide a proof explicitly.}

\subsection{Some auxiliary propositions and lemmas}

\begin{proposition}%
\label{prop:assignmentsets}
 Consider some \T-formulas \via{}, $\via{1}$, and $\via{2}$. We have that: 
  \begin{enumerate}[(a)]
{\item\label{item:assignmentsets:disjoint}
    $\CTTA{\vi{}}$, $\CTTA{\neg\vi{}}$, $\ITTA{\vi{}}$, 
    $\ITTA{\neg\vi{}}$ are all pairwise disjoint;}
{\item\label{item:assignmentsets:exhaustive}
 $\CTTA{\vi{}} \cup \CTTA{\neg\vi{}} \cup \ITTA{\vi{}} \cup
 \ITTA{\neg\vi{}}={\CTTA{\top}\cup\ITTA{\top}=}\
 2^{\allalpha}$;}
  {\item\label{item:assignmentsets:tsat}
   $\vi{}$ is \T-unsatisfiable if and only if  $\CTTA{\vi{}}=\emptyset$;}
 {\item\label{item:assignmentsets:bsat}
   $\vi{}$ is \B-unsatisfiable if and only if  $\CTTA{\vi{}}=\emptyset$
   and $\ITTA{\vi{}}=\emptyset$;}
 {\item\label{item:assignmentsets:tvalid}
   $\vi{}$ is \T-valid  if and only if $\CTTA{\neg\vi{}}=\emptyset$;}
 {\item\label{item:assignmentsets:bvalid}
   $\vi{}$ is \B-valid  if and only if $\CTTA{\neg\vi{}}=\emptyset$ and $\ITTA{\neg\vi{}}=\emptyset$;}
  \item\label{item:assignmentsets:tequiv}
 $\vi{1}\Tequiv\vi{2}$ if and only if
 $\CTTA{\vi{1}}=\CTTA{\vi{2}}$;
  \item\label{item:assignmentsets:bequiv}
 $\vi{1}\bequiv\vi{2}$ if and only if
    $\CTTA{\vi{1}}=\CTTA{\vi{2}}$ and $\ITTA{\vi{1}}=\ITTA{\vi{2}}$;
  {\item\label{item:assignmentsets:tentail}
 $\vi{1}\Tmodels\vi{2}$ if and only if
 $\CTTA{\vi{1}}\subseteq\CTTA{\vi{2}}$;}
  {\item\label{item:assignmentsets:bentail}
 $\vi{1}\pmodels\vi{2}$ if and only if
    $\CTTA{\vi{1}}\subseteq\CTTA{\vi{2}}$ and $\ITTA{\vi{1}}\subseteq\ITTA{\vi{2}}$.}
 
  \end{enumerate}
\end{proposition}     

\begin{proposition}%
\label{prop:assignmentdiffsets}
    Consider some \T-formula \vi{} and two supersets of \atoms{\vi{}},
    $\allalpha,\allalpha{}'$, s.t $\allalpha\neq\allalpha{}'$. We have
    that:
\begin{enumerate}[(a)]
\item\label{item:assignmentdiffsets:diff}      
  $\CTTA{\vi{}}\neq \CTTAprime{\vi{}}$ unless
  $\CTTA{\vi{}}=\CTTAprime{\vi{}}=\emptyset$, and\\
  $\ITTA{\vi{}}\neq \ITTAprime{\vi{}}$ unless $\ITTA{\vi{}}=\ITTAprime{\vi{}}=\emptyset$;
\item\label{item:assignmentdiffsets:bequiv}      
  $%
\OrOf{\CTTA{\vi{}}}{\eta{}} \not\bequiv 
    \OrOf{\CTTAprime{\vi{}}}{\etaprime{}} 
 \mbox{ unless $\CTTA{\vi{}}=\CTTAprime{\vi{}}=\emptyset$, and}\\
\OrOf{\ITTA{\vi{}}}{\rho{}} \not\bequiv 
    \OrOf{\ITTAprime{\vi{}}}{\rhoprime{}}  
  \mbox{ unless $\ITTA{\vi{}}=\ITTAprime{\vi{}}=\emptyset$;}
      $
\item\label{item:assignmentdiffsets:tequiv}      
  $%
 \OrOf{\CTTA{\vi{}}}{\eta{}} \Tequiv \!
    \OrOf{\CTTAprime{\vi{}}}{\etaprime{}}
    ,\\
\OrOf{\ITTA{\vi{}}}{\rho{}} \Tequiv 
    \OrOf{\ITTAprime{\vi{}}}{\rhoprime{}}  \Tequiv \bot.
      $
\item\label{item:assignmentdiffsets:true}
  $\CTTA{\top}\neq H_{\allalpha{}'}(\top)$, 
  $\ITTA{\top}\neq P_{\allalpha{}'}(\top)$.

\end{enumerate}
\end{proposition}

  \begin{lemma}[\Treduction of $\resvil{}\wedge l$]
  \label{teo:tresidualconjoinistreduced}
    Let $\via{}$ be a  \Treduceda \T-formula, and let $l$ be some literal
    on \allalpha{}.
    Then       $\resvil{}\wedge l$ is \Treduceda. 
\end{lemma}
\begin{proof}
    $\resvil{}\wedge l \bequiv \vi{}\wedge l$.
 Then $\ITTA{\resvil{}\wedge
   l}=\ITTA{\vi{}\wedge l}\subseteq \ITTA{\vi{}}=\emptyset$,
 Thus $\resvil{}\wedge l$ is \Treduceda. 
\end{proof}

\begin{lemma}[\Textension of $\neg l \vee \resvil{}$ and
  $\neg l \vee (\resvil{}\wedge l)$]
  \label{teo:tresidualdisjoinistextended}
    Let $\via{}$ be a  \Textendeda \T-formula, and let $l$ be some literal
    on \allalpha{}.
    Then $\neg l \vee \resvil{}$ and      $\neg l \vee
    (\resvil{}\wedge l)$
    are \Textendeda. 
\end{lemma}
\begin{proof}
  Let $\vi{}$ be \Textendeda. Then $\neg\vi{}$ is \Treduceda.
  By \cref{teo:tresidualconjoinistreduced}, 
  $\residual{(\neg\vi{})}{l}\wedge l$ is \Treduceda.
  Thus $\neg(\neg(\resvil{})\wedge l) $ is \Textendeda.
 Since $ \neg(\neg(\resvil{})\wedge l)
 \bequiv\resvil{}\vee \neg l \bequiv
  (\resvil{}\wedge l)\vee \neg l$,
  then both $\neg l \vee \resvil{}$ and
  $\neg l \vee (\resvil{}\wedge l)$ are \Textendeda.
\end{proof}

  \begin{lemma}[\T-literal implication check with a \Textendeda{} \T-formula{}]
    \label{teo:littentails}
  Let $\via{}$ be a  \Textendeda \T-formula.
    Let $l$ be a literal in \allalpha.
   Then $l\Tmodels \vi{}$ if and only if $l\pmodels\vi{}$.
  \end{lemma}
  \begin{proof} \ \\
    $l \Tmodels \vi{} $  if and only if
    $\neg l\vee\vi{} $ is \T-valid, that is, if and only if
    $\neg l\vee\resvil{}$ is \T-valid.
    \\
    $l \pmodels \vi{} $  if and only if
    $\neg l\vee\vi{} $ is \B-valid, that is, if and only if
    $\neg l\vee\resvil{}$ is \B-valid.
    \\
    By     applying 
    \cref{teo:tresidualdisjoinistextended},
    we have that
    $\neg l\vee\resvil{}$ is also 
    \Textendeda.
    Thus, by~\cref{teo:tvaliditytextended},
    $\neg l\vee\resvil{}$ is \T-valid iff it
    is \B-valid.\\
    Combining the three facts above, $l\Tmodels \vi{}$ if and only if $l\pmodels\vi{}$. 
  \end{proof}

\newpage\subsection{Proof of \cref{teo:treducedallqueries}}

The items CO, CE, EQ, SE, CT and ME of \cref{teo:treducedallqueries} are proven by
\cref{teo:treducedtsatisfiability,,teo:tentailment,,teo:treducedtequivalence,,teo:treducedtentailment,,teo:treducedcounting,,teo:treducedenumeration} respectively.

\begin{lemma}[CO: \T-satisfiability of \Treduceda formulas]
  \label{teo:treducedtsatisfiability}
  Let $\via{}$ be a  \Treduceda \T-formula.
  Then  $\vi{}$ is \T-satisfiable if and only if it is \B-satisfiable.
\end{lemma}
\begin{proof}
By \cref{prop:assignmentsets}~\cref{item:assignmentsets:tsat},
$\vi{}$ is \T-satisfiable if and only if $\CTTA{\vi{}}\neq\emptyset$.\\
By \cref{prop:assignmentsets}~\cref{item:assignmentsets:bsat},
$\vi{}$ is \B-satisfiable if and only if $\CTTA{\vi{}}\cup\ITTA{\vi{}}\neq\emptyset$.\\
$\ITTA{\vi{}}=\emptyset$ because \vi{} is \Treduceda, thus $\vi{}$ is \T-satisfiable if and only if it is \B-satisfiable.
\end{proof}

  \begin{lemma}[CE: \T-entailment with a \Treduceda{} \T-formula{}]
    \label{teo:tentailment}
  Let $\via{}$ be a  \Treduceda \T-formula.
    Let $C$ be some clause on \allalpha.
    Then $\vi{}\Tmodels C$ if and only if $\vi{} \pmodels C$.
  \end{lemma}
  \begin{proof}
    Let $\gamma$ be the cube s.t.\ $\gamma\pequiv\neg C$. Then:\\
    $\vi{} \Tmodels C$  if and only if
    $\vi{} \wedge \gamma$ is \T-unsatisfiable, that is, if and only if
    $\residual{\vi{}}{\gamma} \wedge \gamma$ is \T-unsatisfiable.
    \\
    $\vi{} \pmodels C$  if and only if
    $\vi{} \wedge \gamma$ is \B-unsatisfiable, that is, if and only if
    $\residual{\vi{}}{\gamma} \wedge \gamma$ is \B-unsatisfiable.
    \\
    By     applying iteratively
    \cref{teo:tresidualconjoinistreduced}
    to all literals in $\gamma$, we have that
    $\residual{\vi{}}{\gamma} \wedge \gamma$ is also 
    \Treduceda.
    Thus, by~\cref{teo:treducedtsatisfiability},
    $\residual{\vi{}}{\gamma} \wedge \gamma$ is \T-satisfiable iff it
    is \B-satisfiable.\\
    Combining the three facts above, $\vi{}\Tmodels C$ if and
    only if $\vi{} \pmodels C$. 
     \end{proof}

\begin{lemma}[EQ: \T-equivalence of \Treduceda formulas]
  \label{teo:treducedtequivalence}
  Let $\via{1}$ and $\via{2}$ be  \Treduceda \T-formulas.
  Then  $\vi{1}\Tequiv\vi{2}$ if and only if $\vi{1}\bequiv\vi{2}$.
\end{lemma}
\begin{proof}
  By \cref{prop:assignmentsets}~\cref{item:assignmentsets:tequiv},
  $\vi{1}\Tequiv\vi{2}$ if and only if $\CTTA{\vi{1}}=\CTTA{\vi{2}}$.\\
  By \cref{prop:assignmentsets}~\cref{item:assignmentsets:bequiv},
  $\vi{1}\bequiv\vi{2}$ if and only if
  $\CTTA{\vi{1}}=\CTTA{\vi{2}}$ and
  $\ITTA{\vi{1}}=\ITTA{\vi{2}}$. \\
  $\ITTA{\vi{1}}=\ITTA{\vi{2}}=\emptyset$ because $\vi{1}$ and
  $\vi{2}$ are \Treduceda.\\
  Thus $\vi{1}\Tequiv\vi{2}$ if and only if $\vi{1}\bequiv\vi{2}$.
\end{proof}

\begin{lemma}[SE: \T-entailment of \Treduceda formulas]
  \label{teo:treducedtentailment}
  Let $\via{1}$ and $\via{2}$ be  \Treduceda \T-formulas.
  Then  $\vi{1}\Tmodels\vi{2}$ if and only if $\vi{1}\pmodels\vi{2}$.
\end{lemma}
\begin{proof}
  By \cref{prop:assignmentsets}~\cref{item:assignmentsets:tentail},
  $\vi{1}\Tmodels\vi{2}$ if and only if $\CTTA{\vi{1}}\subseteq\CTTA{\vi{2}}$.\\
  By \cref{prop:assignmentsets}~\cref{item:assignmentsets:bentail},
  $\vi{1}\pmodels\vi{2}$ if and only if
  $\CTTA{\vi{1}}\subseteq\CTTA{\vi{2}}$ and
  $\ITTA{\vi{1}}\subseteq\ITTA{\vi{2}}$. \\
  $\ITTA{\vi{1}}=\ITTA{\vi{2}}=\emptyset$ because $\vi{1}$ and
  $\vi{2}$ are \Treduceda.\\
  Thus $\vi{1}\Tmodels\vi{2}$ if and only if $\vi{1}\pmodels\vi{2}$.
\end{proof}

\begin{lemma}[CT: \SharpSMT{} of \Treduceda \T-formula]
  Let \via{} be a  \Treduceda \T-formula.
  Then $\SharpSMTOf{\vi{}}=\SharpSATOf{\vip{}}$.
\label{teo:treducedcounting}
  \end{lemma}
  \begin{proof}
By \cref{def:cttaitta}, for a generic \T-formula \vi{},
$\SharpSATOf{\viap{}}=|\CTTA{\vi{}}|+|\ITTA{\vi{}}|$.
If \vi{} is  \Treduceda, then $|\ITTA{\vi{}}|=0$.
Thus $\SharpSMTOf{\vi{}}=|\CTTA{\vi{}}|=\SharpSATOf{\vip{}}$.
\end{proof}

\begin{lemma}[ME: \T-consistent truth assignment enumeration for
  \Treduceda \T-formulas]
\label{teo:treducedenumeration}
Let \via{} be a  \Treduceda \T-formula.
  Then $\AllSMTOf{\vi{}}$ is the refinement of $\AllSATOf{\vip{}}$.
\end{lemma}
\begin{proof}
For a generic formula $\vi{}$, the refinement of
$\AllSATOf{\vip{}}$ is $\CTTA{\vi{}}\cup\ITTA{\vi{}}$. Since $\vi{}$ is
\Treduceda so that $\ITTA{\vi{}}=\emptyset$ by
\cref{def:treduced}, it reduces to $\CTTA{\vi{}}=\AllSMTOf{\vi{}}$.
\end{proof}

\newpage\subsection{Proof of \cref{teo:textendedallqueries}}

The items VA, IM, EQ,  and ME of \cref{teo:textendedallqueries} are proven by
\cref{teo:tvaliditytextended,,teo:revtentailment,,teo:textendedtequivalence,,teo:textendedtentailment}
 respectively.

\begin{lemma}[VA: \T-validity of \Textendeda \T-formulas]
\label{teo:tvaliditytextended}
A \Textendeda \T-formula $\via{}$ is \T-valid if and only if it is
\B-valid.
\end{lemma}
\begin{proof}
(If case) A \B-valid \T-formula $\vi{}$ is obviously also
\T-valid by definition.\\
(Only if case)
$\CTTA{\neg\vi{}}=\emptyset$ because $\vi{}$ is \T-valid
(\cref{prop:assignmentsets}, \cref{item:assignmentsets:tvalid}).
$\ITTA{\neg\vi{}}=\emptyset$ because $\vi{}$ is \Textendeda
(\cref{def:textended}). Thus $\neg\vi{}$ is \B-unsatisfiable, so that $\vi{}$ is \B-valid.
\end{proof}

  \begin{lemma}[IM: \T-implicant check with a \Textendeda{} \T-formula{}]
    \label{teo:revtentailment}
  Let $\via{}$ be a  \Textendeda \T-formula.
    Let $\gamma$ be some cube on \allalpha.
   Then $\gamma\Tmodels \vi{}$ if and only if $\gamma\pmodels\vi{}$.
 \end{lemma}
\begin{proof}
Let $\gamma\defas\bigwedge_{i=1}^n l_i$ for some $n\ge 1$. We show
that $\bigwedge_{i=1}^n l_i\Tmodels \vi{}$ if and only if $\bigwedge_{i=1}^n l_i\bmodels \vi{}$ by
induction on $n$:\\
(base):  By \cref{teo:littentails}, if $n=1$, then  $l_1\Tmodels \vi{}$ if and only if
$l_1\pmodels\vi{}$.\\
(step): We recall that
$\neg l \vee \vi{}
\bequiv
\neg l \vee (l \wedge \vi{})
\bequiv
\neg l \vee (l \wedge \residual{\vi{}}{l})
\bequiv
\neg l \vee \residual{\vi{}}{l}
$ (Lemma A.7 in \cite{darwicheKnowledgeCompilationMap2002}).
\\
$\bigwedge_{i=1}^n l_i\Tmodels \vi{}$ if and only if
$\bigwedge_{i=1}^{n-1} l_i\Tmodels (\neg l_n \vee \vi{})$,
if and only if $\bigwedge_{i=1}^{n-1} l_i\Tmodels (\neg l_n \vee
\residual{\vi{}}{l_n})$.
\\
$\bigwedge_{i=1}^n l_i\pmodels \vi{}$ if and only if
$\bigwedge_{i=1}^{n-1} l_i\pmodels (\neg l_n \vee \vi{})$,
if and only if $\bigwedge_{i=1}^{n-1} l_i\pmodels (\neg l_n \vee
\residual{\vi{}}{l_n})$.
\\
By \cref{teo:tresidualdisjoinistextended}, $(\neg l_n \vee
\residual{\vi{}}{l_n}[\allalpha])$ is \Textendeda.
\\
Thus, by induction,
$\bigwedge_{i=1}^{n-1} l_i\Tmodels (\neg l_n \vee
\residual{\vi{}}{l_n})$
if and only if 
$\bigwedge_{i=1}^{n-1} l_i\pmodels (\neg l_n \vee
\residual{\vi{}}{l_n})$, that is,
$\bigwedge_{i=1}^n l_i\Tmodels \vi{}$ if and only if $\bigwedge_{i=1}^n l_i\pmodels\vi{}$.
\end{proof}

\begin{lemma}[EQ: \T-equivalence of \Textendeda formulas]
  \label{teo:textendedtequivalence}
  Let $\via{1}$ and $\via{2}$ be  \Textendeda \T-formulas.
  Then  $\vi{1}\Tequiv\vi{2}$ if and only if $\vi{1}\bequiv\vi{2}$.
\end{lemma}
\begin{proof}
By \cref{def:textended}, $\neg\vi{1}$ and $\neg\vi{2}$ are
\Treduceda.
\\Thus, by \cref{teo:treducedtequivalence},
$\neg\vi{1}\Tequiv\neg\vi{2}$ if and only if
$\neg\vi{1}\bequiv\neg\vi{2}$.\\
Therefore, $\vi{1}\Tequiv\vi{2}$ if and only if $\vi{1}\bequiv\vi{2}$.
\end{proof}

\begin{lemma}[SE: \T-entailment of \Textendeda formulas]
  \label{teo:textendedtentailment}
  Let $\via{1}$ and $\via{2}$ be  \Textendeda \T-formulas.
  Then  $\vi{1}\Tmodels\vi{2}$ if and only if $\vi{1}\bmodels\vi{2}$.
\end{lemma}
\begin{proof}
By \cref{def:textended}, $\neg\vi{1}$ and $\neg\vi{2}$ are
\Treduceda.
\\Thus, by \cref{teo:treducedtentailment},
$\neg\vi{2}\Tmodels\neg\vi{1}$ if and only if
$\neg\vi{2}\bmodels\neg\vi{1}$.\\
Therefore, $\vi{1}\Tmodels\vi{2}$ if and only if $\vi{1}\bmodels\vi{2}$.
\end{proof}

Comparing \cref{teo:textendedtequivalence,,teo:textendedtentailment}
with \cref{teo:treducedtequivalence,,teo:treducedtentailment}
respectively, we notice that the same property holds for both
\Treduceda and \Textendeda pairs of \T-formulas.
This should not be a surprise, remembering that
$\vi{1}\equiv\vi{2}$ if and only if $\neg\vi{1}\equiv\neg\vi{2}$ and that
$\vi{1}\models\vi{2}$ if and only if $\neg\vi{2}\models\neg\vi{1}$.

\newpage\subsection{Proof of \cref{teo:treduce}}
\begin{proof}
  By construction, all $\C{l}$'s in \TLEMMAS{\vi{}} are \T-valid, so
  that $\bigwedge_{\C{l}\in\TLEMMAS{\vi{}}} \C{l}$ is \T-valid. Thus:
\begin{enumerate}[(i)]
\item By \eqref{eq:treduce}, $\TredOf{\vi{}}\Tequiv\vi{}$ because all  \C{l}'s are \T-valid. 
\item By \eqref{eq:treduce}, $\TredOf{\vi{}}\bmodels\vi{}$.
\item By \cref{item:treduce:tequiv}, $\CTTA{\TredOf{\vi{}}}=\CTTA{\vi{}}$; by \cref{item:treduce:bmodels}, 
  $\ITTA{\TredOf{\vi{}}}\subseteq\ITTA{\vi{}}$.
  Consider some $\rho_{j}\in\ITTA{\vi{}}$. 
By \cref{def:ruleout-nobetas}, $\rho_{j}\bmodels\neg\C{l}$ for some
$\C{l}\in\TLEMMAS{\vi{}}$. Since $\rho_{j}$ is a total assignment on
$\allalpha$, then
$\rho_{j}\not\bmodels\C{l}$,
that is, $\rho_{j}\not\bmodels\TredOf{\vi{}}$,
that is, $\rho_{j}\not\in\ITTA{\TredOf{\vi{}}}$.
Thus $\ITTA{\TredOf{\vi{}}}=\emptyset$, that is,  $\TredOf{\vi{}}$ is
\Treduceda.
\end{enumerate}  
\end{proof}

\subsection{Proof of \cref{teo:textend}}

\begin{proof}
  Comparing \eqref{eq:treduce} and \eqref{eq:textend}, we see that
  $\TextOf{\vi{}}\bequiv\neg\TredOf{\neg\vi{}}$. Thus, 
  \cref{item:textend:tequiv,,item:textend:bmodels,,item:textend:textended}
  follow directly from \cref{teo:treduce}.
\end{proof}

\newpage
\subsection{Proof of \cref{teo:polinomialwithtreduced}}
\begin{proof}\ 
\begin{enumerate}[(a)]
\item{} [CO] By \cref{teo:treducedtsatisfiability}, $\psi{}$ is
  \T-satisfiable iff it is \B-satisfiable, which can be computed in
  polynomial time wrt.\ the size of $\psi{}$  for a \dDNNF{} formula \cite{darwicheKnowledgeCompilationMap2002}.
\item{} [CE] By \cref{teo:tentailment},
  $\psi{}\Tmodels C{}$ iff $\psi{}\pmodels C{}$, which can be computed in
  polynomial time wrt.\ the size of $\psi{}$  and $C$  for a
  \dDNNF{} formula \cite{darwicheKnowledgeCompilationMap2002}.
\item{} [CT] By \cref{teo:treducedcounting},
  $\SharpSMT{}(\psia{})=\SharpSATOf{\psiap{}}$, 
  which can be computed in
  polynomial time wrt.\ the size of $\psi{}$ for a \dDNNF{} formula \cite{darwicheKnowledgeCompilationMap2002}.
\item{} [ME] By \cref{teo:treducedenumeration},
  $\AllSMT{}(\psia{})=\AllSATOf{\psiap{}}$, 
  which  can be computed in
  polynomial time wrt.\ the size of $\psi{}$ and the number of output
  models for a \dDNNF{} formula \cite{darwicheKnowledgeCompilationMap2002}.
\end{enumerate}
\end{proof}

\subsection{Proof of \cref{teo:polinomialequivalencewithtreduced}}
\begin{proof}\ 
\begin{enumerate}[(a)]
\item{} [EQ] By \cref{teo:treducedtequivalence}, $\psi_{1}\Tequiv\psi_{2}$ if and only
if $\psi_{1}\bequiv\psi_{2}$, 
  which  can be computed in
  polynomial time wrt.\ the size of $\psi_{1}$ and $\psi_{2}$ for the
  listed \NNF{} formulas \cite{darwicheKnowledgeCompilationMap2002,darwicheSDDNewCanonical2011}. 
\item{} [SE] By \cref{teo:treducedtentailment}, $\psi_{1}\Tmodels\psi_{2}$ if and only
if $\psi_{1}\bmodels\psi_{2}$, 
  which  can be computed in
  polynomial time wrt.\ the size of $\psi_{1}$ and $\psi_{2}$ for the
  listed \NNF{} formulas \cite{darwicheKnowledgeCompilationMap2002,darwicheSDDNewCanonical2011}. 
\end{enumerate}
\end{proof}

\subsection{Proof of \cref{teo:polinomialwithtextended}}
\begin{proof}\ 
  \begin{enumerate}[(a)]
\item{}[VA] By \cref{teo:tvaliditytextended}, $\psi{}$ is
  \T-valid iff it is \B-valid, which can be computed in
  polynomial time wrt.\ the size of $\psi{}$ for a \dDNNF{} formula \cite{darwicheKnowledgeCompilationMap2002}.
\item{}[IM]
By \cref{teo:revtentailment}, $\gamma\Tmodels\psi{}$ if and only if
$\gamma\pmodels\psi{}$, which can be computed in
  polynomial time wrt.\ the size of $\psi{}$ and $\gamma$ for a \dDNNF{} formula \cite{darwicheKnowledgeCompilationMap2002}. 
\end{enumerate}
\end{proof}

\subsection{Proof of \cref{teo:polinomialequivalencewithtextended}}
\begin{proof}\ 
\begin{enumerate}[(a)]
\item{} [EQ] By \cref{teo:textendedtequivalence}, $\psi_{1}\Tequiv\psi_{2}$ if and only
if $\psi_{1}\bequiv\psi_{2}$, 
  which can be computed in
  polynomial time wrt.\ the size of $\psi_{1}$ and $\psi_{2}$ for the listed \NNF{} formulas \cite{darwicheKnowledgeCompilationMap2002,darwicheSDDNewCanonical2011}. 
\item{} [SE] By \cref{teo:textendedtentailment}, $\psi_{1}\Tmodels\psi_{2}$ if and only
if $\psi_{1}\bmodels\psi_{2}$, 
  which can be computed in
  polynomial time wrt.\ the size of $\psi_{1}$ and $\psi_{2}$ for the listed \NNF{} formulas \cite{darwicheKnowledgeCompilationMap2002,darwicheSDDNewCanonical2011}. 
\end{enumerate}
\end{proof}

\end{document}